\documentclass{article}

\title{Cluster extent inference revisited: quantification and localization of brain activity} 

\author {
Jelle J. Goeman\footnote{Biomedical Data Sciences, Leiden University Medical Center, Leiden, The Netherlands} \and
Pawe\l\ G\'orecki\footnote{Institute of Informatics, Faculty of Mathematics, Informatics and Mechanics,  
University of Warsaw, Poland} \and
Ramin Monajemi$^*$ \and
Xu Chen$^*$ \and
Thomas E. Nichols\footnote{Big Data Institute, Li Ka Shing Centre for Health Information and Discovery, Nuffield Department of Population Health, University of Oxford, UK}~\footnote{Wellcome Centre for Integrative Neuroimaging, FMRIB, Nuffield Department of Clinical Neurosciences, University of Oxford, UK} \and 
Wouter Weeda\footnote{Methodology and Statistics, Psychology, Leiden University, The Netherlands}
}

\date{August 9, 2022} 

\usepackage{natbib}
\usepackage{amssymb, dsfont, amsthm, times, amsmath, tikz}
\usepackage{appendix, comment}

\newtheorem{lemma}{Lemma}
\newtheorem{theorem}{Theorem}
\usepackage{hyperref}
\usepackage{booktabs}
\usepackage[normalem]{ulem}

\usepackage[noend,ruled,vlined]{algorithm2e}

\usepackage[margin=3cm]{geometry}

\usepackage{todonotes}

\usepackage{longtable, pdflscape}
\usepackage{caption}
\usepackage{subcaption}

\makeatletter
\newtheorem*{rep@theorem}{\rep@title}
\newcommand{\newreptheorem}[2]{%
\newenvironment{rep#1}[1]{%
 \def\rep@title{#2 \ref{##1}}%
 \begin{rep@theorem}}%
 {\end{rep@theorem}}}
\makeatother

\newreptheorem{lemma}{Lemma}
\newreptheorem{theorem}{Theorem}

\begin{document}
\maketitle

\begin{abstract}
    Cluster inference based on spatial extent thresholding is the most popular analysis method for finding activated brain areas in neuroimaging. However, the method has several well-known issues. While powerful for finding brain regions with some activation, the method as currently defined does not allow any further quantification or localization of signal. In this paper we repair this gap.  We show that cluster-extent inference can be used (1.)\ to infer the presence of signal in anatomical regions of interest and (2.)\ to quantify the percentage of active voxels in any cluster or region of interest. These additional inferences come for free, i.e.\ they do not require any further adjustment of the alpha-level of tests, while retaining full familywise error control. We achieve this extension of the possibilities of cluster inference by an embedding of the method into a closed testing procedure, and solving the graph-theoretic $k$-separator problem that results from this embedding. The new method can be used in combination with random field theory or permutations. We demonstrate the usefulness of the method in a large-scale application to neuroimaging data from the Neurovault database.
\end{abstract}

\section{Introduction}

Functional Magnetic Resonance Imaging (fMRI) studies aim to find brain regions that are activated in response to a mental task. The activity of the brain is measured by the proxy of changes in blood oxygenation levels (BOLD), and researchers look for areas in which these changes are associated with the pattern of the experimental stimulus, e.g.\ the alternation of task and rest \citep{Ogawa1992}.

From a statistical perspective an fMRI experiment is a huge multiple testing problem. The brain is partitioned into around 200,000 voxels, 3-dimensional equivalents of pixels. For each such voxel a $z$-score test statistic is calculated that combines the evidence from the BOLD measurements of the experimental subjects. Inference based on these test statistics can be done at the voxel level, resulting in a multiple testing problem with around 200,000 null hypotheses. More commonly, however, fMRI researchers are interested in inference at the level of clusters, sets of connected voxels, with the aim of relating these clusters of activation to certain anatomical areas in the brain (e.g.\ "listening to sounds is related to increased activation in the left auditory cortex"). 

The standard method for cluster inference is cluster extent thresholding \citep{Friston1994,Forman1995,Nichols2012}. The researcher chooses a $z$-score cut-off $z$, finds all voxels with a $z$-score above $z$, and identifies supra-threshold connected voxels as clusters. Next, all clusters with an extent (number of voxels) larger than the extent threshold $k$ are declared significant. To control the cluster familywise error rate (FWER), the extent threshold $k$ must be the $(1-\alpha)$-quantile of the distribution of the maximal extent of such clusters under the global null hypothesis. It can be determined either analytically, using the assumption that the $z$-scores come from a Gaussian random field \citep{Worsley1996, Friston1994, Eklund2016}, or more robustly by permutations \citep{Hayasaka2003}.  Alternatively, the cluster false discovery rate can be controlled, by submitting uncorrected cluster p-values to the Benjamini-Hochberg procedure \citep{Chumbley2010}; other proposals have included controlling the expected number of false positive clusters \citep{Bullmore1999}. 

Although the FWER extent threshold $k$ is calculated under the complete null hypothesis, it has been shown that cluster inference has strong control of the FWER  \citep{Worsley1992}. This implies that, regardless of the amount of signal present in the data, with probability at least $1-\alpha$ no cluster null hypothesis is falsely rejected. The cluster null hypothesis is the hypothesis that none of the voxels in the cluster is truly ``active'', i.e.\ associated with the experimental stimulus. The inferential statement that can be made from cluster inference is, therefore, that, with $1-\alpha$ simultaneous confidence, every significant cluster contains at least one active voxel. 

While this cluster-level FWER control is the de facto approach to cluster inference, it has been criticized as insufficient to support the conclusions researchers would typically like to draw from neuroimaging experiments.  For example, \cite{Woo2014} argued that, especially at low $z$ thresholds, clusters can become too large and span multiple brain areas, challenging the interpretation of the results.  
The following three inferential conclusions are often (implicitly or explicitly) drawn from cluster inference result, though they are not supported by the theory.

\begin{enumerate}
\item \emph{``A large significant cluster contains a substantial number of active voxels.''} Cluster-level FWER control only supports the statement that at least one voxel in the cluster is confidently active, not that many, or let alone, all voxels are active. This is perhaps one of the most frequent misunderstandings of the current state-of-the-art in cluster inference \citep{Woo2014}. 

\item \emph{``A large significant cluster is a more substantial scientific finding than a small significant cluster.''} In fact, the assertion that at least one voxel in a large cluster is active, is a less precise, and therefore weaker finding than the same assertion in a small cluster. This counter-intuitive property is known as the Spatial Specificity Paradox \citep{Woo2014}.

\item \emph{``Substantial overlap between a significant cluster and an anatomical brain area indicates evidence for the presence of activity in that anatomical brain area.''} A significant cluster confidently contains at least one active voxel, but unless that cluster is completely contained in the anatomical area, such activity may lie outside the anatomical brain area \citep{Woo2014}.
\end{enumerate}

Despite its widespread use, cluster-level FWER provides very weak inferences on the nature of non-null signal within a cluster. Still, the three desired conclusions from cluster inference, sketched above, are intuitively quite reasonable. If a cluster exceeds the minimal size $k$ for a significant cluster by a large margin, it is natural to suppose that there is a substantial amount of signal in the cluster, and at least more than in another cluster with an extent just over $k$. If the large cluster largely overlaps with an anatomical region, it is reasonable to suppose that some of the signal in the cluster must be in the anatomical region. 

This paper strengthens cluster inference by presenting an improvement of the method that allows much stronger and more informative conclusions to be drawn, avoiding the problems sketched above. Rather than returning a $p$-value for each supra-threshold cluster, the new method returns a \emph{true discovery proportion} (TDP) for every region, a simultaneous lower confidence bound for the proportion of truly active voxels in the region \citep{Genovese2006, Goeman2011}. By quantifying how widely spread a signal is within a brain region, TDP-based inference avoids the spatial specificity paradox \citep{Rosenblatt2018}. Moreover, TDP can be calculated for any brain region, not just for supra-threshold clusters; this way also the amount of signal in anatomical regions may be assessed.

Analysis of neuroimaging data in terms of TDP rather than $p$-values was pioneered by \cite{Rosenblatt2018}, who proposed the ARI method based on closed testing with the Simes test \citep{Goeman2019}. Other methods for TDP inference suitable for brain imaging include \citet{Blanchard2020, Andreella2020, Vesely2021, Blain2022}. The proposed method differs from these methods because it is based on classic extent-based cluster inference, and therefore aligns much more closely with standard practice. Unlike these methods, the new method will always yield TDP $>0$ for any cluster that is significant according to classic cluster-based inference. In fact, there is no power loss when switching from classic cluster-based inference to the method proposed in this paper; the new method is a uniform improvement \citep[in the sense of][]{Goeman2020} of classic cluster-based inference. Moreover, the new method retains strict FWER control over all reported findings: with probability at least $1-\alpha$ no reported TDP is greater than the proportion of truly active voxels in the corresponding region.

We construct the improvement of cluster inference by remarking that cluster inference is a special case of a true discovery guarantee method, as defined by \cite{Goeman2020}. Viewed in this way, cluster inference is not admissible, but can be uniformly improved by embedding it into a closed testing procedure, which we will construct. The local test of this procedure rejects the null hypothesis of no activity in a subset of the brain whenever that subset contains a connected subset of size at least $k$ for which all voxel $z$-scores are above $z$. 

A major challenge of constructing closed testing procedures is, as always, computational. We will show that calculating TDP for a brain region amounts to solving an instance of a graph-theoretic $k$-separator problem \citep{ben2015separator}. We propose two novel and fast algorithms to solve the $k$-separator problem in the lattice graph induced by brain connectivity, in order to find shortcuts for the closed testing procedure. 

To illustrate the performance of the method we will apply the novel lower bound on 818 data sets from the Neurovault database \citep{Gorgolewski2015}. We will first illustrate the intended workflow of the new method using an $n$-back working memory data set \citep{Barch2013}, which we will introduce in the next section as a motivating example.

\section{Motivating example} \label{sec motivate}

We will first illustrate and preview the new method with a concrete motivating example. The Human Connectome Project \citep[HCP;][]{VanEssen2013} consists of neuroimaging data of over 5000 subjects performing multiple cognitive tasks. In our example we will use fMRI data obtained from 80 unrelated individuals, each performing an $n$-back working memory task \citep{Barch2013}. During this task participants are sequentially shown a series of letters (e.g.\ ``E", ``D", ``Z", ``X", ``M"). After the sequence is shown participants are asked to recall letters from a specific position in the sequence. For example, in the 0-back condition this is the last letter shown (``M"), in the 2-back condition this is the letter in second-to-last position (``Z"). In $n$-back tasks, higher values of $n$ are theoretically associated with larger memory load for the participants. We focused on the 2-back versus 0-back contrast, for which the null hypothesis of interest per voxel was that the BOLD signal was identically distributed between the 2-back and 0-back conditions. For the calculation of per-voxel test statistics, we followed a standard processing pipeline \citep{Glasser2013} using FSL \citep{Woolrich2001}, a popular software package for cluster extent inference. This is a two-stage analysis, in which the 2-back versus 0-back contrast is first analyzed for each subject separately and the results are subsequently aggregated across subjects into a group-level $z$-statistic for each of the 257,659 voxels in the brain, using standard methods described by \citet{Beckmann2003}. Each of these $z$-statistics is standard normal under their respective per-voxel null hypothesis.

Before seeing the data, a cluster-forming threshold of $z=3.1$ was chosen. Clusters were formed by all connected neighboring supra-threshold voxels. Using standard theory, which we will revisit in Section \ref{sec classical}, a permutation-based extent threshold of 72 was found, indicating that all clusters consisting of more than 72 voxels are significant. This led to 6 significant clusters and several non-significant clusters. The details of the significant clusters are shown in Figure \ref{fig motivate} and Table \ref{tbl motivate}. 

\begin{figure}[ht]
\begin{center}
\includegraphics[trim = 0mm 0mm 0mm 0mm, clip=true, width=\textwidth]{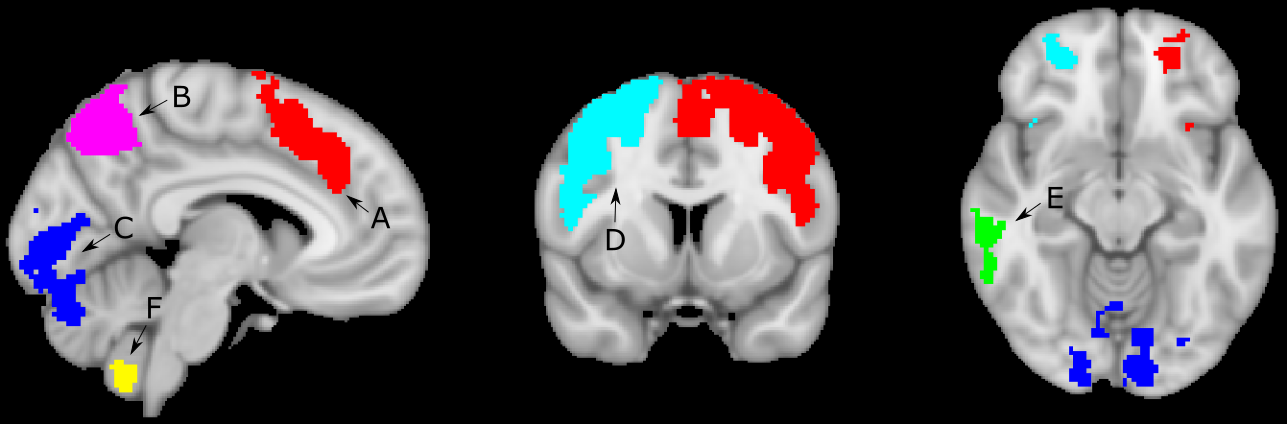}
\caption{Task-related brain activation for the 2-back versus 0-back contrast across all subjects. Six significant clusters A, B, C, D, E, F are displayed using different colors.}
\label{fig motivate}
\end{center}
\end{figure}

\begin{table}[ht]
\caption{Task-related brain activation for the 2-back versus 0-back contrast across all subjects. Columns show the size, p-value, maximum $z$-statistic, and coordinates of the maximum for all clusters.}
\label{tbl motivate}
\centering
\begin{tabular}{c c c c c c c}
\toprule
Cluster & Size & $p$-value & $\max(z)$ & X & Y & Z \\
\midrule
A & $8870$ & $<0.001$ & $8.87$ & $44$ & $72$ & $60$ \\
B & $8526$ & $<0.001$ & $9.51$ & $19$ & $42$ & $61$ \\
C & $7956$ & $<0.001$ & $9.20$ & $63$ & $33$ & $20$ \\
D & $6652$ & $<0.001$ & $9.73$ & $31$ & $67$ & $64$ \\
E & $350$  & $0.004$ & $5.18$ & $15$ & $46$ & $28$ \\
F & $100$  & $0.027$ & $6.56$ & $49$ & $35$ & $10$ \\
\bottomrule
\end{tabular}
\end{table}

With classic cluster inference, the analysis ends here. The researchers may claim that some signal is present in each significant cluster, but the amount of signal is undetermined. This is especially tantalizing for the biggest cluster A, that visually consists of several sub-regions. No statement can be made about the presence of signal in these sub-clusters. Cluster C overlaps for a large part with the cerebellum, but since it is not fully contained in the cerebellum, the researcher may not confidently claim the presence of signal here from the overlap with Cluster C. In contrast, Cluster F, which is relatively small and would not attract the most attention in the publication, does substantiate a claim about the presence of signal in the cerebellum since it is completely contained in it. Paradoxically, Cluster F is the most precise finding, since it localizes the presence of signal to a precision of no more than 100 voxels. 

The theory developed in this paper will allow much more informative statements to be made about clusters A, to F.
\begin{enumerate}
\item We calculate a true discovery proportion (TDP) per cluster, a lower bound to the number of truly active voxels. Clusters A, B, C, D, E, F get TDPs of 37\%, 40\%, 33,\%, 37\%, 19\% and 10\%, respectively. This indicates that clusters A to D are the main findings of the experiment, but shows that the localization of the signal is only moderately precise. 
\item We also find TDPs for any other (anatomical) brain regions of interest. We find, for example, significant evidence of signal in cerebellum, mostly from the overlap with cluster C, though with a small TDP of 5.8\%.
\end{enumerate}

The TDP values we find are guaranteed to be consistent with the cluster $p$-values in the sense that $p\leq 0.05$ if and only if TDP is positive. Compared to the $p$-values, the TDP is more informative since it quantifies the pervasiveness of the signal within the cluster. The full analysis results are given in Section \ref{sec appl}.

\section{Classic cluster inference} \label{sec classical}

We start by briefly revisiting classic cluster inference. We will follow the notational conventions used in \cite{Goeman2020} that, except for the probability distribution $\mathrm{P}$, all capitals are sets and all lower case variables are scalars or vectors. Random variables are in boldface.

\subsection{Voxels and clusters}

The brain is partitioned into hundreds of thousands of voxels, forming a rectangular grid. With suitable coordinates each voxel can be identified as a point in $\mathbb{Z}^d$. We will usually think of $d=3$, but we will write our theory for general $d\geq 1$. The brain $B \subset \mathbb{Z}^d$ is an irregularly shaped, finite collection of voxels. It is not always the entire brain that is of interest to the researcher, and a mask $M \subseteq B$ is chosen, before seeing the data, limiting all inference to voxels in $M$.

We define a neighbor relationship between voxels, saying that voxels $v, w \in \mathbb{Z}^d$ are neighbors if $v-w \in \{-1,0,1\}^d$. This neighborhood definition is known as 26-connectivity in neuroimaging since it gives each voxel 27 neighbors (26 plus itself) if $d=3$. 

The voxels and the neighbor relation together induce an undirected graph when the voxels are seen as nodes and the neighbor relationships as edges. We call a voxel set $V\subseteq \mathbb{Z}^d$ a \emph{cluster} if its induced graph is connected, i.e.\ if we can traverse from every voxel in $V$ to every other voxel in $V$ by passing from neighbor to neighbor.  We call voxel sets $V$ and $W$ \emph{disconnected} if no voxel of $V$ is a neighbor of a voxel of $W$. 

\subsection{Voxel null hypotheses and $z$-scores}

Let $\Omega$ be our statistical model and $\mathrm{P} \in \Omega$ the unknown probability distribution of the data. For each voxel $v \in M$ we define a voxel-wise null hypothesis $H_v \subseteq \Omega$ stating that the voxel $v$ is not active, i.e.\ that the BOLD signal for that voxel is not related to the experimental stimulus. Note that in general a hypothesis $H$ is true if and only if $\mathrm{P} \in H$. 

An fMRI experiment typically involves several subjects that are measured for a prolonged time period, leading to a huge data set with a BOLD observation per subject per voxel per time point. In the first steps of the analysis, for every $v \in B$, these data are aggregated to a single $z$-score $\mathbf{z}_v$ per voxel that represents the evidence against the voxel null hypothesis from the experiment. We refer to \cite{Lindquist2008} for a description of the analysis steps involved. In this paper we assume that the first steps of the analysis have already been done, and we start from $z$-scores $(\mathbf{z}_v)_{v\in B}$. The $z$-score $\mathbf{z}_v$ is expected to be small in absolute value if $H_v$ is true and large if $H_v$ is false.

\subsection{Voxel set null hypotheses and the cluster extent threshold}

Researchers are usually not particularly interested in individual voxels, since these are considered too small to represent relevant brain processes. Instead, researchers look at clusters of neighboring voxels. For every voxel set $V \subseteq M$, we define the voxel set null hypothesis as $H_V = \bigcap_{v \in V} H_v$. This hypothesis states that all of the voxelwise null hypotheses for voxels in $V$ are true, i.e.\ that none of the voxels in $V$ are active. The hypothesis $H_\emptyset = \Omega$ is always true.

Cluster inference uses the voxel $z$-scores to make inference at the cluster level. First, before seeing the data the researcher selects a $z$-score cut-off $z$. Next, the researcher finds the set of all supra-threshold voxels in the mask, $M \cap \mathbf{Z}$, where 
\begin{equation} \label{def Z}
\mathbf{Z} = \{v \in B\colon \mathbf{z}_v > z\}
\end{equation}
is the collection of all supra-threshold voxels. Equation (\ref{def Z}) uses one-sided tests. Two-sided tests can be done either using $|\mathbf{z}_v| > z$ in \eqref{def Z} or by repeating the analysis twice: once with $\mathbf{z}_v$ and once with $-\mathbf{z}_v$, using half the $\alpha$-level. 

The supra-threshold voxel set $\mathbf{Z} \cap M$ is not in general a cluster, but it is always a union of clusters. We can uniquely write $\mathbf{Z} \cap M = \mathbf{C}_1 \cup \cdots \cup \mathbf{C}_\mathbf{n}$, where $\mathbf{C}_1,\ldots, \mathbf{C}_\mathbf{n}$ are disconnected clusters. Cluster inference now claims the presence of signal in every $\mathbf{C}_i$ for which $|\mathbf{C}_i|>k_M$, where $|\cdot|$ is the cardinality of a set, and $k_M$ is the cluster extent threshold calculated for mask $M$. The cluster extent threshold is defined as the $(1-\alpha)$-quantile of the maximum size of a supra-threshold cluster under the global null. Formally, the size of the largest supra-threshold voxel is $\chi_{M \cap \mathbf{Z}}$, where
\[
\chi_V = \max\{|C|\colon \textrm{$C\subseteq V$ is a cluster}\}.
\]
This maximum is always defined since the empty set is a cluster. The cluster extent threshold $k_M$ therefore has the property that, for every $\mathrm{P} \in H_M$, 
\begin{equation} \label{eq CI weak FWER}
\mathrm{P}(\chi_{M \cap \mathbf{Z}} > k_M) \leq \alpha.
\end{equation}
We remark that $k_M$ is allowed to be random, as it would be e.g.\ in permutation approaches. We also remark that we deviate slightly from the usual definition of $k_M$, which uses $\geq$ in the first inequality in (\ref{eq CI weak FWER}).

To achieve (\ref{eq CI weak FWER}) various assumptions have been proposed. \cite{Friston1994} assumes that $(z_v)_{v\in M}$ follows a stationary Gaussian random field on $M$, and that each $H_v$, $v \in B$, is the hypothesis that $z_v$ has zero mean. In this case, $k_M$ can be approximated using the expected Euler characteristic of the field, and (\ref{eq CI weak FWER}) holds as long as $z$ is large enough and the field is sufficiently smooth \citep{Worsley1996, Eklund2016}. Alternatively, a $k_M$ achieving (\ref{eq CI weak FWER}) may be calculated from other assumptions, e.g.\ using permutations \citep{Hayasaka2003}, $t$-fields, ${\chi}^2$-fields, or $F$-fields \citep{Worsley1996}. In the rest of the paper we will not use any specific set of distributional assumptions. We will simply assume $k_M$ can be calculated for every $M \subseteq B$ such that (\ref{eq CI weak FWER}) holds.

Larger masks allow larger supra-threshold clusters, and therefore larger cluster extent thresholds. We will assume that if $M \subseteq N$, then,
\begin{equation} \label{ass h monotone}
k_M \leq k_N. 
\end{equation}
This relationship is natural since $\chi_{M \cap \mathbf{Z}} \leq \chi_{N \cap \mathbf{Z}}$, surely. It can be verified that (\ref{ass h monotone}) holds for all ways of calculating $k_M$ described above, provided in Gaussian random fields the smoothness is estimated once based on the largest mask.

\section{Closed testing for cluster inference}

Having described classic cluster inference we can now construct its embedding into a closed testing procedure. We will use the theory of \cite{Goeman2020}, who provide a general method to construct a closed testing procedure from an existing multiple testing procedure. 
The proofs of all Lemmas and Theorems are in the Supplemental Information, Section A.

\subsection{Local test}

A closed testing procedure is built from local tests, which are hypothesis tests for a voxel set null hypothesis $H_V$. We will define such a local test for every voxel set $V\subseteq M$. For $V=\emptyset$ we may take $k_V =0$ without loss of generality.

Following \cite{Goeman2020} we note that in the discussion in the previous section the mask $M \subseteq B$ was arbitrary, and that the conclusions of that section hold for any fixed $M\subseteq B$. Following \cite{Goeman2020}, Theorem 2, we define as the local test for $H_V$ the test that rejects when cluster inference with mask $M=V$ rejects at least one voxel set null hypothesis. This test rejects when $\boldsymbol\phi_V=1$, where
\begin{equation} \label{def local test}
\boldsymbol\phi_V = \mathds{1}\{\chi_{V\cap\mathbf{Z}} > k_V\}.
\end{equation}
This is a valid local test due to the assumption that (\ref{eq CI weak FWER}) holds for every $M \subseteq B$, and therefore for $M=V$: we have for every $\mathrm{P} \in H_V$ that $\mathrm{P}(\boldsymbol\phi_V=1)\leq\alpha$. If $V=\emptyset$, then $\boldsymbol\phi_V = 0$, so the test never rejects. We will use the local test (\ref{def local test}) for every $V \subseteq M$ as the building block for the new closed testing procedure.

\subsection{Effective local test}

The local test ${\boldsymbol\phi}_V$ is a valid hypothesis test for the presence of signal in $V$ if the researcher restricted attention to $V$ before seeing the data. If the researcher chooses $V \subseteq M$ after seeing the data, a multiple testing correction needs to be performed over all $2^{|M|}$ hypothesis choices $(H_V)_{V\subseteq M}$. This is what closed testing does. 

\cite{Marcus1976} proved that such correction for multiple testing can be achieved by the effective local test, defined for any local test as
\[
\boldsymbol{\psi}_V = \min\{{\boldsymbol\phi}_W\colon V \subseteq W \subseteq M\}.
\]
The effective local test controls voxel set-level FWER over all $(H_V)_{V\subseteq M}$, having the property that for every $\mathrm{P} \in \Omega$,
\begin{equation} \label{eq CT FWER}
\mathrm{P}(\textrm{$\boldsymbol{\psi}_V = 0$ for all $V \subseteq M$ with $\mathrm{P} \in H_V$}) \geq 1-\alpha.
\end{equation}
Remembering that $\mathrm{P} \in H_V$ if and only if $H_V$ is true, we see that with probability at least $1-\alpha$ no true voxel set null hypothesis is rejected even when $\boldsymbol{\psi}_V$ is applied on all $V \subseteq M$.

\subsection{Shortcut} \label{sec short 1}

However, $\boldsymbol{\psi}_V$ is difficult to calculate, since it involves calculating $\boldsymbol\phi_W$, and therefore $k_W$, for exponentially many $V\subseteq W\subseteq M$. We propose to approximate ${\boldsymbol\psi}_V$ for every $V \subseteq M$ by an alternative test that is easier to compute:
\[
\underline{\boldsymbol\psi}_V = \mathds{1}\{\chi_{V \cap \mathbf{Z}} > k_M\}.
\]
For every $V\subseteq M$, the test $\underline{\boldsymbol\psi}_V$ rejects at most as often as ${\boldsymbol\psi}_V$, as Lemma \ref{lem use h} states. 

\begin{lemma} \label{lem use h}
For every $V \subseteq M$, we have $\underline{\boldsymbol\psi}_V \leq{\boldsymbol\psi}_V$.
\end{lemma}

The alternative test $\underline{\boldsymbol\psi}_V$ is a shortcut for the effective local test $\boldsymbol\psi_V$: it sacrifices some power for ease of computation. By Lemma \ref{lem use h}, $\underline{\boldsymbol\psi}_V$ retains the error guarentees of $\boldsymbol\psi_V$. Combining the lemma with (\ref{eq CT FWER}) we obtain voxel set-level FWER for $\underline{\boldsymbol\psi}_V$. For every $\mathrm{P} \in \Omega$,
\[
\mathrm{P}(\textrm{$\underline{\boldsymbol\psi}_V = 0$ for all $V \subseteq M$ with $\mathrm{P} \in H_V$}) \geq 1-\alpha.
\]

We can check that the test $\underline{\boldsymbol\psi}_V$ reproduces all the results of classic cluster inference. Classic cluster inference rejects all clusters $\mathbf{C} \subseteq M \cap \mathbf{Z}$ with $|\mathbf{C}| > k_M$. For such $\mathbf{C}$, we have $\chi_{\mathbf{C} \cap \mathbf{Z}} = \chi_{\mathbf{C}} = |C| > k_M$, so that $\underline{\boldsymbol\psi}_{\mathbf{C}} = 1$. 

However, $\underline{\boldsymbol\psi}_V$ allows useful additional conclusions that are not endorsed by classic cluster inference. If $A \subseteq B$ is an anatomical region of interest, we may reject $H_A$ and claim the presence of activity in $A$ if $\chi_{A\cap\mathbf{Z}} > k_M$, that is when there are at least $k_M$ connected supra-threshold voxels within $A$. This provides a partial solution to the desired inference problem 3 in the introduction to this paper, since it defines precisely how large a `substantial overlap' between a significant cluster and an anatomical region must be to allow a claim of activity in the region: the overlap must contain a connected area of size at least $k_M$. Note that the region of interest $A$ does not have to be chosen before seeing the data for such inference to be valid, since FWER control is over all $V \subseteq M$.

\subsection{True discovery proportions from closed testing}

The major gain of the closed testing formulation is not in voxel-set level FWER control, but in simultaneous TDP lower bounds for every cluster. We will use the methods of \cite{Genovese2006} and \cite{Goeman2011}.

Let $A_\mathrm{P} = \{v \in B\colon \mathrm{P} \notin H_v\}$ be the set of all truly active voxels in the brain. For voxel set $V \subseteq B$ the number of truly active voxels in $V$ is \[
a_\mathrm{P}(V) = |V \cap A_\mathrm{P}|.\] If the researcher would claim that voxel set $V$ is active, the researcher would be right about $a_\mathrm{P}(V)$ voxels, and wrong about $|V|-a_\mathrm{P}(V)$ of them. We call 
\[
\pi_\mathrm{P}(V) = \frac{a_\mathrm{P}(V)}{|V|},
\]
or 0 if $V=\emptyset$, the true discovery proportion (TDP) of set $V$. This is our target of inference. We will infer on $\pi_\mathrm{P}(V)$ through $a_\mathrm{P}(V)$, which is easier to work with.

\cite{Goeman2011} proved that, for any closed testing procedure with effective local tests $(\boldsymbol{\psi}_V)_{V\subseteq M}$, random variables defined, for all $V \subseteq M$, as
\begin{equation} \label{eq bound}
\mathbf{a}(V) = \min\{|V\setminus W|\colon W \subseteq V,\ \boldsymbol{\psi}_W = 0\},
\end{equation}
have the property that, for all $\mathrm{P} \in \Omega$, 
\begin{equation} \label{eq simultaneous tilde}
\mathrm{P}(\textrm{$\mathbf{a}(V) \leq a_\mathrm{P}(V)$ for all $V \subseteq M$}) \geq 1-\alpha.   
\end{equation}
A lower bound for the TDP follows immediately: $\boldsymbol\pi(V) = \mathbf{a}(V)/|V|$, or 0 if $V=\emptyset$, is a simultaneous lower bound for the TDP all $V \subseteq M$. By (\ref{eq simultaneous tilde}), for all $\mathrm{P} \in\Omega$, we have
\[
\mathrm{P}(\textrm{$\boldsymbol\pi(V) \leq \pi_\mathrm{P}(V)$ for all $V \subseteq M$}) \geq 1-\alpha.
\]

As argued by \cite{Goeman2011}, the lower bound $\mathbf{a}(V)$, and its companion $\boldsymbol\pi(V)$ provide much stronger statements than the effective local test. Where $\boldsymbol{\psi}_V$ only gives confidence whether or not there is signal present in $V$, $\mathbf{a}(V)$ gives confidence for the amount of signal. There is no information lost in reporting the TDP $\mathbf{a}(V)$ rather than rejection or non-rejection $\boldsymbol\psi_V$, since $\mathbf{a}(V) \geq \boldsymbol{\psi}_V$, as follows immediately from the definition. The simultaneity of (\ref{eq simultaneous tilde}) implies familywise error control over all $V \subseteq M$ considered or reported: with probability at least $1-\alpha$ no reported $\mathbf{a}(V)$, $V \subseteq M$, overestimates the number of truly active voxels $a_\mathrm{P}(V)$ in $V$, even if $V$ was chosen after seeing the data.

\subsection{Applying the shortcut}

Since $\mathbf{a}(V)$ involves the expression $\boldsymbol\psi_V$, which is difficult to calculate, we use the shortcut $\underline{\boldsymbol\psi}_V$ to get a partial shortcut for $\mathbf{a}(V)$. We write
\[
\check{\mathbf{a}}(V) = \min\{|V\setminus W|\colon W \subseteq V,\ \underline{\boldsymbol\psi}_W = 0\}.
\]
By Lemma \ref{lem use h}, $\check{\mathbf{a}}(V) \leq \mathbf{a}(V)$, so $\check{\mathbf{a}}(V)$ inherits the property (\ref{eq simultaneous tilde}). Moreover, $\check{\mathbf{a}}(V)$ can be rewritten in a relatively simple form. The formulation of $\check{\mathbf{a}}(V)$ and its property are our first main result. We formulate it as a theorem.

\begin{theorem} \label{thm first TDP bound}
Let 
\begin{equation} \label{def aV}
\check{\mathbf{a}}(V)= s_{k_M}(V \cap \mathbf{Z}),
\end{equation}
where $s_k(V) = \min\{|R|\colon \chi_{V\setminus R} \leq k\}$. Then, for all $\mathrm{P} \in\Omega$,
\begin{equation} \label{eq simultaneous}
\mathrm{P}(\textrm{$\check {\mathbf{a}}(V) \leq a_\mathrm{P}(V)$ for all $V \subseteq M$}) \geq 1-\alpha.   
\end{equation}
\end{theorem}

Although $\check{\mathbf{a}}(V)$ may yield smaller TDP than $\mathbf{a}(V)$, the resulting TDP lower bounds are still at least as powerful as the statements of classic cluster inference, as the next theorem asserts: all clusters found by classic cluster inference have a strictly positive TDP bound.

\begin{theorem} \label{thm uniform improvement}
If $\mathbf{C} \subseteq (\mathbf{Z} \cap M)$, with $|\mathbf{C}|>k_M$, is a cluster, then $\check{\mathbf{a}}(\mathbf{C})>0$.
\end{theorem}

\section{Calculating true discovery proportions}

The shortcut (\ref{def aV}) reduces a computation time of $\mathbf{a}(V)$ that is exponential in $|M|$ to a computation time for $\check{\mathbf{a}}(V)$ that is exponential in $|V|$. This is still prohibitive for most regions $V$. In this section we discuss algorithms for $\check{\mathbf{a}}(V)$. We show that this calculation is equivalent to solving a problem known as the $k$-separator problem in graph theory. For the specific case of that problem in the voxel graph with 26-connectivity, we obtain a lower bound to $\check{\mathbf{a}}(V)$ that has computation time $O(|V|^{1+1/d})$, and a fast heuristic algorithm, coupled with simulated annealing,  that approaches $\check{ \mathbf{a}}(V)$ from above. Both the lower bound and the simulated annealing algorithm rely on a duality between our $k$-separator problem and tiling problem on a slightly larger object, which we will derive and explain.

\subsection{The $k$-separator problem}

From Theorem \ref{thm first TDP bound} we see that we have efficient computation of $\check{\mathbf{a}}(V)$ whenever we can efficiently compute $s_k(V)$, for $V \subseteq \mathbf{Z}$. The value of $s_k(V)$ is the minimum number of voxels that must be removed from $V$ in order that the remainder falls apart into disconnected components of size $k$.
The quantity $s_k(V)$ can be defined for any graph, and is known in graph theory literature as the $k$-separator problem~\citep{ben2015separator}.
The $k$-separator problem is NP-hard, even for small fixed values of $k$. 
For example, with $k=1$ we have a classic vertex cover problem (NP-hard),
while for $k=2$ the problem is equivalent to the computation of dissociation number which is NP-complete for a class of bipartite graphs~\citep{yannakakis1981node}.
\cite{ben2015separator} proposed polynomial time solutions 
to several constrained variants of the $k$-separator problem; however, none of them is applicable in our case. 
In the next few sections we present novel solutions tailored to the specific type of graph induced by the neuroimaging context.

\subsection{Preliminaries}

Any voxel set $V$ can always be written as a union of disconnected clusters. 
The next lemma says that it is sufficient to calculate $s_k$ for these clusters.

\begin{lemma} \label{lem decompose clusters}
If $V = C_1 \cup \cdots \cup C_n$, where $C_1\ldots, C_n$ are disconnected clusters, then
\[
s_k(V) = \sum_{i=1}^n s_k(C_i).
\]
\end{lemma}

Without loss of generality, therefore, we can focus on calculating $s_k(V)$ only for $V \subseteq B$ that are clusters. However, the results in the remainder of this section are for general voxel sets $V$.

\subsection{Positive neighbors}

For our solutions to the $k$-separator problem we will exploit a duality between $k$-separating $V$ and tiling a somewhat larger object. To construct this duality we first need to introduce to $\mathbf{Z}^d$ the directed relationship of being `positive neighbors'.

We say that $w \in \mathbb{Z}^d$ is a \emph{positive neighbor} of $v \in \mathbb{Z}^d$ if $w-v \in \{0,1\}^d$. We write
\[
\{v\}^+ = \{v+e\colon e\in \{0,1\}^d\}
\]
for the voxel set of all positive neighbors of $v$. If $w \in \{v\}^+$ we call $v$ a \emph{negative neighbor} of $w$, since $v-w \in \{-1,0\}^d$. Note that the positive and negative neighbors do not partition the neighbors. For example, if $d=2$, $w=(-1,1)$, though a neighbor of $v=(0,0)$, is neither its positive or its negative neighbor. Moreover, every $v$ is always both a positive and a negative neighbor of itself.

The concept of the positive neighbors allows the definition of three useful derived voxel sets from every finite voxel set $V \subset \mathbb{Z}^d$. We define the \emph{cover} $V^+$ of $V$ as 
\[
V^+ = \{v+e\colon v \in V, e \in \{0,1\}^d\} = \bigcup_{v\in V} \{v\}^+
\]
the set of all voxels in $V$ and their positive neighbors. The \emph{interior} $V^-$ of $V$ is 
\[
V^- = \{v \in V\colon \textrm{\ $v+e \in V$ for all $e \in \{0,1\}^d$}\}.
\]
the set of all $v\in V$ that only have positive neighbors in $V$. Finally, the \emph{shave} of $V$ is $V^0 = V \setminus V^-$. This is the `positive edge' of $V$, the set of voxels in $V$ that have at least one positive neighbor outside $V$. These three derived voxel sets will allow us to rewrite the $k$-separator problem into a tiling problem.

\subsection{Tiling}

To calculate $s_k(V)$ we are interested in $k$-separators, defined as voxel sets $R\subseteq V$ with the property that $\chi_{V\setminus R} \leq k$. The value of $s_k(V)$ is the minimum $|R|$ over all $k$-separators. In this section we will show that minimizing $|R|$ over all $k$-separators is equivalent to minimizing a function $t_k(T_1,\ldots, T_n)$ over all tilings $T_1,\ldots, T_n$ of $V^+$. The latter will turn out to be an easier problem formulation to work with.

Define a \emph{tiling} of $V^+$ as a collection of pairwise disjoint voxel sets $T_1,\ldots,T_n$, called \emph{tiles}, such that
$
\bigcup_{i=1}^n T_i = V^+.
$
Note that every two distinct tiles from a tilling are disjoint as sets but their voxels may induce a connected graph. Given a tiling 
$T_1,\ldots,T_n$ of $V^+$, we will be interested in the function
\begin{equation} \label{def t}
t_k(T_1,\ldots,T_n) = \sum_{i=1}^n |T_i^0 \cap V| + \sum_{i=1}^n(|T_i^- \cap V| - k)_+,
\end{equation}
where $(\cdot)_+$ is the positive part function. This function is the link between tilings and $k$-separators, as the following two lemmas state.

\begin{lemma} \label{lem T to R}
For every tiling $T_1,\ldots,T_n$ of $V^+$ there exists a $k$-separator $R$ of $V$ such that
\[
|R| = t_k(T_1,\ldots,T_n).
\]
\end{lemma}

\begin{lemma} \label{lem R to T}
For every $k$-separator $R$ of $V$ there exists a tiling $T_1,\ldots,T_n$ of $V^+$ such that $T_1, \ldots, T_n$ are clusters, and
\[
|R| \geq t_k(T_1,\ldots,T_n).
\]
\end{lemma}

To get some intuition why these lemmas are true, it is helpful to consider a property of neighbors and positive neighbors proven as Lemma 9 in the Supplemental Information, Section A: two voxels are neighbors if and only if they have a common positive neighbor. It follows that voxel sets $V$ and $W$ are disconnected if and only if $V^+$ and $W^+$ are disjoint. It is this connection between disconnectedness of sets and simple disjointness of slightly larger sets that is exploited in Lemmas \ref{lem T to R} and \ref{lem R to T}. Loosely, if $R$ cuts $V$ as $V\setminus R = C_1 \cup\ldots\cup C_n$, with $C_1, \ldots, C_n$ pairwise disconnected, then $C_1^+, \ldots, C_n^+ \subseteq V^+$ are pairwise disjoint tiles. Vice versa if $T_1, \ldots, T_n \subseteq V^+$ are pairwise disjoint tiles, then their interiors $T_1^-, \ldots, T_n^-\subseteq V$ are paiwise disconnected; if these interiors are of size a most $k$, then $R=(T^0_1\cup \ldots \cup T^0_n)\cap V$ separates $V$. We illustrate the link between $k$-separator and tiling with an example in Figure \ref{fig tiling}.

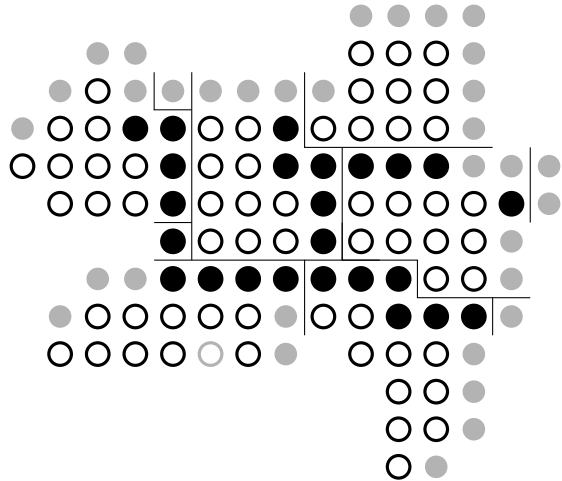
\begin{figure}[!ht]
\centering
\begin{tikzpicture}[x=-1cm, scale=.5]
\begin{scope}[very thick]
\foreach \x in {4} \draw (\x,1) circle(.3cm);
\foreach \x in {3,4} \draw (\x,2) circle(.3cm);
\foreach \x in {3,4} \draw (\x,3) circle(.3cm);
\foreach \x in {3,4,5,8,10,11,...,13} \draw (\x,4) circle(.3cm);
\foreach \x in {2,3,...,6,8,9,...,12} \draw (\x,5) circle(.3cm);
\foreach \x in {2,3,...,10} \draw (\x,6) circle(.3cm);
\foreach \x in {2,3,...,10} \draw (\x,7) circle(.3cm);
\foreach \x in {1,2,...,13} \draw (\x,8) circle(.3cm);
\foreach \x in {3,4,...,14} \draw (\x,9) circle(.3cm);
\foreach \x in {3,4,...,13} \draw (\x,10) circle(.3cm);
\foreach \x in {3,4,5,12} \draw (\x,11) circle(.3cm);
\foreach \x in {3,4,5} \draw (\x,12) circle(.3cm);
\end{scope}

\begin{scope}[black]
\foreach \x in {2,3,4} \fill (\x,5) circle(.3cm);
\foreach \x in {4,5,...,10} \fill (\x,6) circle(.3cm);
\foreach \x in {6,10} \fill (\x,7) circle(.3cm);
\foreach \x in {1,6,10} \fill (\x,8) circle(.3cm);
\foreach \x in {3,4,...,7,10} \fill (\x,9) circle(.3cm);
\foreach \x in {7,10,11} \fill (\x,10) circle(.3cm);
\end{scope}

\begin{scope}[fill=gray!60]
\foreach \x in {3} \fill (\x,1) circle(.3cm);
\foreach \x in {2} \fill (\x,2) circle(.3cm);
\foreach \x in {2} \fill (\x,3) circle(.3cm);
\foreach \x in {2,7} \fill (\x,4) circle(.3cm);
\foreach \x in {1,7,13} \fill (\x,5) circle(.3cm);
\foreach \x in {1,11,12} \fill (\x,6) circle(.3cm);
\foreach \x in {1} \fill (\x,7) circle(.3cm);
\foreach \x in {0} \fill (\x,8) circle(.3cm);
\foreach \x in {0,1,2} \fill (\x,9) circle(.3cm);
\foreach \x in {2,14} \fill (\x,10) circle(.3cm);
\foreach \x in {2,6,7,...,11,13} \fill (\x,11) circle(.3cm);
\foreach \x in {2,11,12} \fill (\x,12) circle(.3cm);
\foreach \x in {2,3,4,5} \fill (\x,13) circle(.3cm);
\end{scope}

\begin{scope}[gray!60, very thick]
\foreach \x in {9} \draw (\x,4) circle(.3cm);
\end{scope}

\draw (1.5,9.5) -- (6.5,9.5) -- (6.5, 11.5);
\draw (5.5,9.5) -- (5.5,6.5) --
		(3.5,6.5) -- (3.5,5.5) -- (0.5,5.5);
\draw (1.5,5.5) -- (1.5,4.5);
\draw (4.5,6.5) -- (6.5, 6.5) -- (6.5,4.5);
\draw (6.5, 6.5) -- (10.5,6.5);
\draw (5.5,6.5) -- (5.5, 7.5);
\draw (9.5,6.5) -- (9.5, 11.5);
\draw (9.5, 7.5) -- (10.5, 7.5);
\draw (9.5, 10.5) -- (10.5, 10.5) -- (10.5, 11.5);
\draw (0.5, 7.5) -- (0.5, 9.5);

\end{tikzpicture}
\caption{Illustration of a $k$-separator and a corresponding tiling, with $d=2$ and $k=10$. The voxel set $V$ comprises of all black voxels (open and filled). The set $V^+$ comprises of $V$ and all the gray voxels (open and closed). The $k$-separator $R$ is the set of all filled black voxels. The corresponding tiling is indicated by the lines. All filled voxels are part of the shave $T^0$ for their respective tile $T$; open voxels are part of the interior $T^-$.} \label{fig tiling}
\end{figure}

Combining Lemmas \ref{lem T to R} and \ref{lem R to T}, it follows that minimizing $|R|$ over all $k$-separators is equivalent to minimizing $t_k(T_1,\ldots,T_n)$ over all tilings. We formulate this result as a theorem.

\begin{theorem} \label{thm tiling}
We have
\[
s_k(V) = \min\{t_k(T_1,\ldots,T_n)\colon \textrm{\ $T_1,\ldots,T_n$  is a tiling of $V^+$}\}.
\]
The minimum is attained for a tiling for which $T_1,\ldots,T_n$ are all clusters.
\end{theorem}

Theorem \ref{thm tiling} rewrites the $k$-separator problem but does not simplify it. There is no obvious way to minimize $t_k(T_1,\ldots,T_n)$ in polynomial time.
However, we will exploit this theorem in the next three sections to construct a lower bound to $s_k(V)$, and a heuristic approximation to it.

\subsection{A lower bound} \label{sec short 2}

First, we construct a lower bound to $s_k(V)$. Replacing $s_k(V)$ by its lower bound in Theorem \ref{thm first TDP bound} retains the TDP guarantee implied by that theorem. As a consequence, the lower bound will be a shortcut to the closed testing procedure: it retains the guarantee on the TDP, but sacrifices some inferential power for computational reasons. We will derive this shortcut in two stages. First, in this section, we will calculate a shortcut with $O(|V|)$ time complexity. Next, in Section \ref{sec prune}, we will construct a more powerful shortcut in $O(|V|^{1+1/d})$ time.

The rationale behind the shortcut is that to minimize the expression (\ref{def t}) we should favor tiles $T$ with $|T^- \cap V| \leq k$, since for such tiles the second term of (\ref{def t}) disappears. For such tiles, minimizing $t$ amounts to finding tiles $T$ with as small as possible edge ratio $|T^0|/|T|$. However, if $|T^-| \leq k$, the edge ratio is bounded from below by the most efficient such ratio possible. This optimal edge ratio $r_k$ can be used to bound $s_k(V)$. We formulate this result as Theorem \ref{thm shortcut 2}.

\begin{theorem} \label{thm shortcut 2}
\[
s_k(V) \geq  r_k\cdot |V^+|-|V^+\setminus V|,
\]
where
\begin{equation} \label{def rk}
r_k = \min\{|V^0|/|V|\colon \emptyset \neq V \subset \mathbb{Z}^d,\ |V^-| \leq k\}.
\end{equation}
\end{theorem}

Define $\underline{s}_k(V) = r_k\cdot |V^+|-|V^+\setminus V|$. How can we interpret this lower bound? We see that $\underline{s}_k(V)$ is large if its size $|V|$ is large relative to the size $|V^+|$ of its cover. It takes large values therefore for large and compact $V$, and small values for smaller or irregular sets $V$. The calculation of $r_k$ is given in Lemma \ref{lem calculate rk}. We plot $r_k$ for $k=1,\ldots, 100$ and $d=2,3,4$ in figure \ref{fig rk}.

\begin{figure}[!ht]
\includegraphics[width=\textwidth]{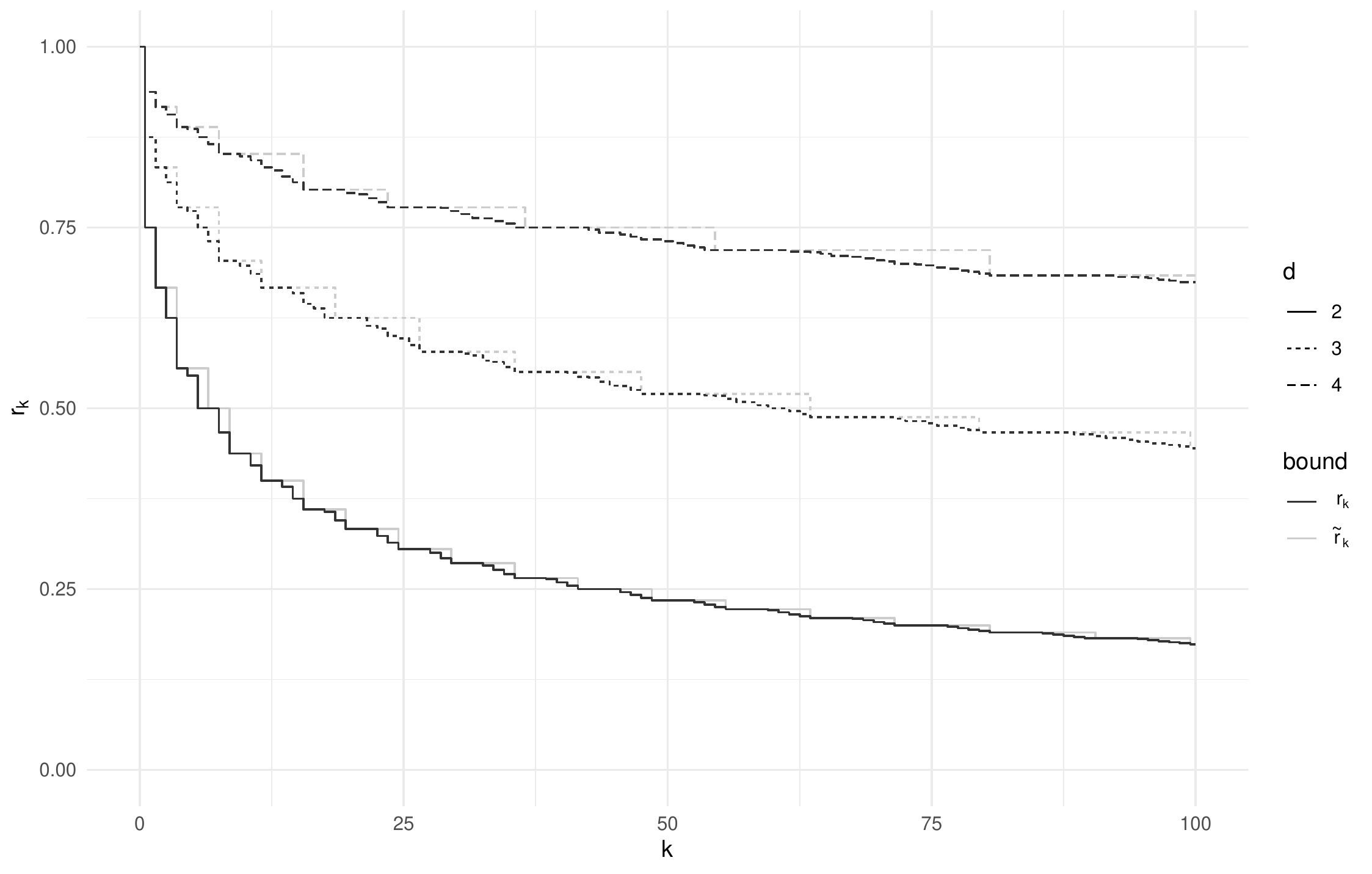}
\caption{The thresholds $r_k$ and $\tilde r_k$, defined in Theorem \ref{thm shortcut 2} and Lemma \ref{lem tilde r}, respectively, as a function of the extent threshold $k$ for dimensions $d=2,3,4$.} \label{fig rk} 
\end{figure}

\begin{lemma} \label{lem calculate rk}
If $k=0$, we have $r_k=1$. If $k>0$, we have 
\[
r_k = \min_{1\leq j \leq k} \frac{f_{d, j} - j}{f_{d, j}},
\]
where $f_{d,k}=0$ if $d=0$ or $k=0$, and, for $d>1$, we have recursively \[ f_{d,k} = b^+_{d,k} + f_{d-1, k- b_{d,k}}.\] Here, 
\[
b_{d,k} = \big(\lfloor k^{1/d}\rfloor\big)^{d-l_{d,k}} \big(\lfloor k^{1/d}\rfloor+1\big)^{l_{d,k}},
\]
and 
\[
b^+_{d,k} = \big(\lfloor k^{1/d}\rfloor+1\big)^{d-l_{d,k}} \big(\lfloor k^{1/d}\rfloor+2\big)^{l_{d,k}},
\]
where 
\[
l_{d,k} = \Big\lfloor \frac{\log(k) - d\log(\lfloor k^{1/d}\rfloor)}{\log(\lfloor k^{1/d}\rfloor+1) - \log(\lfloor k^{1/d}\rfloor)} \Big\rfloor.
\]
\end{lemma}

In the example object of Figure \ref{fig tiling}, we find from Lemma \ref{lem calculate rk} that with $d=2$ and $k=10$ we have $r_k=7/16$. With $|V|=84$ and $|V^+| = 118$, we get $\underline s_k(V) = 17.6$, so $s_k(V) \geq 18$.

\subsection{Pruning} \label{sec prune}

Irregularly shaped objects $V$ have low $\underline s_k(V)$. It can therefore pay to prune $V$ to $V' \subseteq V$ in order to bound $s_k(V)$ from below by $\underline s_k(V') \leq s_k(V') \leq s_k(V)$. We will use this to get an improved bound on $s_k(V)$.

Suitable choices are $V' = (V^-)^+$, $V'' = (((V^-)^-)^+)^+$, etc., which prune away increasingly broad extremities of $V$. We illustrate $V'$ in Figure \ref{fig prune}. In this example we have $|V'| = 78$, $|(V')^+| = 106$, and we find $\underline s_k(V') = 18.4$, so $s_k(V) \geq 19$. Further pruning to $V''$ leads to $|V''| = 69$ and $|(V'')^+| = 92$ for $\underline s_k(V'') = 17.3$ (see figures in the Supplemental Information, Section C). Further pruning does not lead to better bounds. In any case, as Lemma \ref{lem prune stop} states, pruning more than $|V|^{1/d}$ times is never necessary.

\begin{lemma} \label{lem prune stop}
If $i \geq \lfloor |V|^{1/d}\rfloor$, then $V^{(i)} = \emptyset$.  
\end{lemma}

\begin{figure}[!ht]
\centering
\begin{tikzpicture}[x=-1cm, scale=.5]
\begin{scope}[black]
\foreach \x in {3,4} \fill (\x,2) circle(.3cm);
\foreach \x in {3,4} \fill (\x,3) circle(.3cm);
\foreach \x in {3,4,5,10,11,...,12} \fill (\x,4) circle(.3cm);
\foreach \x in {2,3,...,6,8,9,...,12} \fill (\x,5) circle(.3cm);
\foreach \x in {2,3,...,10} \fill (\x,6) circle(.3cm);
\foreach \x in {2,3,...,10} \fill (\x,7) circle(.3cm);
\foreach \x in {2,...,13} \fill (\x,8) circle(.3cm);
\foreach \x in {3,4,...,13} \fill (\x,9) circle(.3cm);
\foreach \x in {3,4,...,13} \fill (\x,10) circle(.3cm);
\foreach \x in {3,4,5} \fill (\x,11) circle(.3cm);
\foreach \x in {3,4,5} \fill (\x,12) circle(.3cm);
\end{scope}

\begin{scope}[fill=gray!60]
\foreach \x in {2} \fill (\x,2) circle(.3cm);
\foreach \x in {2} \fill (\x,3) circle(.3cm);
\foreach \x in {2,9} \fill (\x,4) circle(.3cm);
\foreach \x in {1,7} \fill (\x,5) circle(.3cm);
\foreach \x in {1,11,12} \fill (\x,6) circle(.3cm);
\foreach \x in {1} \fill (\x,7) circle(.3cm);
\foreach \x in {1} \fill (\x,8) circle(.3cm);
\foreach \x in {1,2} \fill (\x,9) circle(.3cm);
\foreach \x in {2} \fill (\x,10) circle(.3cm);
\foreach \x in {2,6,7,...,13} \fill (\x,11) circle(.3cm);
\foreach \x in {2} \fill (\x,12) circle(.3cm);
\foreach \x in {2,3,4,5} \fill (\x,13) circle(.3cm);
\end{scope}

\begin{scope}[very thick]
\foreach \x in {4} \draw (\x,1) circle(.3cm);
\foreach \x in {8,13} \draw (\x,4) circle(.3cm);
\foreach \x in {1} \draw (\x,8) circle(.3cm);
\foreach \x in {14} \draw (\x,9) circle(.3cm);
\foreach \x in {12} \draw (\x,11) circle(.3cm);
\end{scope}

\begin{scope}[gray!60, very thick]
\foreach \x in {3} \draw (\x,1) circle(.3cm);
\foreach \x in {7} \draw (\x,4) circle(.3cm);
\foreach \x in {13} \draw (\x,5) circle(.3cm);
\foreach \x in {0} \draw (\x,8) circle(.3cm);
\foreach \x in {0} \draw (\x,9) circle(.3cm);
\foreach \x in {14} \draw (\x,10) circle(.3cm);
\foreach \x in {11,12} \draw (\x,12) circle(.3cm);
\end{scope}

\end{tikzpicture}
\caption{Illustration of the pruning $V'$ of the voxel set $V$ from Figure \ref{fig tiling}. The voxel set $V$ consists of all black voxels (open and filled); the set $V^+$ additionally comprises of the gray  voxels (open and filled). The pruned set $V'=(V^-)^+$ consists of the filled black voxels, and its cover $(V')^+$ of all filled grey voxels. We see that each voxel removed to obtain $V'$ nets a reduction in size of 2 voxels for $(V')^+$, resulting in a net gain in $\underline s_k(V')$ relative to $\underline s_k(V)$, since $r_k \leq 1/2$.} \label{fig prune}
\end{figure}
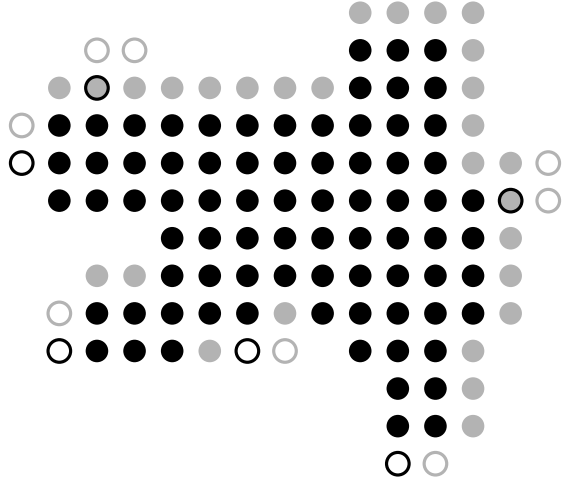
 
Taking pruning into account, and using that $s_k(V)>0$ if $|V|>k$, we define the improved bound
\[
\check s_k(V) = \mathds{1}\{\chi_V>k\} \vee \max \Big\{\big\lceil\underline  s_k\big(V^{(i)}\big)\big\rceil\colon i=0,1,\ldots, |V|^{1/d}\Big\},
\] 
where $V^{(i)}$ is obtained from $V$ by performing the $(\cdot)^-$ operation $i$ times, followed by the $(\cdot)^+$ operation $i$ times.

Taking everything together, the proposed procedure and its TDP guarantee property are summarized in the following theorem, which proves that the lower bound is a shortcut to the closed testing procedure.

\begin{theorem} \label{thm final}
For every $V \subseteq M$, let 
\[
\underline{\mathbf{a}}(V) = \sum_{i=1}^\mathbf{n} \check s_{k_M} (\mathbf{C}_i),
\]
where $\mathbf{C}_1,\ldots,\mathbf{C}_\mathbf{n}$ are disconnected clusters such that $\mathbf{C}_1 \cup \cdots \cup \mathbf{C}_\mathbf{n} = V \cap \mathbf{Z}$. Then, for all $\mathrm{P} \in\Omega$,
\[
\mathrm{P}(\textrm{$\underline {\mathbf{a}}(V) \leq a_\mathrm{P}(V)$ for all $V \subseteq M$}) \geq 1-\alpha.   
\]
\end{theorem}

Computational complexity for $\underline s_k(V)$ is $O(|V|)$, and for $\check s(V)$ is $O(|V|^{1+1/d})$, so that is also the computational complexity of $\underline{\mathbf{a}}(V)$ if $V$ is a supra-threshold cluster. For general $V$, complexity is the sum of the complexity of its comprising clusters, which is  $O(|V|^{1+1/d})$ in the worst case that $V$ is a supra-threshold cluster.

It is easy to verify that the shortcut of Theorem \ref{thm final} also retains the property of Theorem \ref{thm uniform improvement} that it uniformly improves classic cluster inference. Still, it sacrifices some power, since the lower bound $\check s_k(V)$ may be (much) smaller than $s_k(V)$. The difference between $s_k(V)$ and $\check s(V)$ can be expected to be relatively large especially if $|V|/k$ is small and if $V$ is irregularly shaped.

\subsection{Heuristic algorithms to minimize $k$-separators} \label{sec simul anneal}

The strength of the shortcut of the previous paragraph is its guaranteed TDP control, as expressed in Theorem \ref{thm final}. To obtain this control the shortcut sacrifices power in exchange for computational efficiency. In this section we present an alternative computational approach that aims to approximate  $s_k(V)$ heuristically as closely as possible, instead of bounding it from below. The algorithm has two parts. First, a heuristic algorithm finds a good separator. Next, an attempt is made to find a local improvement of the solution using simulated annealing. The second phase of the algorithm uses Theorem \ref{thm tiling}.

The first heuristic algorithm finds clusterings with acceptable sizes of separator sets. The algorithm consists of two phases: inferring an initial clustering, and improving regions consisting of a small number of neighbouring clusters. In the first phase, the algorithm starts from an empty clustering. It generates a small number of candidate clusters, where the number is a small integer, usually between 1 and 10. Each candidate cluster is created starting from a randomly chosen available voxel by a sequence of insertions of adjacent voxels such that the induced size of its separator is kept small. Then, the best candidate cluster, i.e., the cluster with the separator's minimal size, is inserted into the current clustering. The procedure is repeated until there is no space to insert a new cluster. The second phase consists of repetitions of local improvements. The algorithm randomly takes a small number of neighbouring clusters, removes them from the current clustering, and applies a procedure similar to the first phase to find a better setting of clusters.  

We follow up on the optimal heuristic separator using a simulated annealing algorithm, as follows. The separator of $V$ is translated to a tiling of $V^+$ according to Lemma \ref{lem R to T}. In each step, the algorithm chooses a random voxel $v \in V^+$ and a random neighbor $w \in V^+$ of $v$. If $v$ and $w$ are part of the same tile $T$ with interior size $|T^- \cap V|> k$, the algorithm proposes to start a new tile $\{v\}$; otherwise it proposes to reassign $v$ from its old tile to the tile of $w$. If the target function $t'$ of the proposed tiling is lower than or equal to the target function $t$ of the previous step, the proposal is always accepted. Otherwise, the proposal is accepted with a probability that is a decreasing function of $t'-t$ and of the current iteration number. After a maximum number of iterations is reached, the algorithm returns the best solution it found during its travels through the search space. 

The first algorithm was implemented in C and the simulated annealing in Python. The algorithms are usually invoked with a time limit setting. Pseudo-code for both heuristic algorithms are given in the Supplemental Information, Section B.



The heuristic algorithms are not guaranteed to find the global minimum with a finite running time. 
If the algorithm did not find the correct solution, the value found is larger than the actual minimum $s_k(V)$, so there is no formal guarantee of TDP control comparable to Theorem \ref{thm final}. Still, the overstatement of $s_k(V)$ may often be less than the understatement of $s_k(V)$ due to the lower bound (\ref{thm shortcut 2}). The heuristic approach may therefore be the preferred solution in practice if computation time is not an issue and a small overstatement of TDP is acceptable. 

\subsection{Heuristic algorithm performance}

A heuristic algorithm for a computationally hard problem cannot guarantee to find the optimal solution. Also estimating the error of such approaches is usually a difficult task. One way to proceed is to use exact solution approaches such as exhaustive enumeration, dynamic programming, or integer linear programming formulations. However, in the case of intractable problems, they can only be applied to small instances. Here, we propose a different approach. First, we show that some instances of the k-separator problem are tractable by showing their exact solution. Next, to estimate an error of the heuristic algorithm given the input consisting of multiple datasets, we generate a number of tractable instances matching properties of the input and jointly apply the heuristic algorithm under the same parameter setting. Finally, knowing the exact solution of tractable instances, we can estimate the solution error of the input datasets. The main result is formulated below in Lemma~\ref{lemma:heurbounds}

\begin{lemma}
\label{lemma:heurbounds}
Let $k=n^d$ and $c$ be a vector of $d$ positive integers. If the dimensions of a hyperrectangle $R$ are $(n+1)c_i - 1$ for $i = 1,\ldots,d$, then the bound of Theorem \ref{thm shortcut 2} is exact, so that the optimal $k$-separator of $R$ has $|R|-n^d\Pi c_i$ voxels.
\end{lemma}

Since our algorithm is not utilizing the information on the shape of the input voxel sets, nor the clusters are formed as cubes in the sampling, we believe that the benchmark of correctness based on hyperrectangles is a good indicator of how scores from the heuristic differ from the optimal ones.

To estimate the error of the heuristic algorithm, we inferred a collection of hyperrectangle tests based on the three-dimensional datasets from the Neurovault repository (see~Section~\ref{sec neurovault}). Our goal was to cover the whole range of $k$ values and datasets sizes from the input repository. Therefore, we set $k$ bounded above $1000$, and the hyperrectangle, i.e., cuboid, sizes to maximum $18000$ voxels. Being consistent with the notation from Lemma~\ref{lemma:heurbounds}, each test is uniquely determined by four integer parameters $n, c_1, c_2, c_3 \leq 10$, where $k$ is $n^3$, and the corresponding cuboid has dimensions $(n+1)c_i-1$, for each $i$.
After rejecting too large cuboids, we obtained $1064$ tests, which enlarged the input repository by nearly $9\%$. 

The experiment indicated that nearly $50\%$ of tests were completed with no error, and the worst errors of $5-6\%$ has only $\sim5\%$ of tests. A more detailed summary is depicted in Figure~\ref{fig hist err} with boxplots of errors for each value of $k$, where the $0\%$ represents no error. 
The results obtained on cuboid tests indicate that the sizes of separators inferred by our heuristic algorithm are optimal in nearly half of the cases. For the rest of the cases, the error is usually below 4\% with high confidence and the median error is below $2\%$. 


\begin{figure}[!ht]
\includegraphics[width=\textwidth]{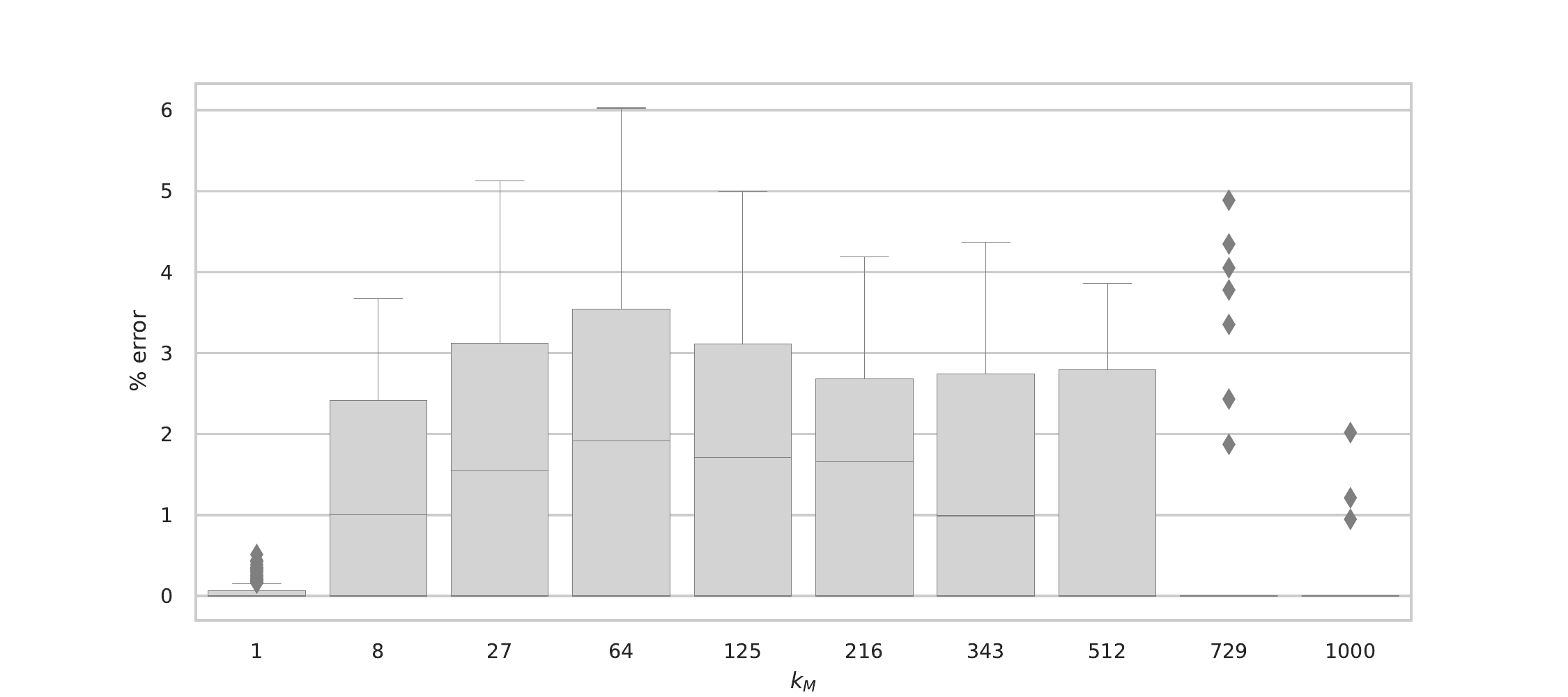}\caption{Upper bound heuristic performance: a boxplots of errors as a percentage of the true $k$-separator. Note that all errors are overestimates by construction.} \label{fig hist err} 
\end{figure}

\section{Choosing thresholds}

\subsection{Voxel-wise inference}

An alternative to cluster extent inference is classic voxel-wise inference. In voxel-wise inference, FWER is controlled over all voxel-wise null hypotheses. This is achieved by finding the $(1-\alpha)$-quantile of the distribution of the maximal $z$-score under the global null hypothesis $H_M$, and rejecting the null hypothesis whenever a voxel's $z$-score exceeds this threshold. 
Though cluster extent inference is often contrasted sharply with voxelwise inference, suggesting that these are two very different modes of operation. It was noted by \cite{Poline1997} and \cite{Friston1994} that classic voxelwise inference is simply a special case of cluster extent inference, obtained by choosing $k_M=0$. It follows that we can get a TDP per cluster from voxelwise inference.

In classic voxelwise inference, we reject $H_v$ for all voxels $v \in \mathbf{Z} = \{v\in M\colon \mathbf{z}_v \geq z\}$, where $z$ is chosen as the smallest value such that 
\begin{equation} \label{eq voxelwise FWER}
\mathrm{P}(|M \cap \mathbf{Z}| > 0) \leq \alpha
\end{equation}
holds for all $\mathrm{P} \in H_M$. It has been shown \citep{Worsley1992, Friston1991} that voxelwise inference controls voxelwise FWER, i.e., for all $\mathrm{P} \in \Omega$, 
\[
\mathrm{P}(\mathbf{Z}\not\subseteq A_\mathrm{P}) \leq \alpha.
\]

We can embed voxelwise inference into the closed testing procedure we have constructed by remarking that $|M \cap \mathbf{Z}| > 0$ if and only if $\chi_{M \cap \mathbf{Z}} >0$. Therefore (\ref{eq voxelwise FWER}) is equivalent to
\[
\mathrm{P}(\chi_{M \cap \mathbf{Z}} > 0) \leq \alpha,
\]
which is simply (\ref{eq CI weak FWER}) with $k_M=0$, and the latter is a valid choice for $k_M$. The closed testing procedure resulting from this choice is a relatively simple one, as the following theorem states.

\begin{theorem} \label{thm kM0}
If $k_M=0$, then for all $V \subseteq M$ we have 
\[
\underline{\mathbf{a}}(V) = \check{\mathbf{a}}(V) = \mathbf{a}(V) = |V \cap \mathbf{Z}|.
\]
\end{theorem}

The theorem says how to calculate TDP for clusters when doing voxelwise inference: the TDP lower bound for a set $V$ is simply the fraction of voxelwise significant voxels among the voxels in $V$. Supra-threshold clusters obtained with $k_M=0$ always have a TDP of 100\%.

\subsection{Choosing $k_M$} \label{sec thresholds}

Cluster extent inference assumes that $z$ and $k_M$ are chosen in such a way that (\ref{eq CI weak FWER}) holds. It is common in cluster extent inference to fix the $z$-score threshold $z$, and to calculate $k_M$ as the smallest value such that (\ref{eq CI weak FWER}) is satisfied \citep{Friston1994}. However we saw in the previous section that the order is reversed in voxelwise inference: there $k_M=0$ is fixed, and $z$ is chosen as the smallest value of $z$ satisfying (\ref{eq CI weak FWER}). In this section, we argue that the order of fixing $k_M$ calculating $z$ should be generally preferred, both from the perspective of power and obtaining a good TDP bound.

It is perfectly valid to choose $k_M$ first, and to find a value of $z$ that corresponds to this $k_M$, as previously proposed by \citet{Bullmore1999}. The relationship between $z$ and $k_M$ depends only on the null model $H_M$, and not on the observed $z$-scores. For cluster inference based on random field theory, the relationship between $z$ and $k_M$ depends on the smoothness of the field, which is estimated from the independent residuals. For cluster inference based on permutations, $k_M$ is calculated from the matrix of all permutation $z$-scores, and can be calculated without knowing which permutation corresponds to the real data. We present a fast algorithm for finding $z$ based on $k_M$ using permutations in the Supplemental Information, Section E.

It is generally (slightly) more powerful to choose $k_M$ rather than $z$. The reason for this is that $k_M$ is discrete, while $z$ is continuous. When fixing $z$ and calculating $k_M$ there is almost always a smaller value of $z$ that would result in the same value of $k_M$. Using this value instead of the previously chosen $z$ would result in a uniformly more powerful method but still controls TDP, since (\ref{eq CI weak FWER}) still holds. We may therefore, after choosing $z$ and finding $k_M$, always re-calibrate our $z$.

Alternatively, we may simply choose $k_M$ and find $z$ as the smallest value such that (\ref{eq CI weak FWER}) holds, as is done in voxelwise inference. This has the important advantage that the achievable TDP can be better controlled.

\section{Upper bounds}

In this section we present two upper bound results that impose hard limits on the TDP that can be achieved with closed testing based on  cluster extent inference. The first bound, in Section \ref{sec thresholds}, limits what can be achieved using the lower bound; this result helps to choose the settings of that method. The second bound, in Section \ref{sec upper bound}, limits what can be achieved in terms of TDP by the full closed procedure (\ref{eq bound}). Since closed testing procedures can only be uniformly improved by improving their local tests \citep{Goeman2020}, and that the room for such improvements is limited if (\ref{eq CI weak FWER}) is tight, this sets a limit on the potential of any method that is consistent with classic cluster extent inference.

The maximal achievable TDP from the shortcut can be calculated as a function of $k_M$ and cluster size $|\mathbf{C}|$ by the following theorem. 

\begin{theorem} \label{thm bound rk}
For every cluster $\mathbf{C} \subseteq \mathbf{Z}$, we have 
\[
\underline{\mathbf{a}}(\mathbf{C}) \leq \Big\lceil \frac{r_{k_M}-r_{|\mathbf{C}|}}{1-r_{|\mathbf{C}|}} \cdot |\mathbf{C}| \Big\rceil \vee \mathds{1}\{|\mathbf{C}| > k_M\}.
\]
\end{theorem}

By Theorem \ref{thm bound rk}, to achieve a TDP of $\gamma$, for some $\gamma > 1/{k_M}$, we need a cluster $\mathbf{C}$ with 
\[
r_{|\mathbf{C}|} \leq \frac{r_{k_M} - \gamma}{1-\gamma}.
\]
The maximal TDP according to $\underline a$ for different values of $k_M$ and different cluster size $|\mathbf{C}|$ is given in Figure \ref{fig max tdp}. Since $r_{|\mathbf{C}|} \to 0$ as $|\mathbf{C}| \to \infty$, the TDP lower bound $\underline{\mathbf{a}}(\mathbf{C}) / |\mathbf{C}|$ achieved by the shortcut of Theorem \ref{thm final} is at most $r_{k_M}$ for very large clusters, and much smaller than that for small and irregular clusters. The maximal TDP values converge to $r_k$ as the cluster size increases. Clusters may achieve the maximal TDP if they are highly compact. Irregular clusters tend to have (much) smaller TDP. 

\begin{figure}[!ht]
\includegraphics[width=\textwidth]{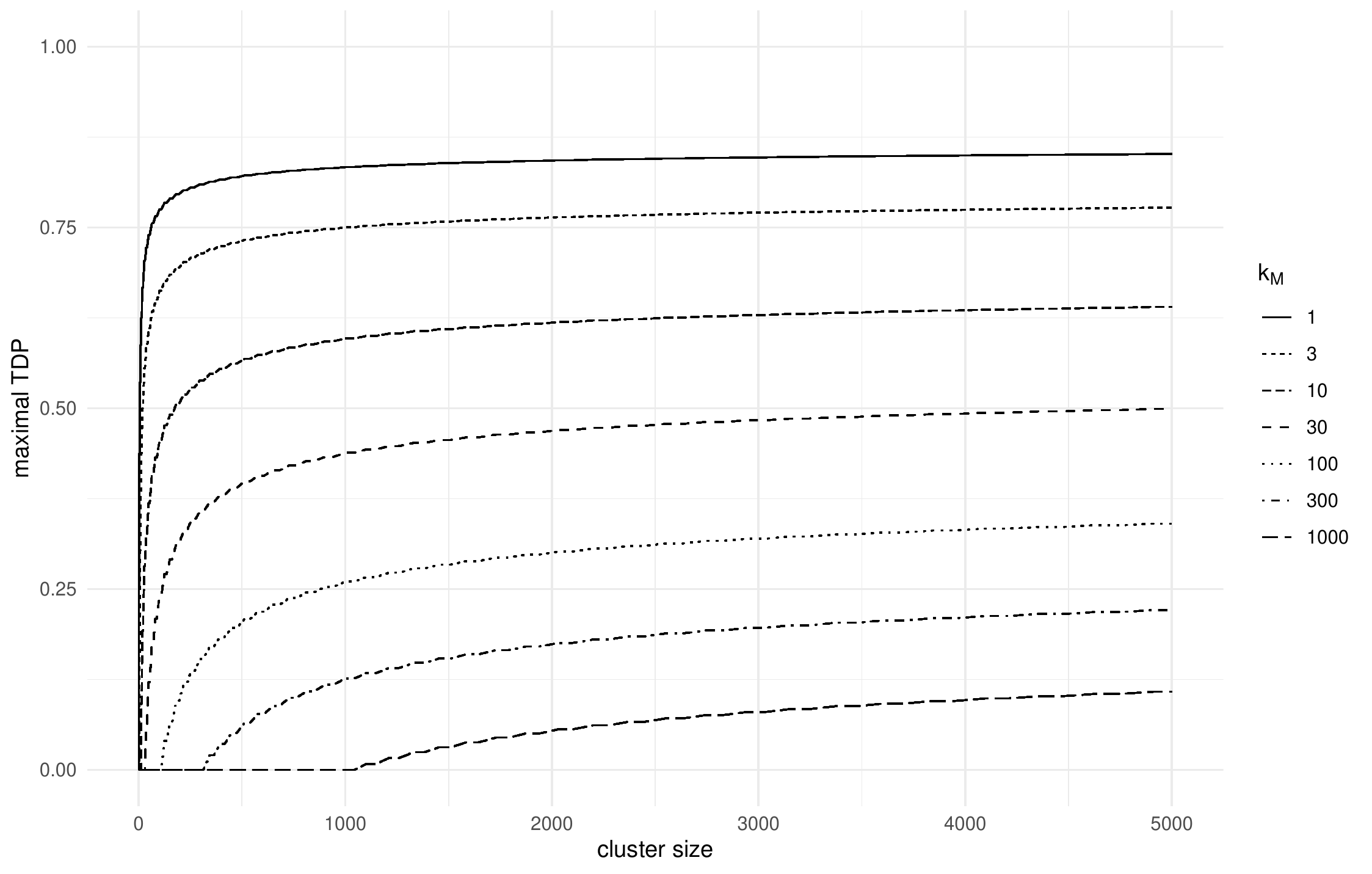}
\caption{The maximal TDP according to the shortcut $\underline a (V)$, defined in Theorem \ref{thm final}, as a function of the extent threshold $k_M$ and cluster size, for dimensions $d=3$.} \label{fig max tdp} 
\end{figure}

We see from Figure \ref{fig max tdp} that, with large values of $k_M$, it is difficult or even impossible to achieve good TDP even for large clusters, so a small value of $k_M$ is recommended if large TDP is desired. Assuming that we are interested in finding clusters with TDP $\geq 1/2$, a sweet spot with $d=3$ seems to be $k_M=14$, for which $r_{k_M}=2/3$. To achieve TDP $\geq 1/2$, clusters need to have $r_{|V|}\leq 1/3$, which implies $|V|\geq 339$.

For $\check{\mathbf{a}}(V)$ we have a weaker bound $\check{\mathbf{a}}(V) \leq \tilde r_{k_M}\cdot |V|$, from Lemma \ref{lem tilde r}, below, that bounds TDP by $\tilde r_k \approx r_k$. The value of $\tilde r_k$ is illustrated in Figure \ref{fig rk} in Section \ref{sec short 2}. This bound suggests that also when using the heuristic approximation to the $k$-separator problem, a researcher would want to use a value of $k_M$ that yields $\tilde r_{k_M}$ substantially above the target TDP. E.g., getting a TDP over 0.5 is impossible if $k_M > 64$, and remains unlikely unless $k_M$ is substantially smaller than 64, since the bound of Lemma \ref{lem tilde r} is not very tight.

\begin{lemma} \label{lem tilde r}
We have $s_k(V) \leq \tilde r_k \cdot |V|$, where $\tilde r_k = (b^+_{d,k} - b_{d,k})/b^+_{d,k}$.
\end{lemma}

Note that $b_{d,k}$ and $b^+_{d,k}$ are defined in Lemma \ref{lem calculate rk}.

\subsection{The limits of cluster extent thresholding} \label{sec upper bound}

In the previous section we considered upper bounds for the shortcuts to the closed testing procedure. Such bounds are useful for researchers intending to use these shortcuts. In this section we consider an upper bound to the full closed testing procedure (\ref{eq simultaneous tilde}), given below in Theorem \ref{thm bound aV}. This bound is of fundamental and practical interest, as we will explain. 

\begin{theorem} \label{thm bound aV}
Let $\overline{\mathbf{a}}(V) = s_{k_{M\setminus \mathbf{Z}}}(V\cap \mathbf{Z})$, then, for every $V \subseteq M$,
\[
\mathbf{a}(V) \leq \overline{\mathbf{a}}(V).
\]
\end{theorem}

In the proof of this theorem in the Supplemental information (Section A) we will prove a slightly tighter bound. Note the similarity of $\overline{\mathbf{a}}(V)$ with $\check{\mathbf{a}}(V)$, the only difference being that $k_M$ is replaced by $k_{M\setminus \mathbf{Z}}$. This difference will be small unless $|\mathbf{Z}|$ is large relative to $|M|$. 

Practically, Theorem \ref{thm bound aV} can be used to bound the loss $\mathbf{a}(V) - \underline{\mathbf{a}}(V)$ of the shortcut $\underline{\mathbf{a}}(V)$ relative to the full closed testing procedure $\mathbf{a}(V)$. It limits the potential for further computational improvements. In practice, unless $|\mathbf{Z}|$ is large relative to $|M|$ we will have $k_M \approx k_{M\setminus \mathbf{Z}}$, so that $\check{\mathbf{a}}(V) \approx  \overline{\mathbf{a}}(V)$, and $\check{\mathbf{a}}(V) \approx  \mathbf{a}(V)$.

More fundamentally, we can combine Theorem \ref{thm bound aV} with the insights from \cite{Goeman2020}. We have constructed $\mathbf{a}(V)$ as the unique closed testing procedure induced by cluster extent inference. By \cite{Goeman2020} closed testing procedures are optimal, so there is no room for improvement of the method outside the closed testing framework. Moreover, improvement within the closed testing framework is limited to improvement of the local test, and there is hardly room for that if $z$ and $k_M$ are optimized for (\ref{eq CI weak FWER}). It follows that Theorem \ref{thm bound aV} gives a clear upper bound to the TDP arising from any method that is based on cluster extent thresholding. Any method that achieves the result of Theorem \ref{thm uniform improvement} would also be constrained by the result of Theorem \ref{thm bound aV}.

\section{Simulation} \label{sec sim}

In this section, Monte Carlo simulation is conducted to demonstrate the validity of our proposed methods, to investigate the tightness the TDP, and to see the gap between the upper and lower TDP bounds.

\subsection{Set-up}

2D images, each with $128 \times 128$ pixels, were simulated. Two spatial signal configurations were considered, shown in Figure \ref{fig: signal}: (1) a focal configuration with a single large circle of signal in the middle and (2) a distributed configuration with 9 small circular regions of signal spread out. The number of pixels with signal was 716 for both configurations. The simulated images were created by filling each pixel with spatially correlated noise, starting from i.i.d.\ standard Gaussian noise and smoothing with a spatial Gaussian smoothing kernel with full width at half maximum (FWHM) of 4 pixels, i.e., with $\sigma=1.7$ pixels. Signal was added according to the chosen configuration at a fixed signal amplitude of $d=0.1$ and $d=0.05$, respectively. We considered 20 sample sizes $n$ between 10 and 200 with an increment of 10, and a total of 1000 images were generated for each simulation setting. We calculated $z$-scores for each voxel using a one-sample $t$-test. Clusters of interest were defined as all connected components of $\mathbf{Z}$ as defined in Section \ref{sec classical}, using $z$-score thresholds $z=0.348 \times \sqrt{n}$ for each sample size $n$. 

To calculate the $k_M$ threshold at $\alpha=0.05$ fulfilling (\ref{eq CI weak FWER}) we simulated a second independent null field without signal for each combination of each sample size and threshold, smoothed in the same way. We calculated $k_M$ as the 95\% quantile of the empirical distribution of the maximum cluster size in this null field. Clusters of size $k_M$ or smaller were discarded in accordance with standard practice. Subsequently, the TDP bound was calculated using both the heuristic algorithm of Section \ref{sec simul anneal} and the lower bound of Theorem \ref{thm final}.
 
\begin{figure}
\begin{subfigure}[b]{0.5\textwidth}
\centering
\includegraphics[width=0.8\textwidth]{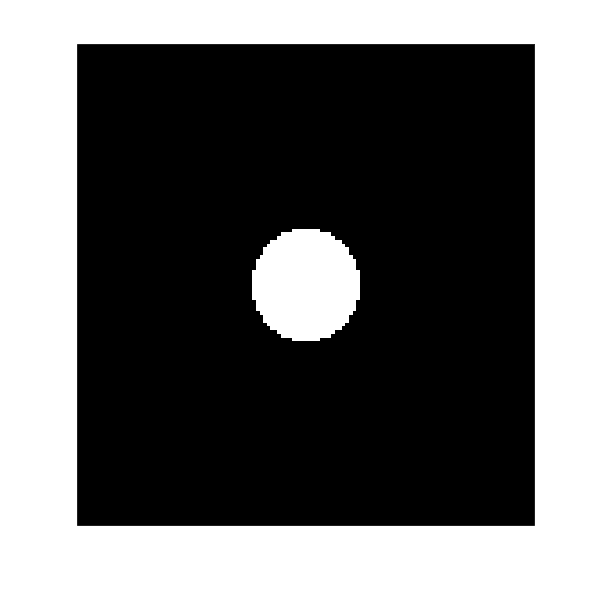}
\label{focal}
\end{subfigure}
\hfill
\begin{subfigure}[b]{0.5\textwidth}
\centering
\includegraphics[width=0.8\textwidth]{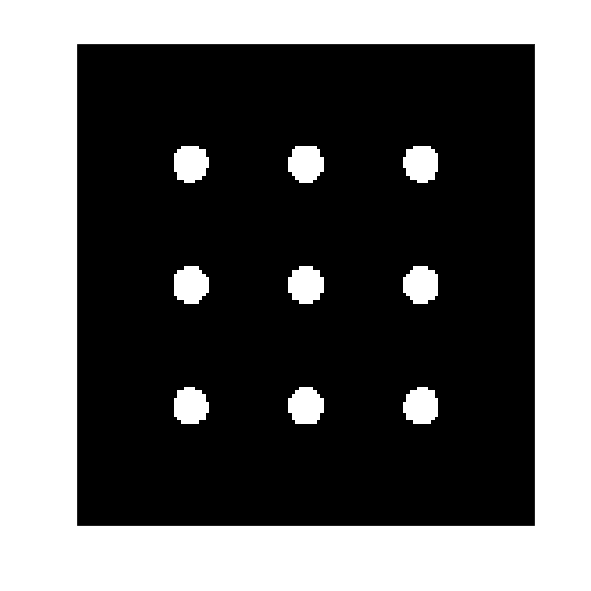}
\label{dist}
\end{subfigure}
\caption{2D simulated signal illustration. Focal signal (left) with one large circle in the middle; distributed signal (right) with nine identical circular regions.}
\label{fig: signal}
\end{figure}

\subsection{Results}

Figure \ref{fig: size} shows the average size of the clusters found, illustrating the qualitative difference between the two signal amplitudes. Here, the cluster size is standardized to a percentage based on the true signal size. At the high amplitude ($d=0.1$) the clusters are consistent for the signal, with clusters converging to the true signal as the sample size increases. In contrast, at the low amplitude the clusters capture a vanishing fraction of the true signal. 

Figure \ref{fig: FWER} shows the error rate of the method, which is well controlled at $\alpha=0.05$ for all settings. The lower bound is conservative for large and for small sample sizes, while the heuristic algorithm is only conservative for large sample size. We explain this for small sample size by the compactness of the chosen signal regions, for which the lower bound method tends to underestimate TDP. For large sample size, conservativeness is due to discreteness of $k_M$, so that the $\alpha$-level in (\ref{eq CI weak FWER}) is not exhausted. The heuristic algorithm also controls its error rate quite well in this simulation, despite the lack of a theoretical guarantee.

Figure \ref{fig: TDP} shows the TDP bounds found by the method. Displayed is the average value of the TDP over all significant clusters, i.e.\ over all clusters with $\text{TDP}>0$. Note that the number of such clusters is much smaller for the low signal amplitude setting than for the high amplitude setting, and much larger for the distributed configuration of signal than for the focal one. We see that in all settings the TDP of significant clusters goes to 1 as sample size increases. This is because the value of $k_M$ decreases with the sample size, eventually reaching $k_M=0$. The difference in TDP between the lower bound and the heuristic algorithm is appreciable but not overly large, almost never exceeding 10\%.

\begin{figure}
\centering
\includegraphics[width=\textwidth]{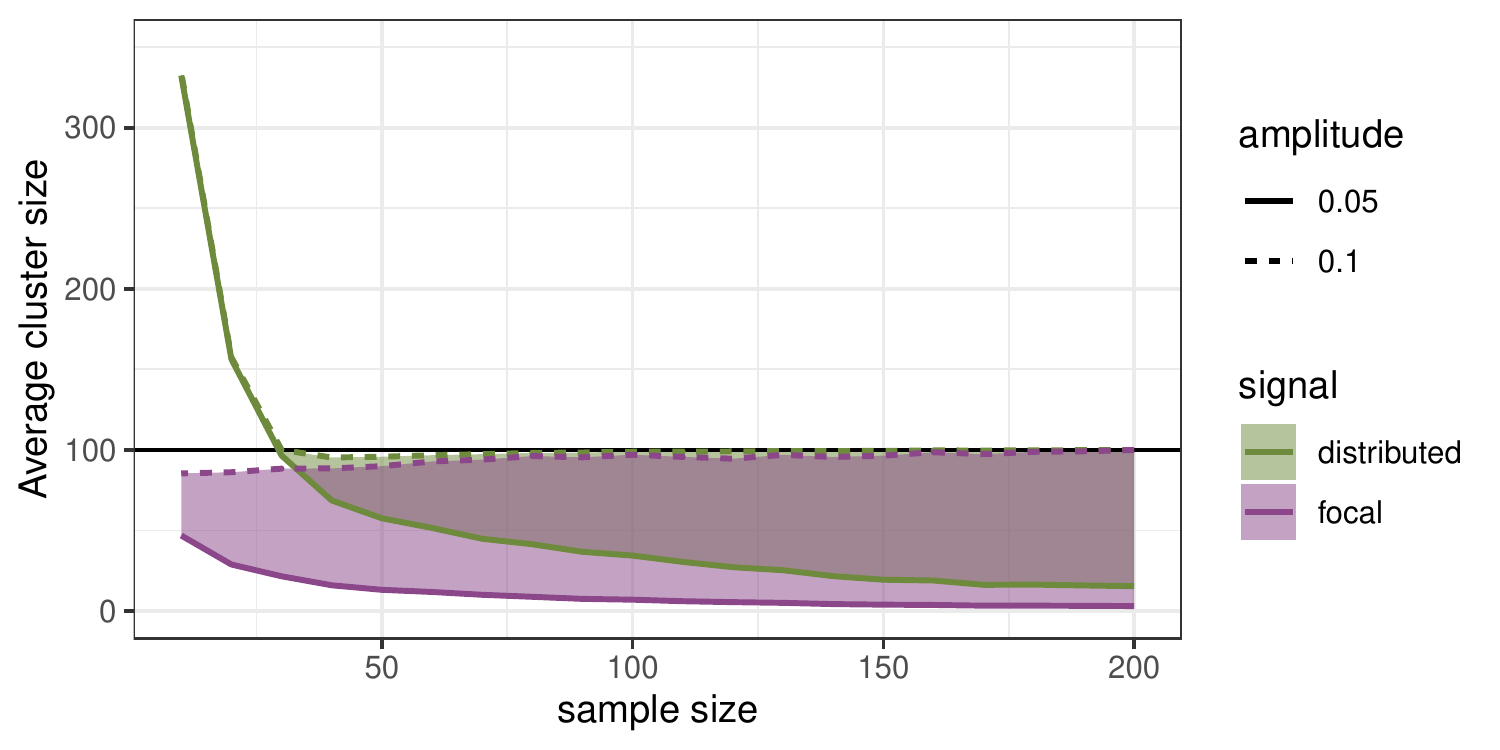}
\caption{Average cluster sizes (standardized and expressed in percent) for focal (purple) and distributed (green) signals with the amplitudes of $d=0.1$ (dashed line) and $d=0.05$ (solid line).}
\label{fig: size}
\end{figure}

\begin{figure}
\centering
\includegraphics[width=0.9\textwidth]{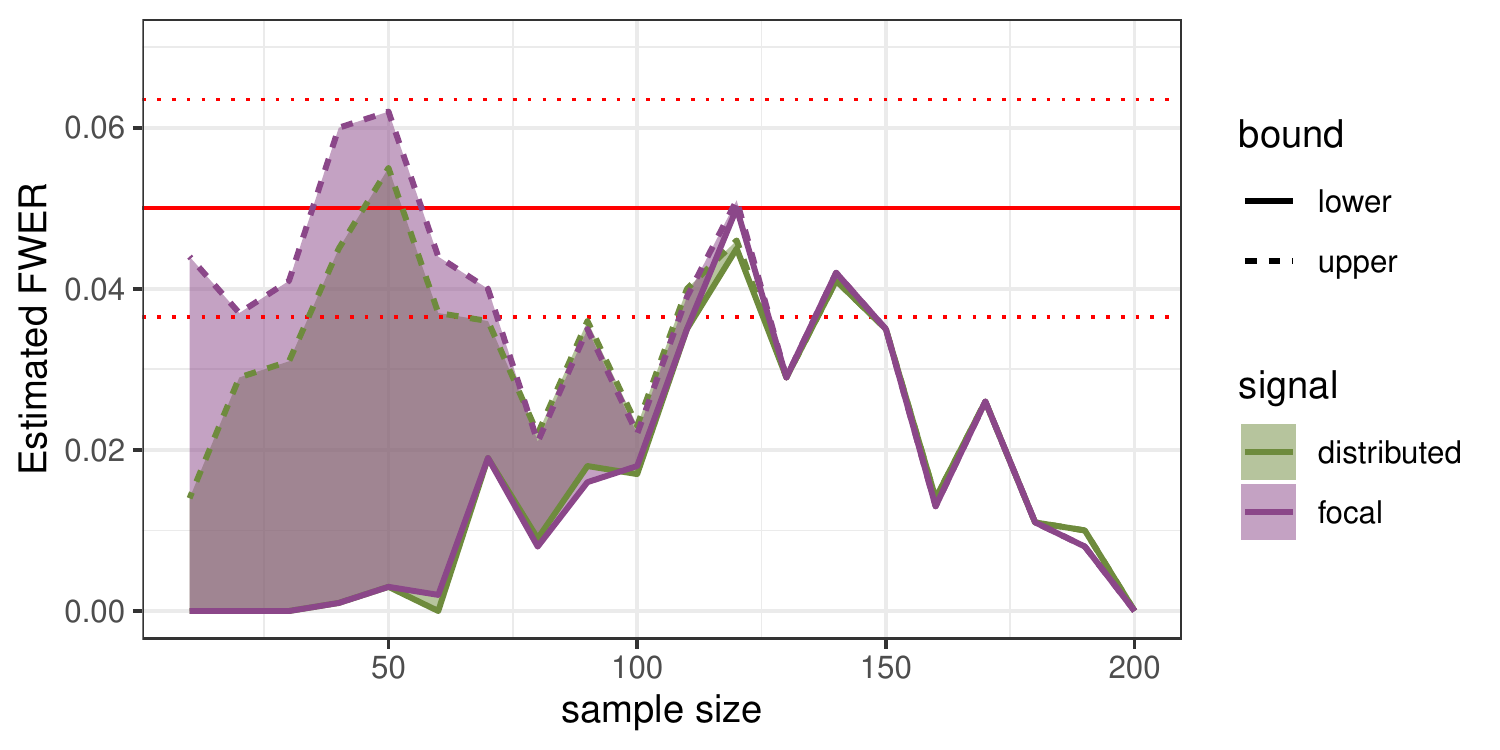}
\par\medskip
\centering
\caption{Estimated family-wise error rates (FWER) for focal (purple) and distributed (green) signals with the amplitudes of $d=0.1$. The red dotted horizontal lines represent the binomial confidence intervals for the FWER at $\alpha=0.05$ (solid horizontal line). Shown are the results for lower-bound (solid line) and upper-bound (dashed line) based on the heuristic algorithm. The results for $d=0.05$ (not shown) are almost identical, since the same realization of the noise field was used for both simulations.}
\label{fig: FWER}
\end{figure}

\begin{figure}
\centering
\subfloat[high signal amplitude of $d=0.1$.]{\label{fig_tdp_01}\includegraphics[width=0.9\textwidth]{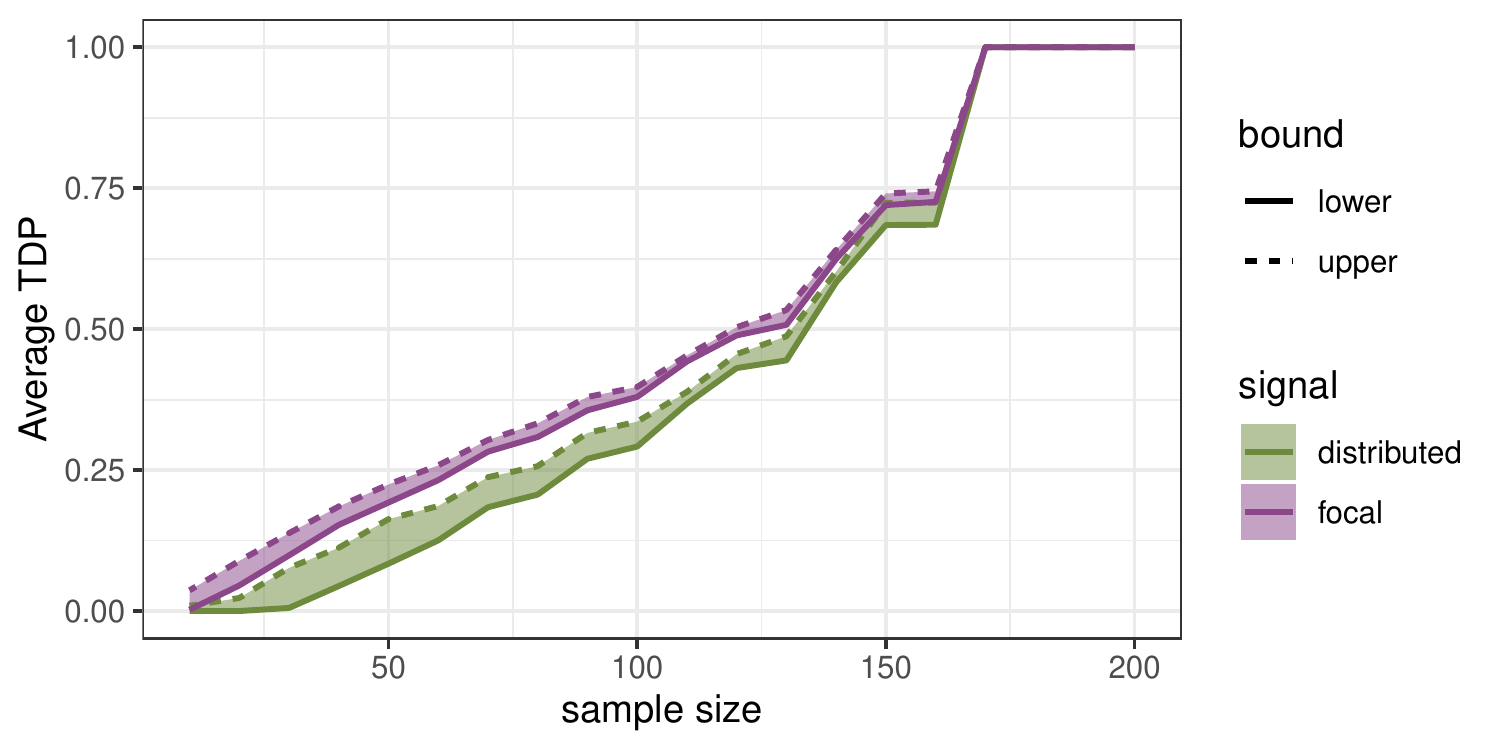}}
\par\medskip
\centering
\subfloat[low signal amplitude of $d=0.05$.]{\label{fig_tdp_005}\includegraphics[width=0.9\textwidth]{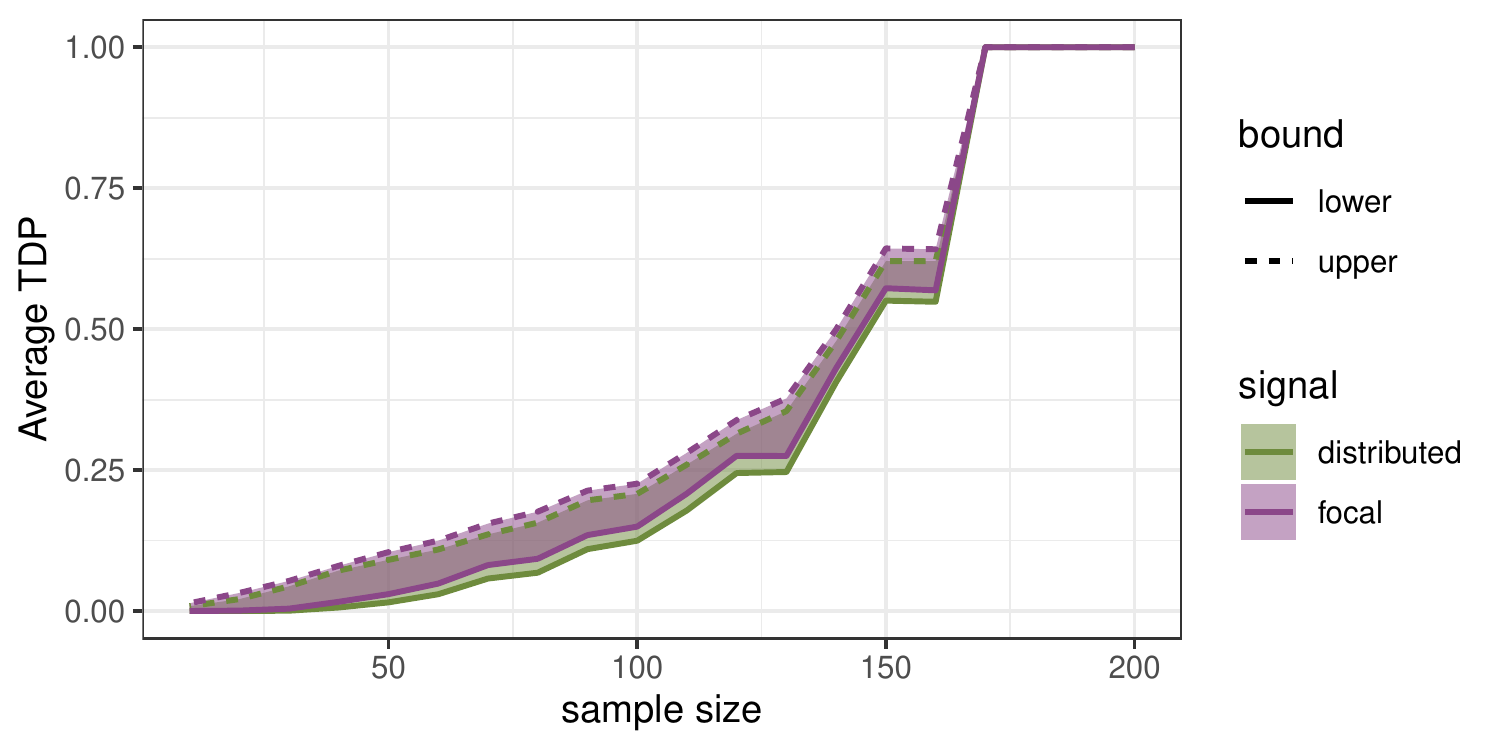}}
\caption{Average TDP bounds for all significant clusters for focal (purple) and distributed (green) signals with the amplitudes of $d=0.1$ and $d=0.05$. Shown are the results for lower-bound (solid line) and upper-bound (dashed line) based on the heuristic algorithm.}
\label{fig: TDP}
\end{figure}

\section{Application: Human Connectome Project $n$-back task revisited} \label{sec appl}

We illustrate the use of the new method using a more extensive analysis of the data set introduced in Section \ref{sec motivate}.

A $z$-score threshold $z$ and cluster extent threshold $k_M$ can be defined in any way that satisfies \eqref{eq CI weak FWER}; that is, fixing one threshold, the smallest value of the other still satisfying \eqref{eq CI weak FWER} can be calculated. We present the permutation-based thresholds in this Section, using the fast algorithm for finding $z$ as a function of $k_M$ using permutations given in the Supplemental Information, Section E. For comparison, the analysis with thresholds based on random field theory is given in the Supplemental Information, Section F.

We present two alternative permutation-based analyses. First, we fixed $z=3.1$, which corresponds to $k_M=72$ in this data (Table \ref{tbl: perm_k72}). Next, we fixed $k_M = 14$ and calculated the corresponding $z$-threshold $z=3.7$ (Table \ref{tbl: perm_k14}). Non-significant supra-threshold clusters were not displayed. The TDP bounds for relevant overlapping anatomical regions are also displayed. 

TDP was calculated both using heuristic algorithms and using the lower bound of Theorem \ref{thm final}. Our heuristic algorithms were run for several hours on a cluster to produce these results, and we believe that these results are sufficiently close to the true minimum. Shorter running times of 20--60 seconds would give TDP results up to only 5\% higher than the reported values. Comparing the heuristic results and the lower bound, the lower bound was closest to the heuristic solution for large clusters and small $k_M$, as expected from the theory.

Comparing the $z=3.1$ and $k_M=14$ settings, the results clearly show a trade-off between detection and TDP. The lower cluster extent threshold $k_M$, that corresponds to a higher $z$-threshold, returns smaller clusters with larger TDP, while the high $k_M$ results in larger clusters with smaller TDP. For anatomical regions it is not a priori clear whether larger TDP would be found with high or low values of $k_M$. In this data set, increased TDP bounds were perceived when $k_M$ was small, i.e.\ when the $z$-threshold was large. Corresponding anatomical regions of the clusters were identified using the Harvard-Oxford cortical structural atlas and MNI structural atlas as available in FSL \citep{Jenkinson2012}.

\begin{table}[!ht]
\caption{Results for supra-threshold clusters, defined by the cluster-forming $z$-threshold of $Z > 3.1$ and the resulting minimal cluster extent threshold $k_M = 72$ based on permutation. The results from the heuristic algorithms are indicated by TDP, the lower bound of Theorem \ref{thm final} by LB. \label{tbl: perm_k72}}
\begin{tabular}{@{\extracolsep{\fill}} c rcc r rrcc cccc}
\toprule
\multicolumn{4}{c}{Cluster} & \multicolumn{5}{c}{Anatomical region} & \multicolumn{3}{c}{Location} & \\
\cmidrule(r){1-4} \cmidrule(lr){5-9} \cmidrule(l){10-12}
ID & size & TDP & LB & Region & size & overlap & TDP & LB & $x$ & $y$ & $z$ & $Z_\text{max}$ \\
\midrule
1 & 8870 & 0.368 & 0.265 & MFG & 18250 & 4049 & 0.082 & 0.061 & 44 & 72 & 60 & 8.87 \\
& & & & FP & 33571 & 2021 & 0.020 & 0.013 & & & & \\
& & & & IC & 6591 & 564 & 0.025 & 0.016 & & & & \\
2 & 8526 & 0.402 & 0.307 & sLOC & 27121 & 5142 & 0.069 & 0.049 & 19 & 42 & 61 & 9.51 \\
& & & & AG & 13689 & 4260 & 0.117 & 0.089 & & & & \\
& & & & pSMG & 14829 & 3804 & 0.097 & 0.074 & & & & \\
& & & & Precuneous & 18119 & 2491 & 0.051 & 0.037 & & & & \\
3 & 7956 & 0.332 & 0.201 & Cerebellum & 39724 & 6551 & 0.057 & 0.037 & 63 & 33 & 20 & 9.20 \\
4 & 6652 & 0.372 & 0.265 & MFG & 18250 & 4035 & 0.083 & 0.061 & 31 & 67 & 64 & 9.73 \\
& & & & FP & 33571 & 2587 & 0.026 & 0.018 & & & & \\
& & & & IC & 6591 & 589 & 0.026 & 0.017 & & & & \\
5 & 350 & 0.191 & 0.037 & pMTG & 11420 & 310 & 0.006 & 0.001 & 15 & 46 & 28 & 5.18 \\
& & & & tMTG & 9735 & 271 & 0.005 & 0.000 & & & & \\
6 & 100 & 0.140 & 0.010 & Cerebellum & 39724 & 100 & 0.000 & 0.000 & 49 & 35 & 10 & 6.56 \\
\cmidrule(lr){1-13}
Total & 32454 & 0.367 & 0.257 & MFG & 18250 & 8084 & 0.165 & 0.122 & & & & \\
& & & & Cerebellum & 39724 & 6651 & 0.058 & 0.037 & & & & \\
& & & & sLOC & 27121 & 5142 & 0.069 & 0.049 & & & & \\
& & & & FP & 33571 & 4608 & 0.046 & 0.031 & & & & \\
& & & & AG & 13689 & 4260 & 0.117 & 0.089 & & & & \\
& & & & pSMG & 14829 & 3804 & 0.097 & 0.074 & & & & \\
& & & & Precuneous & 18119 & 2491 & 0.051 & 0.037 & & & & \\
& & & & IC & 6591 & 1153 & 0.051 & 0.033 & & & & \\
& & & & pMTG & 11420 & 310 & 0.006 & 0.001 & & & & \\
& & & & tMTG & 9735 & 271 & 0.005 & 0.000 & & & & \\
\bottomrule
\end{tabular}
\end{table}

\begin{table}[!ht]
\caption{Results for supra-threshold clusters, defined by cluster extent threshold $k_M = 14$ and the resulting cluster-forming $z$-threshold of $Z > 3.7$, based on permutation. The results from the heuristic algorithms are indicated by TDP, the lower bound of Theorem \ref{thm final} by LB.}  \label{tbl: perm_k14}
\begin{tabular}{@{\extracolsep{\fill}} c rcc r rrcc cccc}
\toprule
\multicolumn{4}{c}{Cluster} & \multicolumn{5}{c}{Anatomical region} & \multicolumn{3}{c}{Location} & \\
\cmidrule(r){1-4} \cmidrule(lr){5-9} \cmidrule(l){10-12}
ID & size & TDP & LB & Region & size & overlap & TDP & LB & $x$ & $y$ & $z$ & $Z_\text{max}$ \\
\midrule
1 & 7231 & 0.606 & 0.532 & sLOC & 27121 & 4293 & 0.091 & 0.078 & 19 & 42 & 61 & 9.51 \\
& & & & Precuneous & 18119 & 2123 & 0.067 & 0.058 & & & & \\
2 & 6899 & 0.577 & 0.487 & MFG & 18250 & 3224 & 0.102 & 0.087 & 44 & 72 & 60 & 8.87 \\
& & & & SFG & 18946 & 2880 & 0.085 & 0.073 & & & & \\
& & & & PCG & 9245 & 1558 & 0.096 & 0.084 & & & & \\
& & & & IC & 6591 & 494 & 0.040 & 0.034 & & & & \\
3 & 5345 & 0.546 & 0.438 & Cerebellum & 39724 & 4840 & 0.067 & 0.054 & 63 & 33 & 20 & 9.20 \\
4 & 5143 & 0.575 & 0.487 & MFG & 18250 & 3285 & 0.104 & 0.089 & 31 & 67 & 64 & 9.73 \\
& & & & FP & 33571 & 1893 & 0.031 & 0.026 & & & & \\
& & & & SFG & 18946 & 1745 & 0.052 & 0.043 & & & & \\
5 & 202	 & 0.391 & 0.158 & OP & 15486 & 156 & 0.004 & 0.001 & 39 & 22 & 36 & 5.72 \\
& & & & ICC & 7134 & 110 & 0.006 & 0.003 & & & & \\
6 & 128	 & 0.375 & 0.148 & pMTG & 11420 & 128 & 0.004 & 0.002 & 15 & 46 & 28 & 5.18 \\
7 & 66 & 0.379 & 0.182 & Cerebellum & 39724 & 66 & 0.001 & 0.000 & 49 & 35 & 10 & 6.56 \\
8 & 61 & 0.361 & 0.115 & FP & 33571 & 61 & 0.001 & 0.000 & 31 & 86 & 29 & 5.77 \\
9 & 56 & 0.321 & 0.143 & FP & 33571 & 56 & 0.001 & 0.000 & 57 & 88 & 29 & 5.16 \\
10 & 39 & 0.308 & 0.103 & OP & 15486 & 39 & 0.001 & 0.000 & 51 & 15 & 42 & 5.35 \\
11 & 22 & 0.182 & 0.045 & Thalamus & 4602 & 17 & 0.000 & 0.000 & 43 & 53 & 43 & 4.55 \\
12 & 21 & 0.095 & 0.048 & Cerebellum & 39724 & 21 & 0.000 & 0.000 & 42 & 36 & 10 & 4.85 \\
\cmidrule(lr){1-13}
Total & 25213 & 0.573 & 0.482 & MFG & 18250 & 6509 & 0.206 & 0.176 & & & & \\
& & & & Cerebellum & 39724 & 4927 & 0.068 & 0.054 & & & & \\
& & & & SFG & 18946 & 4625 & 0.137 & 0.116 & & & & \\
& & & & sLOC & 27121 & 4293 & 0.091 & 0.078 & & & & \\
& & & & Precuneous & 18119 & 2123 & 0.067 & 0.058 & & & & \\
& & & & FP & 33571 & 2010 & 0.032 & 0.026 & & & & \\
& & & & PCG & 9245 & 1558 & 0.096 & 0.084 & & & & \\
& & & & IC & 6591 & 494 & 0.040 & 0.034 & & & & \\
& & & & OP & 15486 & 195 & 0.004 & 0.001 & & & & \\
& & & & pMTG & 11420 & 128 & 0.004 & 0.002 & & & & \\
& & & & ICC & 7134 & 110 & 0.006 & 0.003 & & & & \\
& & & & Thalamus & 4602 & 17 & 0.000 & 0.000 & & & & \\
\bottomrule
\end{tabular}
\end{table}

\section{Application: Neurovault} \label{sec neurovault}
Next, we applied the new algorithm to a selection of 818 datasets from the Neurovault database (neurovault.org; \citet{Gorgolewski2015}). The Neurovault database consists of unthresholded maps from neuroimaging studies. We selected 818 representative functional MRI datasets containing group-level statistics maps. 
For the calculation of clusters we used two settings: a standard $z$-threshold of $z=3.1$, and a $k$-threshold of $k_M=14$. The corresponding $k_M$- and $z$-thresholds, respectively, were estimated using Gaussian Random Field Theory \citep{Forman1995}. As residual data were unavailable, we estimated smoothness of the random field on the $z$-statistics image. The $z=3.1$ setting produced values of $k_M$ ranging from 71 to 507 (1st and 9th decile). Details of the selected images and estimation procedures can be found in the Supplemental Information, Section D. 

For each dataset we estimated the TDP of each supra-threshold cluster obtained using $z=3.1$ and $k_M=14$. We then calculated for each TDP value how many supra-threshold voxels with at least that TDP were significant on average across all datasets. This allows us to visualize the relationship between the size of the clusters detected and the TDP of those clusters for different methods. We plot the theoretical lower-bound of both methods (according to Theorem \ref{thm final}), and the solution as estimated using the heuristic methods. For reference we also calculated the number of voxels above the Gaussian random field voxelwise threshold (equivalent to a $k_M=0$ setting). 

\begin{figure}[!ht]
\includegraphics[width=\textwidth]{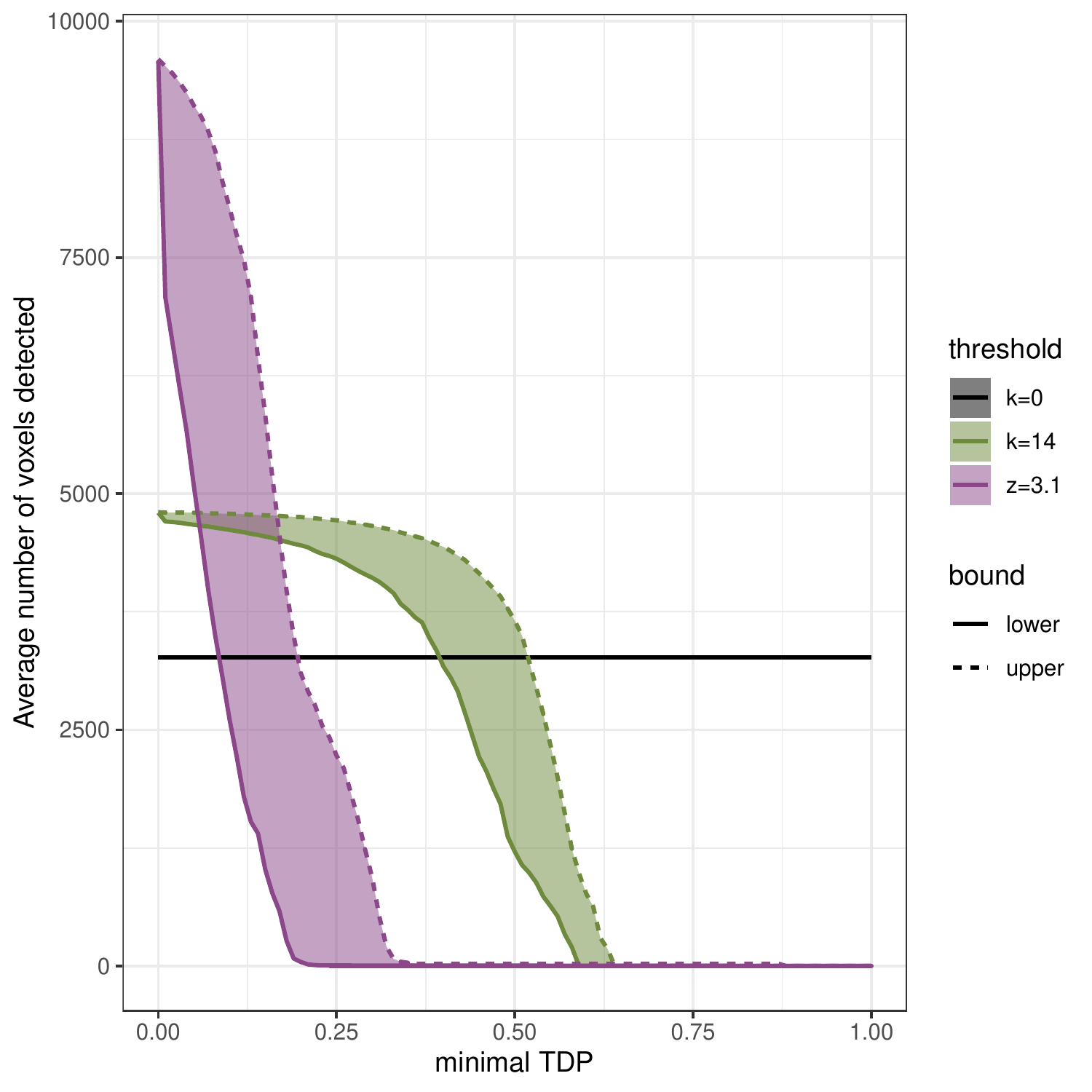}
\caption{Relation between TDP (x-axis) and the average number of voxels detected for three thresholds: $z=3.1$ (purple), $k_M=14$ (green), and $k_M=0$ (black). The results for lower-bound (solid line) and upper-bound (dashed line) are shown based on the  heuristic algorithm.} \label{fig tdp vs voxels} 
\end{figure}

Figure \ref{fig tdp vs voxels} shows the results of the analysis across all datasets. As can be seen the $z=3.1$ setting (purple) leads to larger cluster sizes but with low TDP's. For $k_M=14$ (green), the size of the clusters with low TDP's is smaller, but there are more clusters with a more reasonable (albeit still relatively small) TDP. Both methods detect larger regions than voxelwise inference ($k_M=0$, black line) at low TDP thresholds, but smaller regions at high TDP. The figure shows a clear trade-off between detection and TDP: at low $k_M$ settings, small regions are detected with large TDP; with high $k_M$, larger regions are detected, but TDP is (much) lower.

We note that the estimation of the smoothness using the $z$-statistics rather than the residuals tends to overestimate the smoothness if there is much signal. As a result, it is likely that we have overestimated values of $k_M$ when $z=3.1$ and overestimated $z$ when $k_M=14$. The TDP results in Figure \ref{fig tdp vs voxels} are therefore likely an underestimate of what would be found if the full datasets would have been available.

\section{Discussion}

We have presented a uniform improvement of classic cluster inference that allows much more meaningful and informative inference to be obtained from that method. In the first place, the new method allows inference on anatomical regions of interest and data-driven supra-threshold clusters within the same analysis. Moreover, regions of interest do not have to be specified before seeing the data. Secondly, rather than (only) a $p$-value, the new method provides a true discovery proportion (TDP) for every brain region. Quantifying the spatial extent of activation within the brain region, the TDP is much more informative than the $p$-value, which only quantifies the evidence for the presence of any signal at all. TDP is also less prone to overinterpretation than the $p$-value. In the Neurovault analysis we have found many examples of brain regions with a seemingly impressive $p < 0.001$ that had unremarkable TDPs of 20\% or less. We recommend that TDP is always reported with (or even instead of) the $p$-value in fMRI cluster inference.


Despite making these additional inferences, error control remains as strict as with classic cluster inference: with probability at least $1-\alpha$ no regions get an estimated TDP that is larger than the true value. To guarantee this error control, the method does not require any additional model assumptions. It can assume either that the $z$-scores of inactive voxels follow a Gaussian random field or that they are invariant under permutations. 

Inference on brain regions in terms of TDP can be said to solve the Spatial Specificity Paradox \citep{Woo2014}, but by doing so it makes the same paradox painfully visible. At the usual setting with a cluster-forming threshold of $z=3.1$ most significant brain regions have a TDP less than 20--30\%. Our analyses have made it clear that there is a trade-off involved in choosing the cluster-forming threshold. Low thresholds result in many large clusters but with low TDP; higher thresholds have less detection power but much higher TDP. In the extreme, voxel-wise inference was shown to be a special case of cluster extent inference that always returns a TDP of 100\%. In order to obtain TDP substantially over a reasonably minimal threshold of 50\%, we recommend cluster thresholding with $k_M=14$ or less, resulting in much larger $z$-thresholds than usually recommended in the field \citep{Eklund2016}.  

Computationally, the calculation of the TDP involves solving a $k$-separator problem. We presented two solutions to this problem: the lower bound retains the error control guarantee but is conservative; the heuristic solution is more accurate, but at the cost of losing error control if the method does not fully converge. Together, the two algorithms can be used to bracket the TDP lower confidence bound. We recommend the heuristic solution in practice provided enough computing power is available.

Inference for neuroimaging in terms of TDP rather than $p$-values has been proposed by several authors \citep{Rosenblatt2018, Blanchard2020, Andreella2020, Vesely2021}. None of the proposed methods is expected to outperform any of the others uniformly \citep{Goeman2020}. A systematic and careful inventory should be performed to find out when to prefer which TDP methods with which tuning parameters. This large project is beyond the scope of this paper. In such a comparison, the method proposed in this paper will serve as an important benchmark, representing classic cluster analysis, which it is designed to be consistent with.

\section*{Acknowledgements}

Data were provided in part by the Human Connectome Project, WU-Minn Consortium (Principal Investigators: David Van Essen and Kamil Ugurbil; 1U54MH091657) funded by the 16 NIH Institutes and Centers that support the NIH Blueprint for Neuroscience Research; and by the McDonnell Center for Systems Neuroscience at Washington University. This research was supported by Nederlandse Organisatie voor Wetenschappelijk Onderzoek, Grant Number: 639.072.412.

\appendix

\section{Proofs of Theorems and Lemmas}

\subsection{Proof of Lemma \ref{lem use h}}

\begin{replemma}{lem use h}
For every $V \subseteq M$, we have $\underline{\boldsymbol\psi}_V \leq{\boldsymbol\psi}_V$.
\end{replemma}

\begin{proof}
\begin{eqnarray*}
{\boldsymbol\psi}_V &=& \min\{{\boldsymbol\phi}_W\colon V \subseteq W \subseteq M\} \\
&=& \min\big\{\mathds{1}\{\chi_{V \cap \mathbf{Z}} > k_V\}\colon V \subseteq W \subseteq M\big\} \\
&\geq& \mathds{1}\big\{\min\{\chi_{V \cap \mathbf{Z}}\colon V \subseteq W \subseteq M\}
> \max\{k_V\colon V \subseteq W \subseteq M\}\big\} \\
&=& \mathds{1}\{\chi_{V \cap \mathbf{Z}} > k_M\}  \\
&=& \underline{\boldsymbol\psi}_V.
\end{eqnarray*}
\end{proof}

\subsection{Proof of Theorem \ref{thm first TDP bound}}

\begin{reptheorem}{thm first TDP bound}
Let $\check{\mathbf{a}}(V)= s_{k_M}(V \cap \mathbf{Z}),$
where $s_k(V) = \min\{|R|\colon \chi_{V\setminus R} \leq k\}$. Then, for all $\mathrm{P} \in\Omega$,
\[
\mathrm{P}(\textrm{$\check {\mathbf{a}}(V) \leq a_\mathrm{P}(V)$ for all $V \subseteq M$}) \geq 1-\alpha.   
\]
\end{reptheorem}

\begin{proof}
The minimum in (\ref{def aV}) is achieved when $R \cap \mathbf{Z} = \emptyset$, so we can rewrite
\[
\check{\mathbf{a}}(V) = \min\{|R|\colon R\subseteq V, \chi_{\mathbf{Z} \cap (V\setminus R)} \leq k_M\}. 
\] 
Setting $W = V\setminus R$, we obtain
\begin{eqnarray*}
\check{\mathbf{a}}(V) &=& \min\{|V\setminus W|\colon W \subseteq V, \chi_{\mathbf{Z} \cap W} \leq k_M\}\\
&=& \min\{|V\setminus W|\colon W \subseteq V, \underline{\boldsymbol{\psi}}_{W} = 0\} \\
&\leq& \min\{|V\setminus W|\colon W \subseteq V, \boldsymbol{\psi}_{W} = 0\} \\
&=& \mathbf{a}(V), 
\end{eqnarray*}
where the inequality uses Lemma \ref{lem use h}. The result now follows from (\ref{eq simultaneous tilde}). 
\end{proof}

\subsection{Proof of Theorem \ref{thm uniform improvement}}

\begin{reptheorem}{thm uniform improvement}
If $\mathbf{C} \subseteq (\mathbf{Z} \cap M)$, with $|\mathbf{C}|>k_M$, is a cluster, then $\check{\mathbf{a}}(\mathbf{C})>0$.
\end{reptheorem}

\begin{proof}
For $R = \emptyset$, we have 
\[
\chi_{(\mathbf{Z} \cap \mathbf{V})\setminus R} = \chi_{\mathbf{Z} \cap \mathbf{V}} = \chi_{\mathbf{V}} = |\mathbf{V}| > k_M,
\]
so the minimum in (\ref{def aV}) is attained when $|R|>0$.
\end{proof}

\subsection{Proof of Lemma \ref{lem decompose clusters}}

We first state, perhaps superfluously, that smaller voxel sets contain smaller clusters. 

\begin{lemma} \label{lem monotone chi}
If $V \subseteq W$, then $\chi_V \leq \chi_W$.
\end{lemma}

\begin{proof}
Let $C \subseteq V$ be the largest cluster in $V$, then $C \subseteq V \subseteq W$ is a cluster in $W$. We have $$\chi_V = |C| \leq \chi_W.$$ 
\end{proof}

\begin{replemma}{lem decompose clusters}
If $V = C_1 \cup \cdots \cup C_n$, where $C_1\ldots, C_n$ are disconnected clusters, then
\[
s_k(V) = \sum_{i=1}^n s_k(C_i).
\]
\end{replemma}

\begin{proof}
Suppose $R$ is an optimal $k$-separator of $V$, so that $\chi_{V\setminus R} \leq k$ and $s_k(V) = |R|$. For $i=1\ldots,n$, define $R_i = R \cap C_i$. Since $C_i\setminus R_i \subseteq V \setminus R$, we have, by Lemma \ref{lem monotone chi},
\[
\chi_{V_i\setminus R_i} \leq \chi_{V\setminus R} \leq k,
\]
so $R_i$ separates $C_i$. Since $R_1, \ldots, R_n$ are disjoint, we have
\[
s_k(V) = |R| = \sum_{i=1}^n |R_i| \geq \sum_{i=1}^n s_k(C_i).
\]

Vice versa, suppose for $i=1,\ldots,n$ that $R_i$ is an optimal $k$-separator of $C_i$, so that $\chi_{C_i\setminus R_i} \leq k$ and $s_k(C_i) = |R_i|$. Define $R = R_1 \cup \cdots \cup R_n$. Let $C$ be any cluster in $V \setminus R$. Since $C_1, \ldots, C_n$ are disconnected, we must have $C \subseteq C_i$ for some $1\leq i \leq n$. We have
\[
\chi_{V\setminus R} \leq \max_{1\leq i \leq n} \chi_{C_i\setminus R_i} \leq k,
\]
so $R$ separates $V$. Since $R_1, \ldots, R_n$ are disjoint, we have
\[
\sum_{i=1}^n s_k(C_i) = \sum_{i=1}^n |R_i| = |R| \geq s_k(V).
\]
\end{proof}

\subsection{Proof of Lemma \ref{lem T to R}}

We start with a lemma that is the reason the positive neighbors definition is so useful: two voxels are neighbors if and only if they have a positive neighbor in common.

\begin{lemma} \label{lem positive neighbors}
Voxels $v,w \in \mathbb{Z}^d$ are neighbors if and only if $\{v\}^+ \cap \{w\}^+ \neq \emptyset$.
\end{lemma}

\begin{proof}
Suppose $v$ and $w$ are neighbors. Consider $u = w + (v-w)_+$. Note that $x + (-x)_+ = x_+$. Then
\[
u_i - v_i = (w_i - v_i) + (v_i-w_i)_+ = (w_i-v_i)_+ \in \{0,1\},
\]
so $u \in \{v\}^+$, and $u_i - w_i = (v_i-w_i)_+ \in \{0,1\},$ so $u \in \{w\}^+$. Therefore $u\in \{v\}^+ \cap \{w\}^+ \neq \emptyset$. 

Next, suppose that $u \in \{v\}^+ \cap \{w\}^+$. Then $u = v+e$ and $u=w+h$, with $e,h \in \{0,1\}^d$. We have $v_i-w_i = h_i-e_i \in \{-1,0,1\},$
so $v$ and $w$ are neighbors.
\end{proof}

The next lemma translates the previous lemma to voxel sets: two voxel sets are disconnected if and only if their cover is disjoint.

\begin{lemma} \label{lem disjoint disconnected}
Voxel sets $V,W \subseteq \mathbb{Z}^d$ are disconnected if and only if $V^+$ and $W^+$ are disjoint.
\end{lemma}

\begin{proof}
Suppose $V$ and $W$ are not separated. Then there exist $v \in V$ and $w\in W$ that are neighbors. By Lemma \ref{lem positive neighbors} there exists $u \in \{v\}^+ \cap \{w\}^+$. By definition of $V^+$ and $W^+$, we have $u \in V^+ \cap W^+$, so $V^+$ and $W^+$ are not disjoint. 

Suppose $V^+$ and $W^+$ are not disjoint. Then $u \in V^+ \cap W^+$ exists. By definition of of $V^+$ and $W^+$ there exist $v \in V$ and $w \in W$ such that $u \in \{v\}^+\cap \{w\}^+$. By Lemma \ref{lem positive neighbors}, $v$ and $w$ are neighbors, so $V$ and $W$ are not disconnected. 
\end{proof}

The cover and interior operations are not each other's inverse: the cover of the interior may be a smaller voxel set.

\begin{lemma} \label{lem -+}
$(V^-)^+ \subseteq V$.
\end{lemma}

\begin{proof}
Choose $v \in (V^-)^+$. By definition of the cover there must be a $w \in V^-$ such that $v \in \{w\}^+$. By definition of the interior, every positive neighbor of every $w \in V^-$ is in $V$. Therefore $v \in V$.
\end{proof}

Now we come to prof of the lemma itself.

\begin{replemma}{lem T to R}
For every tiling $T_1,\ldots,T_n$ of $V^+$ there exists a $k$-separator $R$ of $V$ such that
\[
|R| = t_k(T_1,\ldots,T_n).
\]
\end{replemma}

\begin{proof}
For $i=1,\ldots,n$, let $R_i = T_i^0 \cap V$ and let $R_i'$ be any subset of $T_i^- \cap V$ with $|R_i'| = (|T_i^- \cap V|-k)_+$. Then $R = R_1 \cup \cdots \cup R_n \cup
R_1' \cup \cdots \cup R_n'$ has $|R| = t_k(T_1,\ldots, T_n)$. 

We show that $R$ is a $k$-separator of $V$. For $i=1,\ldots,n$, let $C_i = (T_i^- \cap V) \setminus R_i'$. Then $C_i \cap R = \emptyset$ and $V = R \cup C_1 \cup \cdots \cup C_n$. Moreover, $|C_i| \leq k$ by definition of $R_i'$.
Since $C_i \subseteq T_i^-$, we have $C_i^+ \subseteq T_i$ by Lemma \ref{lem -+}. Therefore $C_1^+, \ldots, C_n^+$ are disjoint. By Lemma \ref{lem disjoint disconnected}, $C_1, \ldots, C_n$ are disconnected. It follows that 
\[
\chi_{V\setminus R} \leq \max_{1\leq i\leq n} |C_i| \leq k,
\]
so $R$ is a $k$-separator of $V$.
\end{proof}

\subsection{Proof of Lemma \ref{lem R to T}}

\begin{replemma}{lem R to T}
For every $k$-separator $R$ of $V$ there exists a tiling $T_1,\ldots,T_n$ of $V^+$ such that $T_1, \ldots, T_n$ are clusters, and
\[
|R| \geq t_k(T_1,\ldots,T_n).
\]
\end{replemma}

\begin{proof}
We write $V \setminus R = C_1 \cup \cdots \cup C_n$ with $C_1, \ldots, C_n$ non-empty disconnected clusters. Since $R$ is a $k$-separator of $V$, we have $|C_i|\leq k$ for $i=1,\ldots,n$. For $i=1,\ldots,n$ call $T_i = C_i^+ \subseteq V^+$. These are clusters since $C_1,\ldots,C_n$ are. Write 
\[
V^+ \setminus (T_1 \cup \cdots \cup T_n) = T_{n+1} \cup \cdots \cup T_{n+m},
\]
with $T_{n+1},\ldots,T_{n+m}$ disjoint clusters, so that $V^+ = T_1 \cup \cdots \cup T_{n+m}$. Call $C_{n+1} = \ldots = C_{n+m} = \emptyset$.

We will show that $t_k(T_1, \ldots, T_{n+m}) \leq |R|$. Since $C_i \subset T_i$ for all $1\leq i\leq n+m$, $v \in T_i^- \cap V \subseteq T_i$ cannot be in $C_j$ for $j\neq i$, so $T_i^- \cap V \subseteq C_i \cup R$. Since $C_i$ and $R$ are disjoint, $(T_i^- \cap V) \setminus R = C_i$. Therefore,  
\[
|T_i^- \cap R| \geq |(T_i^- \cap V) \cap R| = |T_i^- \cap V| - |C_i| \geq (|T_i^- \cap V| - k)_+.
\] 
If $v \in C_i$, then $\{v\}^+ \subseteq C_i^+ = T_i$, so $v \notin T_i^0$. It follows that $T_i^0 \cap V = T_i^0 \cap R$ for $i=1,\ldots, n$. The same holds by definition for $i=n+1,\ldots, n+m$. We have 
\begin{eqnarray*}
|R| &=& \sum_{i=1}^{n+m} |T_i^0 \cap R| + \sum_{i=1}^{n+m} |T_i^- \cap R| \\
&\geq& \sum_{i=1}^{n+m} |T_i^0 \cap V|  + \sum_{i=1}^{n+m} (|T_i^- \cap V| - k)_+ \\
&=& t_k(T_1, \ldots, T_{n+m}).
\end{eqnarray*}
\end{proof}

\subsection{Proof of Theorem \ref{thm tiling}}

\begin{reptheorem}{thm tiling}
We have
\[
s_k(V) = \min\{t_k(T_1,\ldots,T_n)\colon \textrm{\ $T_1,\ldots,T_n$  is a tiling of $V^+$}\}.
\]
The minimum is attained for a tiling for which $T_1,\ldots,T_n$ are all clusters.
\end{reptheorem}

\begin{proof}
Suppose $R$ is a $k$-separator of $V$ with $|R| = s_k(V)$. Then by Lemma \ref{lem R to T} a tiling $T_1,\ldots,T_n$ exists such that $t_k(T_1,\ldots, T_n) \leq |R| = s_k(V)$. It follows that
\[
s_k(V) \geq \min\{t_k(T_1,\ldots,T_n)\colon \textrm{\ $T_1,\ldots,T_n$  is a tiling of $V^+$}\}.
\]

Now suppose that $T_1, \ldots, T_n$ minimizes $t_k(T_1,\ldots, T_n)$. By Lemma \ref{lem T to R}, a $k$-separator exists such that $|R| = t_k(T_1,\ldots, T_n)$. It follows that 
\[
s_k(V) \leq \min\{t_k(T_1,\ldots,T_n)\colon \textrm{\ $T_1,\ldots,T_n$  is a tiling of $V^+$}\}.
\]
Combining the two inequalities, the result follows. By Lemma \ref{lem R to T} the optimal tiling can be taken as one that consists of clusters.
\end{proof}

\subsection{Proof of Theorem \ref{thm shortcut 2}}

Though Theorem \ref{thm tiling} allows tilings with a tile interior $>k$, we can always find an alternative solution that does no have such tiles.

\begin{lemma} \label{lem tidy}
Suppose $T_1, \ldots, T_n$ is a tiling of $V^+$. Then there exists a tiling $T_1', \ldots, T_{n'}'$ with $t_k(T_1, \ldots, T_n) = t_k(T_1', \ldots, T_{n'}')$ and $|T^-_i| \leq k$ for $i=1,\ldots, n'$.  
\end{lemma}

\begin{proof}
Choose any $i$ such that $|T^-_i| > k$. We will construct $T'_i$ and $T'_{n+1}$ such that $T'_i \cup T'_{n+1} = T_i$, and $|(T'_i)^-| = |T^-_i|-1$ and $(T'_{n+1})^-= 0$. Repeatedly applying this construction for all tiles with $|T_i^-| > k$ will give us the tiling with the desired property since the newly constructed tiling has the same value of $t_k$, but the interior of tile $T_i$ is reduced in size by 1.

Choose any $v \in T_i^-$ that minimizes $\sum_{i=1}^d v_i$. Define $T'_i = T_i \setminus \{v\}$. Then all negative neighbors of $v$ are not in $(T'_i)^-$, so $$(T'_i)^- = T_i^- \setminus \{w \in T^+_i\colon v \in \{w\}^+\}.$$ Let $w \in \mathbb{Z}^d$ such that $v \in \{w^+\}$. Then either $w=v$ or $v=w+e$ with $e \in \{0,1\}^d$ and $\sum_{i=1}^d e_i \geq 1$. If $w \neq v$, then
$$
\sum_{i=1}^d w_i = \sum_{i=1}^d v_i -\sum_{i=1}^d e_i \leq \sum_{i=1}^d v_i -1,
$$
so $w \notin T_i^-$ since $v$ minimized $\sum_{i=1}^d v_i$ among $v\in T_i^-$. Therefore $\{w \in T^+_i\colon v \in \{w^\}+\} = \{v\}$ and $(T'_i)^- = T_i^- \setminus \{v\}$. So $|(T_i')^-| = |T_i^-| -1$. Defining $T_{n+1}' = \{v\}$, we have $(T_{n+1}')^- = 0$. This gives the required construction.
\end{proof}

If the interior of the cover of a voxel set $V$ is no larger than the original object, then the interior of all its tiles is in $V$.

\begin{lemma} \label{lem no residue}
If $(V^+)^- \subseteq V$, then for every $T \subseteq V^+$, we have $T^- \subseteq V$.
\end{lemma}

\begin{proof}
Since $T \subseteq V^+$, we have $T^- \subseteq (V^+)^- \subseteq V$.
\end{proof}  

The next lemma is a special case of a property that holds in general for closed testing procedures \citep[][Lemma 3]{Goeman2020}. We prove it in context here. 

\begin{lemma} \label{lem coherent}
If $V,W \subseteq \mathbb{Z}^d$ are disjoint, then $s_k(V \cup W) \leq s_k(V) + |W|$.
\end{lemma}

\begin{proof}
Let $R$ be a $k$-separator of $V$ such that $|R| = s_k(V)$. Consider $R' = R \cup W$ and let $C$ any cluster in $(V\cup W) \setminus R'$. Since $(V\cup W) \setminus R' = V \setminus R$, $C$ is also a cluster in $V\setminus R$, so $|C| \leq k$, because $R$ is a $k$-separator of $V$. Therefore, $R'$ is a $k$-separator of $V \cup W$, and we have
\[
s_k(V \cup W) \leq |R'| = |R| + |W| = s_k(V) + |W|.
\] 
\end{proof}

\begin{reptheorem}{thm shortcut 2}
\[
s_k(V) \geq  r_k\cdot |V^+|-|V^+\setminus V|,
\]
where
\[
r_k = \min\{|V^0|/|V|\colon \emptyset \neq V \subset \mathbb{Z}^d,\ |V^-| \leq k\}.
\]
\end{reptheorem}

\begin{proof}
We first consider the special case that $V$ fulfils $(V^+)^- \subseteq V$. In that case, let $T_1, \ldots, T_n$ be a tiling that minimizes $t_k(T_1,\ldots, T_n)$. By Lemma \ref{lem tidy}, we can assume that $|T_i^- \cap V| \leq k$. By Lemma \ref{lem no residue} we have $T_i^- \subseteq V$. Therefore $|T_i^-| \leq k$. We have
\begin{eqnarray*}
s_k(V) &=& \sum_{i=1}^{n} |T_i^0 \cap V|  + \sum_{i=1}^{n} (|T_i^- \cap V| - k)_+ \\
&=& \sum_{i=1}^{n} |T_i^0 \cap V|  \\
&=& \sum_{i=1}^{n} |T_i^0| - |V^+\setminus V| \\
&=& \sum_{i=1}^{n} |T_i|\frac{|T_i^0|}{|T_i|} - |V^+\setminus V| \\
&\geq& \sum_{i=1}^{n} |T_i|\cdot r_k - |V^+\setminus V| \\ 
&=& r_k \cdot |V^+|  - |V^+\setminus V|.
\end{eqnarray*}

In the general case, let $W = (V^+)^-$. For every $v \in V$, all its positive neighbors are in $V^+$, so $v \in (V^+)^- = W$. We conclude that $V \subseteq W$, and  consequently $V^+ \subseteq W^+$.
By Lemma \ref{lem -+}, $W^+ = ((V^+)^-)^+ \subseteq V^+$, so $(W^+)^- \subseteq (V^+)^- = W$. Therefore, the special case above applies to $W$, and we have
\[
s_k(W) \geq r_k \cdot |W^+|  - |W^+\setminus W|.
\]
Since $V^+ \subseteq W^+$ and $W^+ \subseteq V^+$, we have $V^+=W^+$. By Lemma \ref{lem coherent}, we have
\[
s_k(V) \geq s_k(W) - |W \setminus V| \geq r_k \cdot |V^+|  - |V^+\setminus W|- |W \setminus V| = r_k \cdot |V^+|  - |V^+\setminus V|.
\]
\end{proof}

\subsection{Proof of Lemma \ref{lem calculate rk}}

This is the inverse of Lemma \ref{lem -+}.

\begin{lemma} \label{lem +-}
$V \subseteq (V^+)^-$.
\end{lemma}

\begin{proof}
Choose $v \in V \subseteq V^+$. By definition of the cover $\{v\}^+ \subseteq V^+$. By the definition of the interior, we must have $v \in (V^+)^-$.
\end{proof}

This lemma rewrites $r_k$ in preparation for the proof of Lemma \ref{lem calculate rk}.

\begin{lemma} \label{lem rewrite rk}
If $k>0$, we have \[ r_k = \min_{1\leq j \leq k} \frac{f_{d, j} - j}{f_{d, j}},\]
where $f_{d,k} = \min\{|V^+| \colon V\subset \mathbb{Z}^d, |V|=k\}$. 
\end{lemma}

\begin{proof}
Let $k>0$. From the definition of $r_k$, we have
\begin{eqnarray*}
r_k &=& \min_{0\leq j \leq k} \min\{\frac{|V^0|}{|V|} \colon\emptyset \neq V\subset \mathbb{Z}^d, |V^-|=j\} \\
&=& \min_{1\leq j \leq k} \min\{\frac{|V|-j}{|V|} \colon V\subset \mathbb{Z}^d, |V^-|=j\}\\
&=& \min_{1\leq j \leq k} \frac{f_{d, j} - j}{f_{d, j}},
\end{eqnarray*}
where $f_{d,k} = \min\{|V| \colon V\subset \mathbb{Z}^d, |V^-|=k\}$. By Lemma \ref{lem -+}, $V \supseteq (V^-)^+$. Moreover, combining Lemma \ref{lem -+} and Lemma \ref{lem +-}, we have 
\[
V^- \subseteq ((V^-)^+)^- \subseteq V^-,
\]
so $V^- = ((V^-)^+)^-$. It follows that the minimum in the definition of $f_{d,k}$ is attained when $V = (V^-)^+$. Calling $W=V^-$, we get $f_{d,k} = \min\{|W^+| \colon W\subset \mathbb{Z}^d, |W|=k\}$.
\end{proof}

We calculate $f_{d,k}$ for low $d$ and $k$ as a basis for induction.

\begin{lemma} \label{lem small d,k}
If $k=0$, we have $f_{d,k} = 0$. If $d=1$ and $k>0$, we have $f_{d,k} = k+1$.
\end{lemma}

\begin{proof}
We have $0 \leq f_{d,0} \leq |\emptyset^+| = 0$, so $f_{0,k}=0$. Let $k>0$ and $d=1$. Let $v = \max \{v_1\colon v \in V\}$. Then $v+1 \in V^+\setminus V$. Therefore $1+k \leq f_{d,k} \leq \{1,\ldots,k\}^+ = k+1$. So $f_{1,k}=k+1$.
\end{proof}

\begin{lemma} \label{lem project}
If $V \subset \mathbb{Z}^d$ and $1\leq h \leq d$, we have $|V^+| \geq |X_h(V)^+| + \sum_{i\in D_h(V)} |S_{h,i}(V)^+|$, where 
\[
X_h(V) = \{ v \in \mathbb{Z}^{d-1}\colon \textrm{$(v_1,\ldots, v_{h-1}, i, v_{h+1}, \ldots, v_d) \in V$ for at least one $i \in \mathbb{Z}$}\} 
\]
is the projection of $V$ on $(v_1, \ldots, v_h, v_d)$, 
\[
D_h(V) = \{ i \in \mathbb{Z}\colon \textrm{$(v_1,\ldots, v_{h-1}, i, v_{h+1}, \ldots, v_d) \in V$ for at least one $v \in \mathbb{Z}^{d-1}$}\} 
\]
is the projection of $V$ on $v_h$, and 
\[
S_{h,i}(V) = \{ v \in \mathbb{Z}^{d-1}\colon (v_1,\ldots, v_{h-1}, i, v_{h+1}, \ldots, v_d) \in V\}. 
\]
is the slice of $V$ at $v_h=i$. If $V$ is convex in the direction of the $h$th unit vector $u$, i.e.\ if $v\in V$ and $v+iu \in V$, with $i>0$ implies $v+(i-1)u \in V$, then we have $|V^+| = |X_1(V)^+| + \sum_{i\in D_1(V)} |S_{1,i}(V)^+|$.
\end{lemma}

\begin{proof}
Without loss of generality let $h=1$.
Call $$W_0 = \{v+e\colon v\in V, e \in \{0\} \times \{0,1\}^{d-1}\}$$ and $$W_1 = \{v+e\colon v\in V, e \in \{1\} \times \{0,1\}^{d-1}\} \setminus W_0.$$ Then we have $V^+ = W_0\cup W_1$. 

Let $w \in X_1(V^+)$ and $m(w) = \max\{i \in \mathbb{Z}\colon (i, w) \in V^+\}.$ Then $(m(w),w) \notin W_0$, since otherwise $(m(w)+1,w)\in W_1 \subseteq V^+$, which contradicts the definition of $m(w)$. Therefore $(m(w),w) \in V^+\setminus W_0 = W_1 \setminus W_0$. Since $(m(w),w)$ is unique for $w$, we have $|W_1\setminus W_0| \geq |X_1(V^+)|$. Suppose $V$ is convex in the direction $u$. Choose $v \in W_1\setminus W_0 = V^+\setminus W_0$. Then $v=(i,w)$ for some $w \in X_1(V^+)$. Since $v \in W_1$, there is an $e \in \{0,1\}^{d-1}$ such that $(i-1, w+e) \in V$. Since $v\notin W_0$, we have $(i,w+e) \notin V$. Since $V$ is convex in $u$, we must have that $(j, w+e) \notin V$ for all $j>i$, so $i = \max\{j \in \mathbb{Z}\colon (j, w+e) \in V\}$ is unique. Therefore, $|W_1\setminus W_0| \leq |X_1(V^+)|$. We have $|W_1\setminus W_0| = |X_1(V^+)|$ if $V$ is convex in the direction $u$, and $|W_1\setminus W_0| \geq |X_1(V^+)|$ in general.

We will now show that $X_1(V^+) = X_1(V)^+$. We have that $v \in X_1(V)^+$ if and only if there exists $e \in \{0,1\}^{d-1}$ such that $v-e \in X_1(V)$. This happens if and only if there exist $e \in \{0,1\}^{d-1}$ and $i \in \mathbb{Z}$ such that $(i, v-e) \in V$, which is equivalent to the existence of $i \in \mathbb{Z}$ such that $(i, v) \in V^+$, which happens if and only if $v \in X_1(V^+)$. Therefore $X_1(V)^+ = X_1(V^+)$, and we have $|W_1\setminus W_0| \geq |X_1(V)^+|$, with equality if $V$ is convex in the direction $u$. 

Choose $i \in D_1(V)$. We have that $v \in S_{1,i}(V)^+$ happens if and only if there exists $e \in \{0,1\}^{d-1}$ such that $v-e \in S_{1,i}(V)$, equivalently $(i,v-e) \in V$, which happens if and only if $(i,v) \in W_0$, or $v \in S_{1,i}(W_0)$. It follows that  $|S_{1,i}(V)^+| = |S_{1,i}(W_0)|$.

We have
\[
V^+ = (W_1\setminus W_0) \cup W_0 = (W_1\setminus W_0) \cup \bigcup_{i\in D(V)} \{v\in W_0\colon v_1 = i\}.
\]
Since all these sets are disjoint, we have
\[
|V^+| = |W_1\setminus W_0| + \sum_{i\in D_1(V)} |S_{1,i}(W_0)| \geq |X_1(V)^+| + \sum_{i\in D_1(V)} |S_{1,i}(V)^+|,
\]
with equality if $V$ is convex in $u$.
\end{proof}

The next lemma states that $b_{d,k}$ is the largest number of the form $q^{d-l}q^l$ that does not exceed $k$.

\begin{lemma} \label{lem bdk}
We have $0 \leq k - b_{d,k} <  \frac{b_{d,k}}{\lfloor k^{1/d}\rfloor}$. 
\end{lemma}

\begin{proof}
Write $q=\lfloor k^{1/d}\rfloor$. Then $q^{d-l}(q+1)^l \leq k$ is equivalent to
\[
(d-l)\log(q) + l \log(q+1) \leq \log(k),
\]
and to 
\[
l \leq \frac{\log(k) - d\log(q)}{\log(q+1) - \log(q)} = l_{d,k},
\]
where $l_{d,k}$ is defined in Lemma \ref{lem calculate rk}. It follows that
$b_{d,k} = q^{d-l_{d,k}}(q+1)^{l_{d,k}} \leq k$, but 
\[
b_{d,k} \frac{\lfloor k^{1/d}\rfloor+1}{\lfloor k^{1/d}\rfloor} = q^{d-l_{d,k}-1}(q+1)^{l_{d,k}+1} > k.
\]
\end{proof}

\begin{lemma} \label{lem decompose b}
Let $c_{d,k} = b_{d,k} - b_{d,k-1}$ and $c_{d,k}^+ = b_{d,k}^+ - b_{d,k-1}^+$ if $k>0$ and $c_{d,k}=b_{d,k}=1$ and $c_{d,k}^+=b_{d,k}^+=2^d$ if $k=1$. If $k\neq b_{d,k}$, then $c_{d,k}=c_{d,k}^+=0$. If $k=b_{d,k}$, then 
\[
c_{d,k} = k/(q'_{d,k}+1), 
\]
where $q'_{d,k}= \lfloor k^{1/d}\rfloor$ and $l'_{d,k}=l_{d,k}$ if $l_{d,k}>0$ and $q_{d,k}'=\lfloor k^{1/d}\rfloor-1$ and $l'_{d,k} = d$ if $l_{d,k}=0$. If $0<k=b_{d,k}$ and $d>1$, then
\[
c_{d,k}^+ = f_{d-1, c_{d,k}}. 
\]
If $k=1$ and $d>1$ we have $c_{d,1}^+ = 2^d > f_{d-1, c_{d,1}} = f_{d-1, 1} = 2^{d-1}$.
\end{lemma}

\begin{proof}
The part for $k\neq b_{d,k}$ follows immediately from the definition of $b_{d,k}$. Let $k=b_{d,k}$.
We have 
\[
b_{d,k} = (q'_{d,k})^{d-l'_{d,k}}(q'_{d,k}+1)^{l'_{d,k}},
\] 
and 
\[
b_{d,k-1} = (q'_{d,k})^{d-l'_{d,k}+1}(q'_{d,k}+1)^{l'_{d,k}-1}.
\] 
Therefore, if $k>1$,
\[
c_{d,k} = (q_{d,k}')^{d-l_{d,k}'}(q_{d,k}'+1)^{l_{d,k}'-1}(q_{d,k}'+1 - q'_{d,k}) = (q'_{d,k})^{d-l_{d,k}'}(q_{d,k}'+1)^{l_{d,k}'-1} = b_{d,k}/(q'_{d,k}+1).
\]
If $k=1$, we have $l_{d,k} = 0$, so $l'_{d,k} = d$ and $q'_{d,k} = 0$, so the equality still holds.
Completely analogously, we get
\[
c_{d,k}^+ = (q'_{d,k}+1)^{d-l_{d,k}'}(q_{d,k}'+2)^{l_{d,k}'-1} = f_{d-1, c_{d,k}},
\]
where the latter equality is meaningful only if $d>1$. The inequality for $k=1$ is trivial.
\end{proof}

\begin{lemma} \label{lem merge}
The following two statements hold:
\begin{enumerate}
    \item $f_{d,m} \leq f_{d,k}+f_{d,l}$, where $m=k+l$;
    \item If $m \geq \max(n,k,l)$, $m \in B_d$, and $m+n=k+l$, then $f_{d,m}+f_{d,n} \leq f_{d,k} + f_{d,l}$.
\end{enumerate}
\end{lemma}

\begin{proof}
We use induction on $d$ and $k$ (downward). If $d=1$ both statements are easily checked from Lemma \ref{lem small d,k}. If $k=m$ both are trivial.
Fix $m$ and $n$. Suppose that both statements hold for all $m,n,k,l$ in smaller dimensions, and for all larger values of $k\leq m$.

Without loss of generality, assume $k\geq l$. We consider 
\[
f_{d,k} + f_{d,l} = b^+_{d,k} + f_{d-1, k-b_{d,k}}+ b^+_{d,l} + f_{d-1, l-b_{d,l}}.
\]
Let $k' = \min\{i\in B_d\colon i>k\}$.  One of the cases holds:
\begin{enumerate}
    \item $l-d_{b,l}>0$ and $l-d_{b,l}+k-b_{d,k}< c_{d,k'}$;
    \item $l-d_{b,l}>0$ and $l-d_{b,l}+k-b_{d,k}\geq c_{d,k'}$;
    \item $l-d_{b,l}=0$ and $l-d_{b,l}+c_{d,b_{d,l}}< c_{d,k'}$;
    \item $l-d_{b,l}=0$ and $l-d_{b,l}+c_{d,b_{d,l}}\geq c_{d,k'}$.
\end{enumerate}
In Case 1, we use induction on $d$ for Statement 1 and write
\[
f_{d,k} + f_{d,l} \geq b^+_{d,k} + f_{d-1, l-d_{b,l}+k-b_{d,k}}+ b^+_{d,l} f_{d-1,c_{d,i}} 
= f_{d, k+l-d_{b,l}} + f_{d, b_{d,l}}.
\]
In Case 2, we note that $l-b_{d,l} < c_{d, b_{d,l}} \leq c_{d,k'}$, and
use induction on $d$ for Statement 2, and Lemma \ref{lem decompose b}, to write
\[
f_{d,k} + f_{d,l} \geq b^+_{d,k} + 
f_{d-1, c_{d,k'}} + b^+_{d,l} + f_{d-1, l-d_{b,l}+k-b_{d,k}-c_{d,k'}} = f_{d, k'} + f_{d, l+k-k'}.
\]
In Case 3 and 4, we note that $l=b_{d,l}$ and write, using Lemma \ref{lem decompose b},
\[
f_{d,k} + f_{d,l} \geq b^+_{d,k} + f_{d-1, k-b_{d,k}} + b^+_{d,l-c_{d,l}} + f_{d-1, c_{d,l}},  
\]
where we write $\geq$ to cover the case that $l=1$. In Case 3 we use induction on $d$ for Statement 1 and write
\[
f_{d,k} + f_{d,l}   
\geq b^+_{d,k} + f_{d-1, k-b_{d,k} + c_{d,b_{d,l} }} + b^+_{d,l-c_{d,l}}
= f_{d, k+ c_{d,l}} + f_{d, l- c_{d,l}}.
\]
In Case 4, since $c_{d,l} \leq c_{d,k'}$, we use induction on $d$ for Statement 2, and Lemma \ref{lem decompose b}, to write
\[
f_{d,k} + f_{d,l} \geq b^+_{d,k} + 
f_{d-1, c_{d,k'}}  + b^+_{d,l-c_{d,l}} + f_{d-1, c_{d,l}+k-b_{d,k}-c_{d,k'}}
= f_{d, k'} + f_{d, l+k-k'}.
\]

In all cases, we have $f_{d,k}+f_{d,l} \geq f_{d,k''}+f_{d,l''}$ with $k''>k$ and $k''+l''=k+l$. By the induction hypotheses on $k$ we have 
\[
f_{d,k}+f_{d,l} \geq f_{d,k''}+f_{d,l''} \geq f_{d,m}.
\]
Since $k'' \leq k' \leq m$, we retain the conditions of Statement 2, so we can also call on the induction hypothesis for $k$ in Statement 2, obtaining
\[
f_{d,k}+f_{d,l} \geq f_{d,k''}+f_{d,l''} \geq f_{d,m} + f_{d,n},
\]
and we have proved both statements.
\end{proof}

\begin{lemma} \label{lem exist}
For every $d>0$ and $k\geq 0$ there exists $F_{d,k} \subset \mathbb{Z}^d$ such that $|F_{d,k}| = k$ and $|F^+_{d,k}| = f_{d,k}$. Moreover, if $k\leq k'$, then $F_{d,k} \subseteq F_{d,k'}$.
\end{lemma}

\begin{proof}
We use induction on $d$. For $d=1$, $F_{d,k} = \{1,\ldots, k\}$ satisfies $|F_{d,k}| = k$ and $|F^+_{d,k}| = k+1 = f_{d,k}$ by Lemma \ref{lem small d,k}. It is immediate that $F_{d,k} \subseteq F_{d,k'}$ if $k\leq k'$.

Suppose the Lemma holds for all smaller values of $d$. 
Define 
\[
B_{d,k} = \{v\in \mathbb{Z}^d\colon \textrm{$1\leq v_i \leq \lfloor k^{1/d} \rfloor+1$ for $1\leq i \leq l_{d,k}$, and $1\leq v_i \leq \lfloor k^{1/d} \rfloor$ for $l_{d,k}< i \leq d$}\},
\]
and
\[
F_{d,k} = B_{d,k} \cup \{(w_1, \ldots, w_{l_{d,k}}, \lfloor k^{1/d} \rfloor+1, w_{l_{d,k}+1}, \ldots, w_{d-1})\colon w \in F_{d-1, k-b_{d,k}}\}.
\]
Clearly, $|B_{d,k}| = b_{d,k}$ and $B^+_{d,k} = b^+_{d,k}$.
We have $|F_{d,k}| = b_{d,k} + (k-b_{d,k})=k$. Since 
\[
b-b_{d,k} < b_{d,k}/\lfloor k^{1/d} \rfloor = q_{d,k}^{d-l_{d,k}-1}(q_{d,k}+1)^{l_{d,k}},
\] 
by Lemma \ref{lem bdk}, we have 
\begin{eqnarray*}
F_{d-1, b_{d,k}} &\subset& F_{d-1, b_{d,k}/\lfloor k^{1/d}\rfloor} \\
&=& B_{d-1, b_{d,k}/\lfloor k^{1/d}\rfloor} \\
&=& \{v\in \mathbb{Z}^{d-1}\colon \textrm{$1\leq v_i \leq \lfloor k^{1/d} \rfloor+1$ for $1\leq i \leq l_{d,k}$, and $1\leq v_i \leq \lfloor k^{1/d} \rfloor$ for $l_{d,k}< i \leq d-1$}\} \\
&=& X_{l_{d,k}+1}(B_{d,k}).
\end{eqnarray*}
Therefore, by the induction hypothesis and Lemma \ref{lem project} we have $|F^+_{d,k}| = b^+_{d,k} + f_{d-1, k-b_{b,k}} = f_{d,k}$.

To show the inclusion it suffices to take $k'=k+1$. If $b_{d,k} < b_{d,k+1}$, then $F_{d,k+1} = B_{d,k+1} \supset B_{d,k}$, and $|B_{d,k+1} \setminus B_{d,k}| = k+1 - b_{d,k}$ by Lemma \ref{lem decompose b}. We have 
\begin{eqnarray*}
F_{d,k+1} \setminus B_{d,k} 
&=& \{(w_1, \ldots, w_{l_{d,k}}, \lfloor k^{1/d} \rfloor+1, w_{l_{d,k}+1}, \ldots, w_{d-1})\colon w \in B_{d-1, k + 1 - b_{d,k}}\} \\
&=& \{(w_1, \ldots, w_{l_{d,k}}, \lfloor k^{1/d} \rfloor+1, w_{l_{d,k}+1}, \ldots, w_{d-1})\colon w \in F_{d-1, k+1-b_{d,k}}\}.
\end{eqnarray*}
The same holds immediately from the definition if $b_{d,k} = b_{d,k+1}$. The result that $F_{d,k} \subset F_{d,k+1}$ now follows from the induction hypothesis.
\end{proof}

\begin{lemma} \label{lem exist corrolary}
Let $d>1$, and suppose we have $k_1, \ldots, k_m\geq 0$ with $k_0=\max_{1\leq i\leq m} k_i$. Then there exists $V \subset \mathbb{Z}^d$ such that $|V| = \sum_{i=1}^m k_i$ and $|V_+| = \sum_{i=0}^m f_{d-1, k_i}$. 
\end{lemma}

\begin{proof}
Without loss of generality, assume $k_1 \geq \ldots \geq k_m$. For $i=1,\ldots, m$, let $V_i = \{(i, v)\colon v\in F_{d-1, k_i}\}$, where $F_{d,k}$ is defined in Lemma \ref{lem exist}, and let $V = \bigcup{i=1}^m V_i$. Then $|V| = \sum_{i=1}^m k_i$ by Lemma \ref{lem exist}. Moreover, since  $F_{d-1, k_i} \subseteq F_{d-1, k_1}$ for $1\leq i\leq m$, we have $X_1(V) = V_1$, and $V$ is convex in $(1,0,\ldots,0)$. By Lemma \ref{lem project} and Lemma \ref{lem exist},
\[
|V^+| = |X_1(V)| + \sum_{i=1}^m |F_+{d-1, k_i}| = \sum_{i=0}^m f_{d-1, k_i}.
\]
\end{proof}

\begin{lemma} \label{lem gdk fdk}
We have $g_{d,k} = \min\{|V^+|\colon V \subset \mathbb{Z}^d, |V|=k\} = f_{d,k}$. 
\end{lemma}

\begin{proof}
We use induction. If $k=0$, We have $0 \leq g_{d,0} \leq |\emptyset^+| = 0$, so $g_{0,k}=0=f_{d,k}$. Let $k>0$ and $d=1$. Let $v = \max \{v_1\colon v \in V\}$. Then $v+1 \in V^+\setminus V$. Therefore $1+k \leq g_{d,k} \leq \{1,\ldots,k\}^+ = k+1$. If $d=1$, we have $l_{d,k}=0$, so $f_{1,k}=k+1 =g_{1,k}$.

Choose, $k,d$ and assume the statement of the Lemma holds for all lower values of $d$ and $k$. 

Let $V$ be such that $|V|=k$. For $1\leq i \leq d$, let 
\[
m_i = \{j \in \mathbf{Z}\colon \textrm{$v_i=j$ for some $v\in V$}\},
\]
the number of coordinates in dimension $v_i$ that $V$ uses. Remark that if $\max_{1\leq i\leq d} m_i \geq q$, where $q= \lfloor k^{1/d} \rfloor$, or otherwise 
\[
|V| \leq \prod_{i=1}^d m_i \leq (\max_{1\leq i\leq d} m_i)^d < q^d \leq k.
\]
If $\max_{1\leq i\leq d} m_i = k^{1/d}$, then $V$ must be a $d$-cube of dimension $k^{1/d}$. Then $|V^+| = (k^{1/d} +1)^d = b^+_{d,k} = f_{d,k}$, and the result of the lemma follows immediately. In the rest of the proof, we can therefore assume that $\max_{1\leq i\leq d} m_i > q$

Remark also that $|V^+|$ is invariant by definition to a permutation of the dimensions, so without loss of generality we can assume that $m = m_1 \geq q+1$.

By Lemma \ref{lem project}  and the induction hypothesis on $d$ we have 
\[
|V^+| \geq |X_1(V)^+| + \sum_{i=1}^m |S_{1,i}(V)^+| \geq \sum_{i=0}^m f_{d-1,n_i},
\]
where $n_0 = |X_1(V)|$ and $n_i = |S_{1,i}(V)^+|$. Without loss of generality, let $n=n_m = \min_{1\leq i \leq d} n_i$. We have 
\[
k = |V| = \sum_{i=1}^{m} n_i \geq m n \geq (q+1) n,
\]
so $n \leq k/(q+1)$.

By Lemma \ref{lem exist corrolary} there exists a $W$ such that $|W| = \sum_{i=1}^{m-1} n_i$ and $|W^+| = \sum_{i=0}^{m-1} f_{d-1,n_i}$. By the induction hypothesis on $k$ we have
\[
|V^+| \geq |W^+| + f_{d-1,n} \geq f_{d, k-n}+ f_{d-1,n} =
b^+_{d, k-n} + f_{d-1, k-n-b_{d, k-n}} + f_{d-1,n}.
\]
If $b_{d,k-n} = b_{d,k}$, we have, by Lemma \ref{lem merge},
\[
|V^+| \geq b^+_{d, k} + f_{d-1, k-n-b_{d,k}} + f_{d-1,n} \geq b^+_{d, k} + f_{d-1, k-b_{d, k}} = f_{d,k}.
\]
Suppose $b_{d,k-n} < b_{d,k}$. Then, by Lemma \ref{lem decompose b},
\[
k-n \geq b_{d,k} - b_{d,k}/(q+1) \geq b_{d,k} - b_{d,k}/(q'+1) = b_{d,k} - c_{d,k} = b_{d,k-1},
\]
so $b_{d,k-n} \geq b_{d,k} - c_{d,k}$. Then $k-n < b_{d,k}$, so that, also by Lemma \ref{lem decompose b},
\[
k-n - b_{d,k-n} < b_{d,k} - (b_{d,k} - c_{d,k}) = c_{d,k}.
\]
Moreover, $n \leq k/(q+1) \leq k/(q'+1) = c_{d,k}$, by Lemma \ref{lem decompose b}, where $q'=q'_{d,k}$ and $c_{d,k}$ are defined in that lemma. Finally, $k-b_{d,k-n} -c_{d,k} = k-b_{d,k}< c_{d,k}$ by definition of $b_{d,k}$. By Lemma \ref{lem merge} and Lemma \ref{lem decompose b}, we have
\[
|V^+| \geq b^+_{d, k-n} + f_{d-1, c_{d,k}} + f_{d-1,k-c_{d,k}-b_{d, k-n}} = b^+_{d,k} + f_{d-1,k-b_{d,k}} = f_{d,k}. 
\]
Since $|V^+| \geq f+{d,k}$ holds for all $V \subseteq \mathbb{Z}^d$ with $|V|=k$, it follows that $g_{d,k} \geq f_{d,k}$.

By Lemma \ref{lem exist} there exists a $W \subseteq \mathbb{Z}^d$ such that $|W|=k$ and
\[
|W^+| = f_{d,k} = b^+_{d,k} + f_{d-1,k-b_{d,k}} = (q+1)(b_{d-1, b_{d,k}/q}) + f_{d-1,k-b_{d,k}},
\]
so $g_{d,k} \leq f_{d,k}$. Combining, we have $g_{d,k} = f_{d,k}$, as desired.
\end{proof}

\begin{replemma}{lem calculate rk}
If $k=0$, we have $r_k=1$. If $k>0$, we have 
\[
r_k = \min_{1\leq j \leq k} \frac{f_{d, j} - j}{f_{d, j}},
\]
where, $f_{1,k}=(k+1)\mathds{1}\{k>0\}$ and, for $d>1$, we have recursively \[ f_{d,k} = b^+_{d,k} + f_{d-1, k- b_{d,k}}.\] Here, 
\[
b_{d,k} = \big(\lfloor k^{1/d}\rfloor\big)^{d-l_{d,k}} \big(\lfloor k^{1/d}\rfloor+1\big)^{l_{d,k}},
\]
and 
\[
b^+_{d,k} = \big(\lfloor k^{1/d}\rfloor+1\big)^{d-l_{d,k}} \big(\lfloor k^{1/d}\rfloor+2\big)^{l_{d,k}},
\]
where 
\[
l_{d,k} = \Big\lfloor \frac{\log(k) - d\log(\lfloor k^{1/d}\rfloor)}{\log(\lfloor k^{1/d}\rfloor+1) - \log(\lfloor k^{1/d}\rfloor)} \Big\rfloor.
\]
\end{replemma}

\begin{proof}
Combine Lemma \ref{lem rewrite rk} with Lemma \ref{lem gdk fdk}.
\end{proof}

\subsection{Proof of Lemma \ref{lem prune stop}}

The following lemma generalizes Lemma \ref{lem -+}.

\begin{lemma} \label{lem prune}
$V^{(i)} \subseteq V$.
\end{lemma}

\begin{proof}
Let $W = ((V^-)^{\cdots})^-$, where the interior operation $(\cdot)^-$ is done $i$ times. Choose $v \in V^{(i)}$. By definition of the cover there must be a $w \in W$ such that $v = w + e$ with $e \in \{0,\ldots, i\}^d$. By definition of the interior, $w\in W$ implies that every $u = w + e$ with $e \in \{0,\ldots, i\}^d \in V$, so in particular $v \in V$.
\end{proof}

\begin{replemma}{lem prune stop}
If $i \geq \lfloor |V|^{1/d}\rfloor$, then $V^{(i)} = \emptyset$.  
\end{replemma}

\begin{proof}
Let $W = ((V^-)^{\cdots})^-$, where the interior operation $(\cdot)^-$ is done $i$ times. If $W =\emptyset$, then $V^{(i)} = \emptyset$ and we are done. We will assume that $W \neq \emptyset$ and arrive at a contradiction. Let $w \in W$. Then $w+ e \in V^{(i)}$ for all $e \in \{0,\ldots, i\}^d$. Therefore $|V^{(i)}| \geq (i+1)^d > (|V|^{1/d})^d = |V|$, which contradicts Lemma \ref{lem prune}.
\end{proof}

\subsection{Proof of Theorem \ref{thm final}}

We first prove the relevant bound for $s_k(V)$.

\begin{lemma} \label{lem sqrt bound}
$\check s_k(V) \leq s_k(V)$
\end{lemma}

\begin{proof}
If $\chi_{V} >k$, then $s_k(V) >0$, as follows immediately from the definition of $s_k(V)$, so\[ \mathds{1}\{\chi_{V} >k\} \leq s_k(V). \] 

For any $i\geq 0$, by Lemma \ref{lem prune}, $V^{(i)} \subseteq C$. Therefore, by Theorem \ref{thm shortcut 2} and Lemma \ref{lem monotone chi},
\[
\underline s_k(V^{(i)}) \leq s_k(V^{(i)}) = \min\{|R|\colon \chi_{V^{(i)}\setminus R} \leq k\} \leq \min\{|R|\colon \chi_{V\setminus R} \leq k\} \leq s_k(V)
\]
Since $\mathds{1}\{\chi_{V} >k\}$ and $\underline s_k(V^{(i)})$ for $i=0, \ldots, |V|^{1/d}$ are all smaller than $s_k(V)$, so is their maximum. Since $s_k(V)$ is an integer, the result follows.
\end{proof}

\begin{reptheorem}{thm final}
For every $V \subseteq M$, let 
\[
\underline{\mathbf{a}}(V) = \sum_{i=1}^\mathbf{n} \check s_{k_M} (\mathbf{C}_i),
\]
where $\mathbf{C}_1,\ldots,\mathbf{C}_\mathbf{n}$ are disconnected clusters such that $\mathbf{C}_1 \cup \cdots \cup \mathbf{C}_\mathbf{n} = V \cap \mathbf{Z}$. Then, for all $\mathrm{P} \in\Omega$,
\[
\mathrm{P}(\textrm{$\underline {\mathbf{a}}(V) \leq a_\mathrm{P}(V)$ for all $V \subseteq M$}) \geq 1-\alpha.   
\]
\end{reptheorem}

\begin{proof}
By Lemma \ref{lem sqrt bound}, Lemma \ref{lem decompose clusters}, and Theorem \ref{thm first TDP bound}, we have
$$
\underline{\mathbf{a}}(V) \leq \sum_{i=1}^{\mathbf{n}} s_{k_M}(\mathbf{C}_i) = s_{k_M}(\mathbf{C}_1 \cup \cdots \cup \mathbf{C}_{\mathbf{n}}) = s_{k_M}(V \cap \mathbf{Z}) = \check{\mathbf{a}}(V).
$$
By Theorem \ref{thm first TDP bound}, we therefore have, for all $\mathrm{P} \in\Omega$,
\[
\mathrm{P}(\textrm{$\underline {\mathbf{a}}(V) \leq a_\mathrm{P}(V)$ for all $V \subseteq M$}) \leq
\mathrm{P}(\textrm{$\check {\mathbf{a}}(V) \leq a_\mathrm{P}(V)$ for all $V \subseteq M$}) \geq 1-\alpha.   
\]
\end{proof}

\subsection{Proof of Lemma \ref{lemma:heurbounds}}

\begin{replemma}{lemma:heurbounds}
Let $k=n^d$ and $c$ be a vector of $d$ positive integers. If the dimensions of a hyperrectangle $R$ are $(n+1)c_i - 1$ for $i=1,\ldots,d$, then the bound of Theorem \ref{thm shortcut 2} is exact, so that the optimal $k$-separator of $R$ has $|R|-n^d\Pi c_i$ voxels.
\end{replemma}

\begin{proof}
We first infer a $k$-separator $K$ of the size $|R|-n^d\Pi c_i$. Assume, the $\min(R)=1$ (as a vector). A voxel is an element of $K$ if and only if there are $i$ and $j \in \{1,2,\dots,c_i-1\}$ such that 
the $i$-th coordinate of the voxel is equal to $(n+1)j$. Then, $R \setminus K$ consists of $\Pi c_i$ separated $d$-cubes of the size $k=n^d$. Thus, $K$ is a $k$-separator of $R$. This completes the first part of the proof. 

Next, we show that the separator is optimal. It follows from Lemma~\ref{lem calculate rk}, that in our setting when $k=n^d$, $l_{d,k}=0$ and $f(d,k)=(n+1)^d$. Moreover, it follows from the definition of $r_k$ in Lemma~\ref{lem calculate rk} that $\frac{f_{d,j}-j}{f_{d,j}}$ is minimal when $j=k$. Thus, $r_k=r_{n^d}=\frac{(n+1)^d-n^d}{(n+1)^d}$. To complete the proof, it is sufficient to show that  
\[ \frac{(n+1)^d-n^d}{(n+1)^d}|R^+|-|R^+\setminus R| = |R| - n^d \Pi c_i. \]
The left-hand side of the above equation is the lower bound from Theorem~\ref{thm shortcut 2}, while the right hand side is the size our separator set. 
The rest follows by easy transformations, by  applying $|R^+|=\Pi (n+1)c_i = (n+1)^d \Pi c_i$ and $|R^+\setminus R|=|R^+|-|R|$.  

Since the lower bound is reached, the $k$-separator $K$ is optimal. This completes the proof.
\end{proof}

\subsection{Proof of Theorem \ref{thm kM0}}

\begin{reptheorem}{thm kM0}
If $k_M=0$, then for all $V \subseteq M$ we have 
\[
\underline{\mathbf{a}}(V) = \check{\mathbf{a}}(V) = \mathbf{a}(V) = |V \cap \mathbf{Z}|.
\]
\end{reptheorem}

\begin{proof}
We will first show that $\mathbf{a}(V) = \check{\mathbf{a}}(V)$, then that $\check{\mathbf{a}}(V) = |V \cap \mathbf{Z}|$ and $\underline{\mathbf{a}}(V) = |V \cap \mathbf{Z}|$.

From (\ref{def local test}) we have, for all $V \subseteq M$,  
\[
\boldsymbol{\phi}_V = \mathds{1}\{\chi_{V\cap\mathbf{Z}} > 0\} = \mathds{1}\{|V\cap\mathbf{Z}| > 0\}.
\]
Therefore, $\boldsymbol{\phi}_V \leq \boldsymbol{\phi}_W$ if $V \subseteq W$, and we have, for all $V \subseteq M$,
\[
\boldsymbol{\psi}_V = \min\{\boldsymbol{\phi}_W\colon V \subseteq W \subseteq M\}  = \boldsymbol{\phi}_V. 
\]
Therefore, for all $V \subseteq M$, $\underline{\boldsymbol{\psi}}_V = \mathds{1}\{\chi_{V\cap\mathbf{Z}} > 0\} = \boldsymbol{\psi}_V$, so that 
$\mathbf{a}(V) = \check{\mathbf{a}}(V)$.

Now by Theorem \ref{thm first TDP bound} we have $\check{\mathbf{a}}(V) = s_0(V \cap \mathbf{Z})$, and for any $W \subseteq M$, $s_0(W) = \min\{|R|\colon \chi_{W \setminus R} = 0\} = |W|$, so $\check{\mathbf{a}}(V) = |V \cap \mathbf{Z}|$.

If $k = 0$, then $r_k = 1$ by definition of $r_k$, so $\underline s_0(V) = |V^+| - |V^+\setminus V| = |V|$. It follows from Theorem \ref{thm final} that, if $\mathbf{C}_1, \ldots, \mathbf{C}_\mathbf{n}$ are defined as in that theorem, then 
\[
\underline{\mathbf{a}}(V) \geq \sum_{i=1}^\mathbf{n} \underline s_0(\mathbf{C_i}) = 
\sum_{i=1}^\mathbf{n} |\mathbf{C_i}| = |V \cap \mathbf{Z}|.
\]
We have, for all $V \subseteq M$,
\[
|V \cap \mathbf{Z}| \leq \underline{\mathbf{a}}(V) \leq \mathbf{a}(V) = \check{\mathbf{a}}(V) = |V \cap \mathbf{Z}|,
\]
so we must have equality and the statement of the theorem follows.
\end{proof}

\subsection{Proof of Theorem \ref{thm bound rk}}

We first prove an upper bound on $\underline s_k(V)$.

\begin{lemma} \label{lem upper bound of lower bound}
We have 
\[
\underline s_k(V) \leq \frac{r_k - r_{|V|}}{1- r_{|V|}} \cdot |V|.
\]
\end{lemma}

\begin{proof}
If $k=0$, we have $r_k=1$, so the statement of the Lemma reads $\underline s_k(V) \leq |V|$, which follows immediately from the definition. Let $k>0$. We have 
\begin{eqnarray*}
\underline s_k(V) &=& r_k \cdot |V^+| - |V^+ \setminus V| \\
&=& r_k \cdot |V^+| - |V^+| + |V| \\
&=& |V| - (1-r_k) \cdot |V^+| \\
&\leq& |V| - (1-r_k) \min\{|W^+|\colon W \subseteq \mathbb{Z}^d, |W| = |V|\} \\
&=& |V| - (1-r_k) f_{|V|},
\end{eqnarray*}
where $f_k = f_{d,k}$ is defined in Lemma \ref{lem calculate rk}, and we suppress the dependence on $d$ here. We have
\[
f_k = \frac{k}{1-\frac{f_k - k}{f_k}} \geq 
\frac{k}{1-\min_{1\leq j \leq k} \frac{f_j - j}{f_j}} = \frac{k}{1-r_k},
\]
so that 
\[
\underline s_k(V) \leq \Big(1 - \frac{1-r_k}{1-r_{|V|}}\Big) \cdot |V|,
\]
which rewrites to the statement of the Lemma.
\end{proof}

\begin{reptheorem}{thm bound rk}
For every cluster $\mathbf{C} \subseteq \mathbf{Z}$, we have 
\[
\underline{\mathbf{a}}(\mathbf{C}) \leq \Big\lceil \frac{r_k-r_{|\mathbf{C}|}}{1-r_{|\mathbf{C}|}} \cdot |\mathbf{C}| \Big\rceil \vee \mathds{1}\{|\mathbf{C}| > k\}.
\]
\end{reptheorem}

\begin{proof}
Choose any $i\geq 0$. By lemma \ref{lem prune}, $V^{(i)} \subseteq V$. By Lemma \ref{lem upper bound of lower bound}, we have, since $r_k$ is decreasing in $k$ by definition,
\[
\underline s_k(V^{(i)}) \leq \Big(1 - \frac{1-r_k}{1-r_{|V^{(i)}|}}\Big) \cdot |V^{(i)}| \leq \Big(1 - \frac{1-r_k}{1-r_{|V|}}\Big) \cdot |V|. 
\]
Therefore
\[
\check s_k(V) \leq \Big\lceil\Big(1 - \frac{1-r_k}{1-r_{|V|}}\Big) \cdot |V| \Big\rceil \vee \mathds{1}\{\chi_V > k\}.
\]
The result of the Proposition now follows directly from Theorem \ref{thm final}, remarking that if $\mathbf{C}$ is a cluster, that $\mathds{1}\{\chi_{\mathbf{C}} > k\} = \mathds{1}\{|C| > k\}$.
\end{proof}

\subsection{Proof of Lemma \ref{lem tilde r}}

\begin{replemma}{lem tilde r}
We have $s_k(V) \leq \tilde r_k \cdot |V|$, where $\tilde r_k = (b^+_{d,k} - b_{d,k})/b^+_{d,k}$.
\end{replemma}

\begin{proof}
By definition of $b^+_{d,k}$ and $b_{d,k}$ there are integers $q_1, \ldots, q_d$ such that $b_{d,k} = q_1q_2\cdots q_d \leq k$ and $b_{d,k} = (q_1+1)(q_2+1)\cdots (q_d+1)$. Let $R'_1 \subseteq \mathds{Z}^d$ be the set for which the $i$th coordinate is divisible by $q_i+1$, for $i=1,\ldots, d$. Then $R_1'$ is a $k$-separator of $\mathds{Z}^d$, so that $R_1 = R_1' \cap V$ is a $k$-separator of $V$. Let $R'_2, \ldots, R_{b^+-{d,k}}'$ be analogously defined as all translations of $R_1'$ by $\{0, \ldots, q_1\} \times \cdots \times \{0, \ldots, q_d\}$, and define $R_2, \ldots, R_{b^+-{d,k}}$ analogously. 

For every $v \in V$, there are exactly $b^+_{d,k} - b_{d,k}$ sets $i \in \{1\,ldots,b^+_{d,k}\}$ for which $v \in R'_i$, so $v \in R_i$ and $b_{d,k}$ for which it is in $\mathds{Z}^d \setminus R_i'$, so $v \in V \setminus R_i$. We have
\[
\sum_{i=1}^{b^+_{d,k}} |R_i| = (b^+_{d,k} - b_{d,k}) \cdot |V|.
\]
It follows that there exists an $R_i$ for which 
$|R_i| \leq (b^+_{d,k} - b_{d,k})/b^+_{d,k} \cdot |V|$. Since $R_i$ is a $k$-separator of $V$, we have $s_k(V) \leq \tilde r_k \cdot |V|$.
\end{proof}

\subsection{Proof or Theorem \ref{thm bound aV}}

We prove a slightly tighter bound in Lemma  \ref{lem step down}. We first define this bound. 
Let $\mathbf{Z} \cap M = \mathbf{C}_1 \cup \cdots \cup \mathbf{C}_{\mathbf{n}}$, where $\mathbf{C}_1, \ldots, \mathbf{C}_{\mathbf{n}}$ are disconnected clusters. Let $J_0 = \emptyset$, and for $j=1,2,\ldots$, let
\[
J_{j+1} = \{1\leq i \leq \mathbf{n}\colon |\mathbf{C}_i| > k_{M\setminus \mathbf{D}_j}\}, \]
and $\mathbf{D}_i = \bigcup_{j\in J_i} \mathbf{C}_j.$
Define $\mathbf{D} = \lim_{i\to\infty} \mathbf{D}_i$. To obtain $\mathrm{D}$, therefore, we find all significant clusters according to classical cluster-extent thresholding, update the cluster extent threshold by removing those clusters from the mask, and iterate. This procedure does not have the TDP guarantee (\ref{eq simultaneous tilde}) unless $k_M = 0$. Lemma \ref{lem step down} says that the closed testing procedure is at most as powerful as this anti-conservative procedure.

\begin{lemma} \label{lem step down}
Let $\overline{\mathbf{a}}(V) = s_{\underline{k}_{M\setminus \mathbf{D}}}(V\cap \mathbf{Z})$, then, for every $V \subseteq M$,
\[
\mathbf{a}(V) \leq \overline{\mathbf{a}}(V).
\]
\end{lemma}

\begin{proof}
Choose any $V \subseteq M$. Define $\overline{\boldsymbol{\psi}}_V = \mathds{1}\{\chi_{V \cap \mathbf{Z}} > k_{M\setminus\mathbf{D}} \}.$ We will show that $\overline{\boldsymbol{\psi}}_V \geq \boldsymbol{\psi}_V$ by contradiction. Suppose that $\overline{\boldsymbol{\psi}}_V=0$ and $\boldsymbol{\psi}_V=1$. Define $\mathbf{W}= V \cup (M \setminus \mathbf{D})$. Since $W \supseteq V$ and $\boldsymbol{\psi}_V=1$, we have $\chi_{W\cap \mathbf{Z}} > k_W \geq k_{M\setminus \mathbf{D}}$. The largest cluster in $W\cap \mathbf{Z}$ is therefore a subset of a cluster of $M \cap \mathbf{Z}$ of size at least $k_{M\setminus \mathbf{D}}$. All such clusters are in fully contained in $\mathbf{D}$ by definition of $\mathbf{D}$. Therefore, the largest cluster of $W\cap \mathbf{Z}$ is also a cluster of $V\cap \mathbf{Z}$. Therefore, since $\overline{\boldsymbol{\psi}}_V=0$,
\[
\chi_{W\cap \mathbf{Z}} \leq \chi_{V\cap \mathbf{Z}} \leq k_{M\setminus \mathbf{D}} \leq k_{V \cup (M\setminus \mathbf{D})} = k_W,
\]
whence $\boldsymbol{\psi}_V=0$ since $V \subseteq W \subseteq M$, and we have a contradiction.

Starting from $\overline{\boldsymbol{\psi}}_V \geq \boldsymbol{\psi}_V$, the rest of the proof is completely analogous to the proof of Theorem \ref{thm first TDP bound}.
\end{proof}

\begin{reptheorem}{thm bound aV}
Let $\overline{\mathbf{a}}(V) = s_{k_{M\setminus \mathbf{Z}}}(V\cap \mathbf{Z})$, then, for every $V \subseteq M$,
\[
\mathbf{a}(V) \leq \overline{\mathbf{a}}(V).
\]
\end{reptheorem}

\begin{proof}
This is an immediate consequence of Lemma \ref{lem step down} if we remark that $\mathbf{D} \subseteq \mathbf{Z}$, so $k_{M \setminus \mathbf{D}} \geq k_{M \setminus \mathbf{Z}}$.
\end{proof}

\section{Heuristic algorithms to minimize $k$-separators: pseudocode}

In this section we give the pseudocode of the algorithms described in Section \ref{sec simul anneal}.

\newcommand\mycommfont[1]{\textcolor{gray}{#1}}
\SetCommentSty{mycommfont}

\subsection{Sampling algorithm to minimize separator sets.}
\SetKwInput{KwInput}{Input}
\SetKwInput{KwOutput}{Output}
\SetKw{Break}{break}
\SetKwFor{RepTimes}{repeat}{times}{end}
\SetKwFor{For}{for (}{)}{}
\begin{algorithm}[H]
\SetNoFillComment
\small
\KwInput{A set of voxels $V$, 
the number of runs $r$,
the number of candidate clusters $s$, 
allowed missing voxels $a$, 
allowed missing voxels step $A$.}
\KwOutput{a clustering $Z$ of $V$}
\DontPrintSemicolon
\SetKwFunction{gcc}{gencandidatecluster}
\SetKwProg{Fn}{function}{}{} 
\Fn{\gcc{Z}}{   
 \For{$l := k$;\ \ $l \geq k-a$;\ \  $l := l - A$}
 {
	$C:=\{v\}$, where $v$ is free random (i.e., non-cluster or non-separator) voxel\;		
	\lWhile{there is a free $w$ adjacent to a voxel in $C$ and $|C|<l$}{$C:=C \cup \{w\}$}
 }
 \Return the candidate cluster $C$ that minimizes the separator size of $Z \cup \{C\}$\;
 }

\tcp{Main procedure starts here.}
\RepTimes{$r$}{ 
    $Z:=\emptyset$\;
    \While{there is a free voxel that is not adjacent to a cluster}{
	Infer $s$ candidate clusters by \gcc(Z)\;
	$Z:=Z \cup \{C^*\}$, where $C^*$ is the candidate that minimizes the separator size.	
    }
}
\Return the clustering $Z$ with minimal separator set.
\caption{Initial clustering generator: phase 1}\label{alg:heuristic1}
\end{algorithm}

\SetKwRepeat{DoRepeat}{repeat}{until}

\begin{algorithm}[H]
\SetNoFillComment
\small
\DontPrintSemicolon
\KwInput{A clustering $Z$ of a voxel set $V$ with a separator set $X$,
the number of big runs $E$, the number of searches in a fixed subgraph $e$,
the number of neighbouring clusters $M$, 
and  time limit threshold $T$}
\KwOutput{A clustering $Z$ of $V$ with a separator set smaller than $|X|$ (if found)}
\RepTimes{until time limit $T$ is not reached or $E$}{
Take $G \subset Q$ from a connected component $Q$ of the whole voxel graph $V$ 
such that there is a set $A$ of at most $M$ adjacent clusters such that 
$G$ contains every voxel whose neighbours are only in $X \cup \bigcup A$.
\DoRepeat{there is no improvement in the last $e$ steps}
	{
    	   Clear the voxels from $G$\;
    	   Run phase I with $r=1$ to fill the clustering and update $Z$ if the score is improved\;
    }
}
\caption{Local optimization: phase 2}\label{alg:heuristic2}
\end{algorithm}

The algorithm consists of two phases: inferring an initial clustering, and improving regions consisting of a small number of neighbouring clusters. In the first phase, the algorithm starts from an empty clustering. It generates $s$ candidate clusters, where $s$ is a small integer, usually between 1 and 10. Each candidate cluster is created starting from a randomly chosen available voxel by a sequence of insertions of adjacent voxels such that the induced size of its separator is kept small. Then, the best candidate cluster, i.e., the cluster with the separator's minimal size, is inserted into the current clustering. The procedure is repeated until there is no space to insert a new cluster. The second phase consists of repetitions of local improvements. The algorithm randomly takes a small number of neighbouring clusters, removes them from the current clustering, and applies a procedure similar to the first phase to find a better setting of clusters.
In testing, we used the following sets of parameters:
\begin{itemize}
    \item For the initial clustering in phase I: $r=1000$, $s=5$, $a=10$, $A=2$.
    \item For improving a given clustering in phase II: $T=120$ seconds, $s \in \{3 \dots 5\}$, 
       $a \in \{ k, \lfloor \frac{k}{2} \rfloor, \lfloor \frac{k}{3} \rfloor \}$ for $k<20$, and $a \in \{ \lfloor \frac{k}{10} \rfloor, \lfloor \frac{k}{20} \rfloor \}$ for larger $k$'s,
       $A \in \{3, 4\}$, $e=3$, $E=+\infty$, $M\in \{2,\dots,7\}$.
\end{itemize}
The algorithm is implemented in $C$ and allows 
fast inference of  clusterings with acceptable sizes of separator sets. 

\subsection{Simulated Annealing}

\begin{samepage}

Simulated Annealing (SA), \autoref{alg:sa}, is applied to the best effort clustering result found with the heuristic two-phase algorithm explained in \autoref{alg:heuristic1} and \autoref{alg:heuristic2}. The heuristic algorithm finds a good $k$-separator on $V$ upon which SA attempts to improve by $V+$ tiling. By moving voxels around in tiles or newly created tiles the algorithm stages a proposal with the corresponding target $t'$. This proposal can lead to an improved, no change or worse state. A proposal may be rejected if it leads to a worse state but not necessarily, this to allow exploring other minima in vicinity. The condition $U(0,1) < f$ with $f: 1/i^{((t'-t)/tp)}$ controls accepting bad proposals, with $i$ the iteration, $tp$ the tuning parameter and $U(0,1)$ the continuous uniform distribution.

\begin{figure}[h]
\includegraphics[width=5in, height=3in]{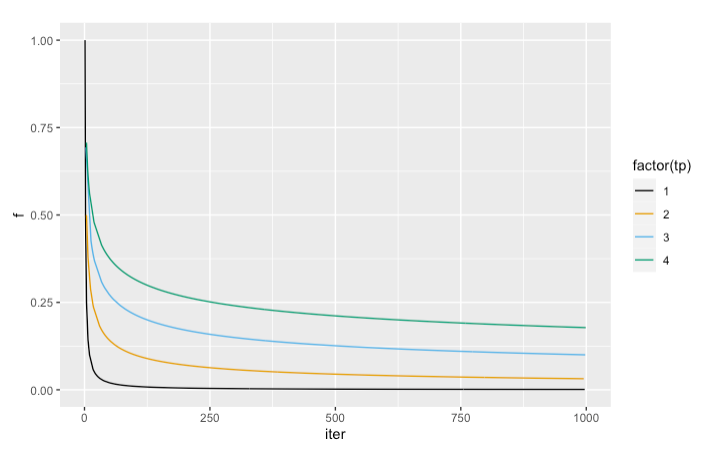}
\centering
\caption{function $f$ with tuning parameter 1 to 4. }
\label{fig:tp}
\end{figure}
Figure.\ref{fig:tp} shows $f$  when $t'>t$, i.e. a bad proposal. It shows that the probability of accepting a bad proposal diminishes with increasing number of iterations, and it does so more quickly with decreasing $tp$. The algorithm terminates when $tni$ (time no improvement) exceeds $iter$, maximum nr. of iterations allowed.
\end{samepage}

\begin{algorithm}[H]
\SetAlgoLined
  \texttt{// input parameters}\;
  $k$ \tcp*[l]{size k}
  tp \tcp*[l]{tuning parameter for acceptance}
  $T_1,\dots,T_n$ \tcp*[l]{construct tiles based on the best effort phase-one algorithm}
  $iter$  \tcp*[l]{max. nr. of iterations}
  \texttt{// init}\;
  $tni \gets 0$   \tcp*[l]{time no improvement}
  $i \gets 0$   \tcp*[l]{total nr. of iterations}
  $t \gets  t_k(T_1,\dots,T_n)$  \tcp*[l]{target }

\While{True} {
    $i++$\;
    $v \in V^+$ \tcp*[l]{random voxel $v$}
    $w \in V^+$ and $v-w \in \{-1,0,1\}^d$ \tcp*[l]{$w$ is a neighbour of $v$}
    \texttt{// $T_{\{v\}}$ : tile containing voxel $v$ }\;
    
    \eIf{( $T_{\{v\}} == T_{\{w\}} \;\;\&\;\; |T_{\{w\}}^- \cap V|>k )$ }{ 
        \texttt{// $v$ and $w$ are in the same tile $T$ and size interior $T$ is $>k$ }\;
        $t' \gets  t_k(T_1,\dots,[v],\dots, T_n)$ \tcp*[l]{ add tile $[v]$ and derive proposed target}
    }{
        $t' \gets  t_k(T_1,\dots,T_{\{w\}} \cup v,\dots, T_n)$ \tcp*[l]{ move $v$ to $T_{\{w\}}$ and derive proposed target}
    } 
    \texttt{//  accept/reject proposal? }\;
    \eIf{$U(0,1) < f(i,t,t',tp)$}{ 
        \texttt{// accept proposal}\;
        \eIf{$t' < t$}{
          \texttt{// $t'$ is the best so far}\;
          $t \gets t'$ \tcp*[l]{update target $t$}
          $tni \gets 0$ \tcp*[l]{reset time no improvement }
        }{
        tni++\;
        }
    }{
        \texttt{// reject proposal}\;
        tni++\; 
    }
    \If{$tni \geq iter$}{
    \textbf{break}\; 
    }
 } 
\caption{simulated annealing}
\label{alg:sa}
\end{algorithm}

\section{Further pruning illustration}

In Figure \ref{fig prune repeated} we illustrate the repeated pruning of the example voxel set $V$. 

\begin{figure}[!ht]
\centering
\begin{tikzpicture}[x=-1cm, scale=.5]

\path (14,12) node {$V$};
\begin{scope}[black]
\foreach \x in {4} \fill (\x,1) circle(.3cm);
\foreach \x in {3,4} \fill (\x,2) circle(.3cm);
\foreach \x in {3,4} \fill (\x,3) circle(.3cm);
\foreach \x in {3,4,5,8,10,11,12,13} \fill (\x,4) circle(.3cm);
\foreach \x in {2,3,...,6,8,9,...,12} \fill (\x,5) circle(.3cm);
\foreach \x in {2,3,...,10} \fill (\x,6) circle(.3cm);
\foreach \x in {2,3,...,10} \fill (\x,7) circle(.3cm);
\foreach \x in {1,2,...,13} \fill (\x,8) circle(.3cm);
\foreach \x in {3,4,...,13,14} \fill (\x,9) circle(.3cm);
\foreach \x in {3,4,...,13} \fill (\x,10) circle(.3cm);
\foreach \x in {3,4,5,12} \fill (\x,11) circle(.3cm);
\foreach \x in {3,4,5} \fill (\x,12) circle(.3cm);
\end{scope}

\begin{scope}[shift={(15 cm,0 cm)}]
\path (14,12) node {$V'$};
\begin{scope}[black]
\foreach \x in {3,4} \fill (\x,2) circle(.3cm);
\foreach \x in {3,4} \fill (\x,3) circle(.3cm);
\foreach \x in {3,4,5,10,11,12} \fill (\x,4) circle(.3cm);
\foreach \x in {2,3,...,6,8,9,...,12} \fill (\x,5) circle(.3cm);
\foreach \x in {2,3,...,10} \fill (\x,6) circle(.3cm);
\foreach \x in {2,3,...,10} \fill (\x,7) circle(.3cm);
\foreach \x in {2,...,13} \fill (\x,8) circle(.3cm);
\foreach \x in {3,4,...,13} \fill (\x,9) circle(.3cm);
\foreach \x in {3,4,...,13} \fill (\x,10) circle(.3cm);
\foreach \x in {3,4,5} \fill (\x,11) circle(.3cm);
\foreach \x in {3,4,5} \fill (\x,12) circle(.3cm);
\end{scope}

\begin{scope}[very thick]
\foreach \x in {4} \draw (\x,1) circle(.3cm);
\foreach \x in {8,13} \draw (\x,4) circle(.3cm);
\foreach \x in {1} \draw (\x,8) circle(.3cm);
\foreach \x in {14} \draw (\x,9) circle(.3cm);
\foreach \x in {12} \draw (\x,11) circle(.3cm);
\end{scope}
\end{scope}

\begin{scope}[shift={(0 cm,-14 cm)}]
\path (14,12) node {$V''$};
\begin{scope}[black]
\foreach \x in {3,4,5} \fill (\x,4) circle(.3cm);
\foreach \x in {2,3,...,6} \fill (\x,5) circle(.3cm);
\foreach \x in {2,3,...,10} \fill (\x,6) circle(.3cm);
\foreach \x in {2,3,...,10} \fill (\x,7) circle(.3cm);
\foreach \x in {2,...,13} \fill (\x,8) circle(.3cm);
\foreach \x in {3,4,...,13} \fill (\x,9) circle(.3cm);
\foreach \x in {3,4,...,13} \fill (\x,10) circle(.3cm);
\foreach \x in {3,4,5} \fill (\x,11) circle(.3cm);
\foreach \x in {3,4,5} \fill (\x,12) circle(.3cm);
\end{scope}

\begin{scope}[very thick]
\foreach \x in {3,4} \draw (\x,2) circle(.3cm);
\foreach \x in {3,4} \draw (\x,3) circle(.3cm);
\foreach \x in {4} \draw (\x,1) circle(.3cm);
\foreach \x in {8,10,11,12,13} \draw (\x,4) circle(.3cm);
\foreach \x in {8,9,...,12} \draw (\x,5) circle(.3cm);
\foreach \x in {1} \draw (\x,8) circle(.3cm);
\foreach \x in {14} \draw (\x,9) circle(.3cm);
\foreach \x in {12} \draw (\x,11) circle(.3cm);
\end{scope}
\end{scope}

\begin{scope}[shift={(15 cm,-14 cm)}]
\path (14,12) node {$V^{(3)}$};
\begin{scope}[black]
\foreach \x in {2,3,...,6} \fill (\x,5) circle(.3cm);
\foreach \x in {2,3,...,10} \fill (\x,6) circle(.3cm);
\foreach \x in {2,3,...,10} \fill (\x,7) circle(.3cm);
\foreach \x in {2,...,10} \fill (\x,8) circle(.3cm);
\foreach \x in {3,4,...,10} \fill (\x,9) circle(.3cm);
\foreach \x in {3,4,...,10} \fill (\x,10) circle(.3cm);
\end{scope}

\begin{scope}[very thick]
\foreach \x in {4} \draw (\x,1) circle(.3cm);
\foreach \x in {3,4} \draw (\x,2) circle(.3cm);
\foreach \x in {3,4} \draw (\x,3) circle(.3cm);
\foreach \x in {3,4,5,8,10,11,12,13} \draw (\x,4) circle(.3cm);
\foreach \x in {8,9,...,12} \draw (\x,5) circle(.3cm);
\foreach \x in {1,11,12,13} \draw (\x,8) circle(.3cm);
\foreach \x in {11,12,13,14} \draw (\x,9) circle(.3cm);
\foreach \x in {11,12,13} \draw (\x,10) circle(.3cm);
\foreach \x in {3,4,5,12} \draw (\x,11) circle(.3cm);
\foreach \x in {3,4,5} \draw (\x,12) circle(.3cm);
\end{scope}
\end{scope}

\begin{scope}[shift={(0 cm,-28 cm)}]
\path (14,12) node {$V^{(4)}$};
\begin{scope}[black]
\foreach \x in {3,...,10} \fill (\x,6) circle(.3cm);
\foreach \x in {3,...,10} \fill (\x,7) circle(.3cm);
\foreach \x in {3,...,10} \fill (\x,8) circle(.3cm);
\foreach \x in {3,4,...,10} \fill (\x,9) circle(.3cm);
\foreach \x in {3,4,...,10} \fill (\x,10) circle(.3cm);
\end{scope}

\begin{scope}[very thick]
\foreach \x in {4} \draw (\x,1) circle(.3cm);
\foreach \x in {3,4} \draw (\x,2) circle(.3cm);
\foreach \x in {3,4} \draw (\x,3) circle(.3cm);
\foreach \x in {3,4,5,8,10,11,12,13} \draw (\x,4) circle(.3cm);
\foreach \x in {2,...,6,8,9,...,12} \draw (\x,5) circle(.3cm);
\foreach \x in {2} \draw (\x,6) circle(.3cm);
\foreach \x in {2} \draw (\x,7) circle(.3cm);
\foreach \x in {1,2,11,12,13} \draw (\x,8) circle(.3cm);
\foreach \x in {11,12,13,14} \draw (\x,9) circle(.3cm);
\foreach \x in {11,12,13} \draw (\x,10) circle(.3cm);
\foreach \x in {3,4,5,12} \draw (\x,11) circle(.3cm);
\foreach \x in {3,4,5} \draw (\x,12) circle(.3cm);
\end{scope}
\end{scope}

\begin{scope}[shift={(15 cm,-28 cm)}]
\path (14,12) node {$V^{(5)}$};

\begin{scope}[very thick]
\foreach \x in {4} \draw (\x,1) circle(.3cm);
\foreach \x in {3,4} \draw (\x,2) circle(.3cm);
\foreach \x in {3,4} \draw (\x,3) circle(.3cm);
\foreach \x in {3,4,5,8,10,11,12,13} \draw (\x,4) circle(.3cm);
\foreach \x in {2,3,...,6,8,9,...,12} \draw (\x,5) circle(.3cm);
\foreach \x in {2,3,...,10} \draw (\x,6) circle(.3cm);
\foreach \x in {2,3,...,10} \draw (\x,7) circle(.3cm);
\foreach \x in {1,2,...,13} \draw (\x,8) circle(.3cm);
\foreach \x in {3,4,...,13,14} \draw (\x,9) circle(.3cm);
\foreach \x in {3,4,...,13} \draw (\x,10) circle(.3cm);
\foreach \x in {3,4,5,12} \draw (\x,11) circle(.3cm);
\foreach \x in {3,4,5} \draw (\x,12) circle(.3cm);
\end{scope}
\end{scope}

\end{tikzpicture}
\caption{Illustration of the repeated pruning of the voxel set $V$ from Figure \ref{fig tiling}.} \label{fig prune repeated}
\end{figure}
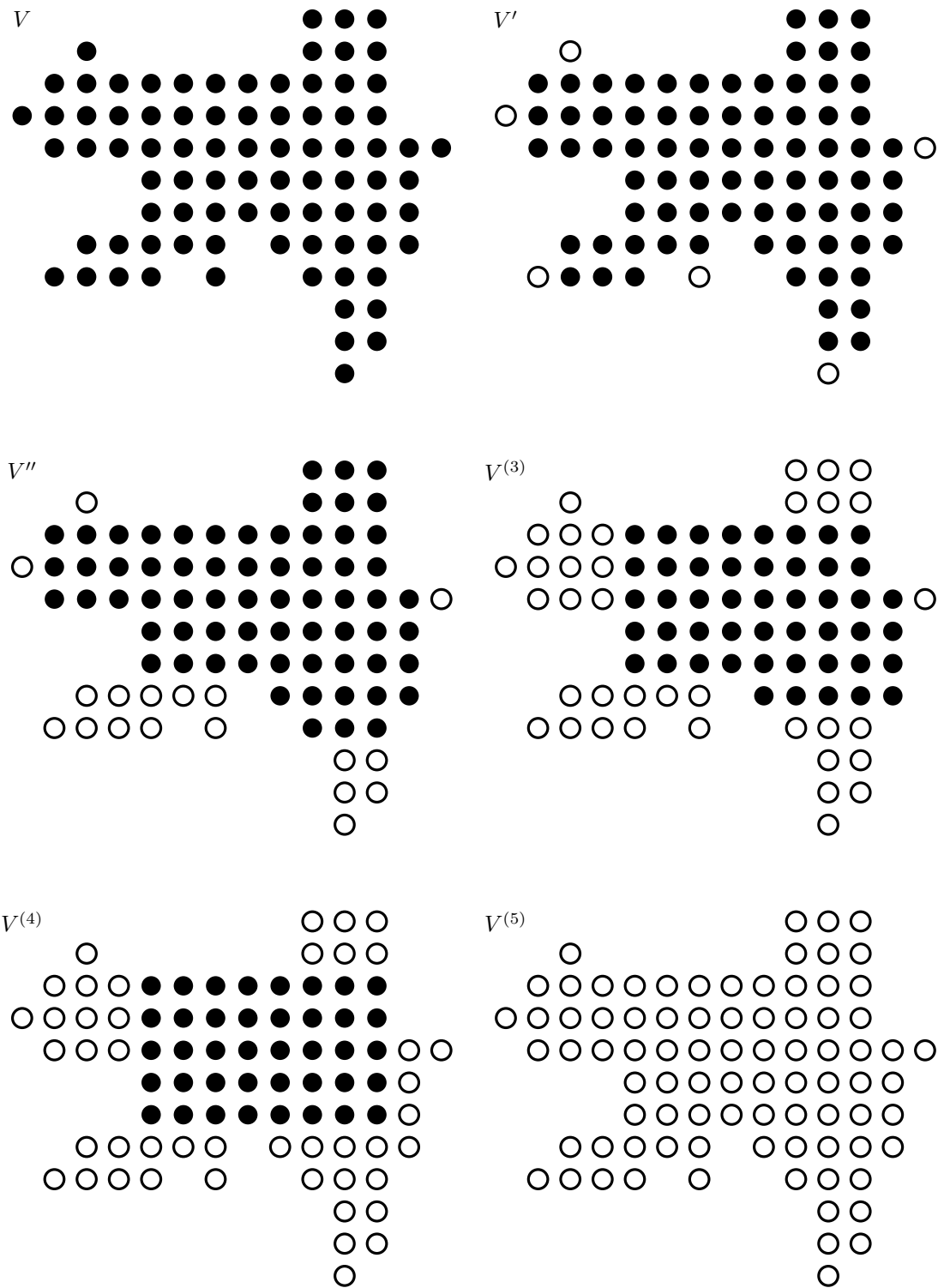

\section{Neurovault analysis}
For the Neurovault analysis we downloaded all 543 available collections (April 2019). From these collections we removed all empty collections (16), collections with no valid images (6), collections with no BOLD-fMRI images (127), collections with no statistics images (68), and collections with no group-level statistics images (114). This resulted in 218 valid collections containing 1909 statistics images.

To prevent extreme results we selected a representative sample of images with the following properties: brain size between 10k and 500k voxels, largest cluster size between 1 and 20k voxels, and estimated smoothness smaller than 13 voxels FWHM. This resulted in 1128 valid images for further analysis.

From these images we removed 310 images that were identical or did not contain any clusters with a size larger than the RFT-based cluster size at $\alpha=.05$. The final analyses were thus performed on 818 images.

For each image we performed the following analysis steps:

\begin{enumerate}
\item Check type of statistics image
\item If $t$-value image, check df, when available convert to $z$-scores, else leave as is.
\item Estimate smoothness on $z$-statistics image using 'smoothest'
\item Estimate contiguous clusters with $Z > 3.1$
\item Estimate $z$-threshold value associated with $k_M=14$
\item Estimate contiguous clusters with $z$-threshold when $k_M = 14$
\item Calculate cluster-extent $p$-values based on RFT for all clusters (both for $Z>3.1$ and $k_M=14$)
\item Estimate cluster True Discovery Proportion (cTDP) for each cluster (both for $Z>3.1$ and $k_M=14$)
\item Calculate voxelwise $p$-value based on RFT
\item Calculate the number of significant voxels based on voxelwise RFT threshold 
\end{enumerate}

\section{Pseudo code for permutations}

Here below we present the pseudo code for finding the $z$-score threshold with a given cluster extent threshold (see algorithm \ref{algZ}), or computing the cluster extent threshold with a given $z$-threshold (see algorithm \ref{algK}), using permutations. The problem of finding supra-threshold clusters for either threshold is equivalent to the standard incremental connectivity problem that can be solved efficiently using a disjoint-set data structure. For each permutation, we find sorted $z$-scores in $O(m\log m)$ time, and implementing the disjoint-set data structure using the optimized path compression and union by size takes linear time in the size of the output, i.e.,\ $O(v)$ for both algorithms. We suggest using at least 1000 permutations if the total number of permutations is too large.

We note that the permutation $z$-scores do not need to be pre-calculated; they may be calculated inside the for-loop if storage space is a consideration.\\

\SetKwInput{KwInput}{Input}
\SetKwInput{KwOutput}{Output}
\SetKwProg{Fn}{Function}{}{}
\SetKwFor{For}{for}{}{}

\SetKwFunction{FindZ}{FindZ}
\begin{algorithm}[!ht]
\SetNoFillComment
\small
\DontPrintSemicolon
\caption{Compute the permutation-based $z$-threshold, corresponding to a given cluster extent threshold, using a disjoint-set data structure.}
\label{algZ}
\KwInput{vectors of size $m$: $\mathbf{z}_{\pi_1}, \ldots, \mathbf{z}_{\pi_N}$ for $N$ permutations (the first permutation is equal to the identity);
a pre-specified cluster extent threshold $k_M$;
a list $\mathcal{N}$ of $m$ vectors, each storing the neighbours of a voxel;
a significance level $\alpha$.}
\KwOutput{a z-score threshold $Z_\alpha$.}
\BlankLine
\Fn{\FindZ{$\mathbf{z}_{\pi_1}, \ldots, \mathbf{z}_{\pi_N}, k_M, \mathcal{N}, \alpha$}}
{
    Initialize a $z$-score vector $\mathbf{Z}$ for $N$ permutations: $Z[j]=0$ for $j=1,\ldots,N$.\;
    \BlankLine
	\For{$j=1$ \KwTo $N$}
	{
		Sort $\mathbf{z}_{\pi_j}$ in descending order such that $z_{\pi_j}[1] \ge \cdots \ge z_{\pi_j}[m]$.\;
		\BlankLine
		Initialize disjoint sets $S_1=\{1\}, \ldots, S_m=\{m\}$.\;
		$v \leftarrow 1$\;
		\While{$|S_v| \le k_M$}
		{
		    $v \leftarrow v+1$\;
			\ForAll{$u \in \mathcal{N}[v]$ such that $u < v$}
			{
			    Merge $S_v$ and $S_u$ using union by size and path compression.\;
			}
		}
		$Z[j] \leftarrow z_{\pi_j}[v]$\;
	}
	\BlankLine
	Sort $\mathbf{Z}$ and find the $(1-\alpha)$-quantile $Z_\alpha = Z[ \lceil N(1-\alpha) \rceil ]$.\;
	\Return $Z_\alpha$
}
\end{algorithm}

\clearpage

\SetKwFunction{FindK}{FindK}
\begin{algorithm}[H]
\SetNoFillComment
\small
\DontPrintSemicolon
\caption{Compute the permutation-based cluster extent threshold for a given $z$-score threshold, using a disjoint-set data structure.}
\label{algK}
\KwInput{vectors of size $m$: $\mathbf{z}_{\pi_1}, \ldots, \mathbf{z}_{\pi_N}$ for $N$ permutations (the first permutation is equal to the identity);
a given $z$-threshold $z$;
a list $\mathcal{N}$ of $m$ vectors, each storing the neighbours of a voxel;
a significance level $\alpha$.}
\KwOutput{a cluster-extent  threshold $K_\alpha$.}
\BlankLine
\Fn{\FindK{$\mathbf{z}_{\pi_1}, \ldots, \mathbf{z}_{\pi_N}, z, \mathcal{N}, \alpha$}}
{
    Initialize a vector $\mathbf{K}$ for $N$ permutations: $K[j]=0$ for $j=1,\ldots,N$.\;
	\BlankLine
	\For{$j = 1$ \KwTo $N$}
	{
		Sort $\mathbf{z}_{\pi_j}$ in descending order such that $z_{\pi_j}[1] \ge \cdots \ge z_{\pi_j}[m]$.\;
		\BlankLine
		Initialize disjoint sets $S_1=\{1\}, \ldots, S_m=\{m\}$.\;
		$v \leftarrow 1$\;
		\While{$z_{\pi_j}[v] > z$}
		{
			\ForAll{$u \in \mathcal{N}[v]$ such that $u < v$}
			{
			 	Merge $S_v$ and $S_u$ using union by size and path compression.\;
			}
			$v \leftarrow v+1$\;
		}
		$K[j] \leftarrow \max \{|S_i| \colon i<v, i \in \mathbb{Z}^+\}$
	}
	\BlankLine
	Sort $\mathbf{K}$ and find the $(1-\alpha)$-quantile $K_\alpha = K[ \lceil N(1-\alpha) \rceil ]$.\; 
	\Return $K_\alpha$ 
}
\end{algorithm}

\section{Application: HCP Working Memory}

Based on the new method, TDP bounds were computed for supra-threshold clusters, formed by choosing either cluster-forming threshold or cluster-extent threshold and finding the other threshold based on the conventional Gaussian random field theory (RFT), and the overlapping anatomical regions. Table \ref{tbl: rft_k32} shows the results for a fixed $z$-threshold of $z=3.10$ and $k_M=32$, and Table \ref{tbl: rft_k14} shows the results for a given cluster-extent threshold $k_M=14$ and $z=3.61$. Consistent with what we observed for permutation inference, the results for RFT also suggest using the cluster-extent threshold $k_M=14$ instead of the standard $z$-threshold $z=3.10$ for better detection power. Similarly, decreasing $k_M$ leads to the increased $z$ and smaller clusters with higher TDP.

\begin{table}[!ht]
\caption{Results for supra-threshold clusters, defined by the cluster-forming $z$-threshold $Z > 3.10$ and minimal cluster extent threshold $k_M = 32$ based on RFT.}  \label{tbl: rft_k32}
\begin{tabular}{@{\extracolsep{\fill}} c rcc r rrcc cccc}
\toprule
\multicolumn{4}{c}{Cluster} & \multicolumn{5}{c}{Anatomical region} & \multicolumn{3}{c}{Position} & \\
\cmidrule(r){1-4} \cmidrule(lr){5-9} \cmidrule(l){10-12}
ID & size & TDP & LB & Region & size & overlap & TDP & LB & $x$ & $y$ & $z$ & $Z_\text{max}$ \\
\midrule
1 & 8870 & 0.479 & 0.384 & MFG & 18250 & 4049 & 0.106 & 0.087 & 44 & 72 & 60 & 8.87 \\
& & & & FP & 33571 & 2021 & 0.026 & 0.020 & & & & \\
& & & & IC & 6591 & 564 & 0.036 & 0.028 & & & & \\
2 & 8526 & 0.508 & 0.421 & sLOC & 27121 & 5142 & 0.089 & 0.071 & 19 & 42 & 61 & 9.51 \\
& & & & AG & 13689 & 4260 & 0.150 & 0.125 & & & & \\
& & & & pSMG & 14829 & 3804 & 0.125 & 0.104 & & & & \\
& & & & Precuneous & 18119 & 2491 & 0.065 & 0.053 & & & & \\
3 & 7956 & 0.444 & 0.323 & Cerebellum & 39724 & 6551 & 0.075 & 0.056 & 63 & 33 & 20 & 9.20 \\
4 & 6652 & 0.479 & 0.383 & MFG & 18250 & 4035 & 0.107 & 0.087 & 31 & 67 & 64 & 9.73 \\
& & & & FP & 33571 & 2587 & 0.035 & 0.027 & & & & \\
& & & & IC & 6591 & 589 & 0.037 & 0.028 & & & & \\
5 & 350 & 0.306 & 0.149 & pMTG & 11420 & 310 & 0.008 & 0.004 & 15 & 46 & 28 & 5.18 \\
& & & & tMTG & 9735 & 271 & 0.008 & 0.003 & & & & \\
6 & 100 & 0.270 & 0.110 & Cerebellum & 39724 & 100 & 0.001 & 0.000 & 49 & 35 & 10 & 6.56 \\
7 & 59 & 0.034 & 0.017 & Caudate & 4571 & 51 & 0.000 & 0.000 & 54 & 68 & 39 & 3.92 \\
8 & 58 & 0.069 & 0.017 & Cerebellum & 39724 & 58 & 0.000 & 0.000 & 42 & 36 & 10 & 4.85 \\
9 & 48 & 0.167 & 0.021 & Thalamus & 4602 & 34 & 0.000 & 0.000 & 43 & 53 & 43 & 4.55 \\
10 & 45 & 0.133 & 0.022 & Caudate & 4571 & 45 & 0.001 & 0.000 & 38 & 67 & 42 & 4.38 \\
11 & 35 & 0.086 & 0.029 & Cerebellum & 39724 & 35 & 0.000 & 0.000 & 44 & 41 & 25 & 4.20 \\
12 & 35 & 0.029 & 0.029 & Thalamus & 4602 & 35 & 0.000 & 0.000 & 42 & 52 & 35 & 5.02 \\
\cmidrule(lr){1-13}
Total & 32734 & 0.472 & 0.372 & MFG & 18250 & 8084 & 0.213 & 0.174 & & & & \\
& & & & Cerebellum & 39724 & 6744 & 0.076 & 0.056 & & & & \\
& & & & sLOC & 27121 & 5142 & 0.089 & 0.071 & & & & \\
& & & & FP & 33571 & 4608 & 0.061 & 0.047 & & & & \\
& & & & AG & 13689 & 4260 & 0.150 & 0.125 & & & & \\
& & & & pSMG & 14829 & 3804 & 0.125 & 0.104 & & & & \\
& & & & Precuneous & 18119 & 2491 & 0.065 & 0.053 & & & & \\
& & & & IC & 6591 & 1153 & 0.073 & 0.056 & & & & \\
& & & & pMTG & 11420 & 310 & 0.008 & 0.004 & & & & \\
& & & & tMTG & 9735 & 271 & 0.008 & 0.003 & & & & \\
& & & & Caudate & 4571 & 96 & 0.002 & 0.000 & & & & \\
& & & & Thalamus & 4602 & 69 & 0.001 & 0.000 & & & & \\
\bottomrule
\end{tabular}
\end{table}

\begin{table}[!ht]
\caption{Results for supra-threshold clusters, defined by the cluster-forming $z$-threshold of $Z > 3.61$, based on RFT, and minimal cluster extent threshold $k_M = 14$.}  \label{tbl: rft_k14}
\begin{tabular}{@{\extracolsep{\fill}} c rcc r rrcc cccc}
\toprule
\multicolumn{4}{c}{Cluster} & \multicolumn{5}{c}{Anatomical region} & \multicolumn{3}{c}{Position} & \\
\cmidrule(r){1-4} \cmidrule(lr){5-9} \cmidrule(l){10-12}
ID & size & TDP & LB & Region & size & overlap & TDP & LB & $x$ & $y$ & $z$ & $Z_\text{max}$ \\
\midrule
1 & 7415 & 0.607 & 0.534 & sLOC & 27121 & 4419 & 0.094 & 0.080 & 19 & 42 & 61 & 9.51 \\
& & & & AG & 13689 & 3846 & 0.166 & 0.147 & & & & \\
& & & & pSMG & 14829 & 3426 & 0.137 & 0.122 & & & & \\
& & & & Precuneous & 18119 & 2182 & 0.069 & 0.060 & & & & \\
2 & 7158 & 0.580 & 0.493 & MFG & 18250 & 3336 & 0.106 & 0.091 & 44 & 72 & 60 & 8.87 \\
& & & & SFG & 18946 & 2976 & 0.089 & 0.075 & & & & \\
& & & & poIFG & 8301 & 1410 & 0.092 & 0.077 & & & & \\
& & & & IC & 6591 & 500 & 0.041 & 0.034 & & & & \\
3 & 5655 & 0.550 & 0.440 & Cerebellum & 39724 & 5081 & 0.071 & 0.058 & 63 & 33 & 20 & 9.20 \\
4 & 5347 & 0.578 & 0.492 & MFG & 18250 & 3405 & 0.109 & 0.094 & 31 & 67 & 64 & 9.73 \\
& & & & FP & 33571 & 1960 & 0.032 & 0.027 & & & & \\
& & & & IC & 6591 & 526 & 0.042 & 0.036 & & & & \\
5 & 223 & 0.413 & 0.202 & OP & 15486 & 173 & 0.004 & 0.002 & 39 & 22 & 36 & 5.72 \\
& & & & ICC & 7134 & 121 & 0.007 & 0.004 & & & & \\
6 & 151 & 0.384 & 0.205 & pMTG & 11420 & 151 & 0.005 & 0.003 & 15 & 46 & 28 & 5.18 \\
7 & 69 & 0.377 & 0.188 & Cerebellum & 39724 & 69 & 0.001 & 0.000 & 49 & 35 & 10 & 6.56 \\
8 & 69 & 0.377 & 0.130 & FP & 33571 & 69 & 0.001 & 0.000 & 31 & 86 & 29 & 5.77 \\
9 & 61 & 0.344 & 0.164 & FP & 33571 & 61 & 0.001 & 0.000 & 57 & 88 & 29 & 5.16 \\
10 & 44 & 0.341 & 0.136 & OP & 15486 & 44 & 0.001 & 0.000 & 51 & 15 & 42 & 5.35 \\
11 & 27 & 0.222 & 0.037 & Thalamus & 4602 & 21 & 0.001 & 0.000 & 43 & 53 & 43 & 4.55 \\
12 & 23 & 0.087 & 0.043 & Cerebellum & 39724 & 23 & 0.000 & 0.000 & 42 & 36 & 10 & 4.85 \\
13 & 20 & 0.250 & 0.050 & Caudate & 4571 & 20 & 0.001 & 0.000 & 38 & 67 & 42 & 4.38 \\
14 & 19 & 0.053 & 0.053 & Cerebellum & 39724 & 19 & 0.000 & 0.000 & 42 & 32 & 28 & 4.19 \\
15 & 17 & 0.176 & 0.059 & tMTG & 9735 & 17 & 0.000 & 0.000 & 18 & 41 & 32 & 4.13 \\
16 & 16 & 0.125 & 0.063 & Thalamus & 4602 & 16 & 0.000 & 0.000 & 42 & 52 & 35 & 5.02 \\
\cmidrule(lr){1-13}
Total & 26314 & 0.574 & 0.484 & MFG & 18250 & 6741 & 0.215 & 0.185 & & & & \\
& & & & Cerebellum & 39724 & 5192 & 0.072 & 0.058 & & & & \\
& & & & sLOC & 27121 & 4419 & 0.094 & 0.080 & & & & \\
& & & & AG & 13689 & 3846 & 0.166 & 0.147 & & & & \\
& & & & pSMG & 14829 & 3426 & 0.137 & 0.122 & & & & \\
& & & & SFG & 18946 & 2976 & 0.089 & 0.075 & & & & \\
& & & & Precuneous & 18119 & 2182 & 0.069 & 0.060 & & & & \\
& & & & FP & 33571 & 2090 & 0.034 & 0.027 & & & & \\
& & & & poIFG & 8301 & 1410 & 0.092 & 0.077 & & & & \\
& & & & IC & 6591 & 1026 & 0.083 & 0.070 & & & & \\
& & & & OP & 15486 & 217 & 0.005 & 0.002 & & & & \\
& & & & pMTG & 11420 & 151 & 0.005 & 0.003 & & & & \\
& & & & ICC & 7134 & 121 & 0.007 & 0.004 & & & & \\
& & & & Thalamus & 4602 & 37 & 0.001 & 0.000 & & & & \\
& & & & Caudate & 4571 & 20 & 0.001 & 0.000 & & & & \\
& & & & tMTG & 9735 & 17 & 0.000 & 0.000 & & & & \\
\bottomrule
\end{tabular}
\end{table}


\begin{thebibliography}{}

\bibitem[Andreella et~al., 2020]{Andreella2020}
Andreella, A., Hemerik, J., Weeda, W., Finos, L., and Goeman, J. (2020).
\newblock Permutation-based true discovery proportions for fmri cluster
  analysis.
\newblock {\em arXiv preprint arXiv:2012.00368}.

\bibitem[Barch et~al., 2013]{Barch2013}
Barch, D.~M., Burgess, G.~C., Harms, M.~P., Petersen, S.~E., Schlaggar, B.~L.,
  Corbetta, M., Glasser, M.~F., Curtiss, S., Dixit, S., Feldt, C., Nolan, D.,
  Bryant, E., Hartley, T., Footer, O., Bjork, J.~M., Poldrack, R., Smith, S.,
  Johansen-Berg, H., Snyder, A.~Z., {Van Essen}, D.~C., and Consortium,
  W.-M.~H. (2013).
\newblock {Function in the human connectome: task-fMRI and individual
  differences in behavior}.
\newblock {\em NeuroImage}, 15:169--189.

\bibitem[Beckmann et~al., 2003]{Beckmann2003}
Beckmann, C.~F., Jenkinson, M., and Smith, S.~M. (2003).
\newblock General multilevel linear modeling for group analysis in fmri.
\newblock {\em {NeuroImage}}, 20(2):1052--1063.

\bibitem[Ben-Ameur et~al., 2015]{ben2015separator}
Ben-Ameur, W., Mohamed-Sidi, M.-A., and Neto, J. (2015).
\newblock The k-separator problem: polyhedra, complexity and approximation
  results.
\newblock {\em Journal of Combinatorial Optimization}, 29(1):276--307.

\bibitem[Blain et~al., 2022]{Blain2022}
Blain, A., Thirion, B., and Neuvial, P. (2022).
\newblock Notip: Non-parametric true discovery proportion control for brain
  imaging.
\newblock {\em NeuroImage}, page 119492.

\bibitem[Blanchard et~al., 2020]{Blanchard2020}
Blanchard, G., Neuvial, P., Roquain, E., et~al. (2020).
\newblock Post hoc confidence bounds on false positives using reference
  families.
\newblock {\em Annals of Statistics}, 48(3):1281--1303.

\bibitem[Bullmore et~al., 1999]{Bullmore1999}
Bullmore, E., Suckling, J., Overmeyer, S., Rabe-Hesketh, S., Taylor, E., and
  Brammer, M. (1999).
\newblock {Global, voxel, and cluster tests, by theory and permutation, for a
  difference between two groups of structural MR images of the brain}.
\newblock {\em IEEE Transactions on Medical Imaging}, 18(1):32--42.

\bibitem[Chumbley et~al., 2010]{Chumbley2010}
Chumbley, J., Worsley, K.~J., Flandin, G., and Friston, K.~J. (2010).
\newblock {Topological FDR for neuroimaging.}
\newblock {\em NeuroImage}, 49(4):3057--64.

\bibitem[Eklund et~al., 2016]{Eklund2016}
Eklund, A., Nichols, T.~E., and Knutsson, H. (2016).
\newblock Cluster failure: Why {fMRI} inferences for spatial extent have
  inflated false-positive rates.
\newblock {\em Proceedings of the national academy of sciences},
  113(28):7900--7905.

\bibitem[Forman et~al., 1995]{Forman1995}
Forman, S.~D., Cohen, J.~D., Fitzgerald, M., Eddy, W.~F., Mintun, M.~A., and
  Noll, D.~C. (1995).
\newblock {Improved Assessment of Significant Activation in Functional Magnetic
  Resonance Imaging (fMRI): Use of a Cluster-Size Threshold}.
\newblock {\em Magnetic Resonance in Medicine}, 33(5):636--647.

\bibitem[Friston et~al., 1991]{Friston1991}
Friston, K.~J., Frith, C.~D., Liddle, P.~F., and Frackowiak, R.~S. (1991).
\newblock {Comparing functional (PET) images: the assessment of significant
  change}.
\newblock {\em Journal of cerebral blood flow and metabolism}, 11(4):690--699.

\bibitem[Friston et~al., 1994]{Friston1994}
Friston, K.~J., Worsley, K.~J., Frackowiak, R.~S., Mazziotta, J.~C., and Evans,
  A.~C. (1994).
\newblock Assessing the significance of focal activations using their spatial
  extent.
\newblock {\em Human brain mapping}, 1(3):210--220.

\bibitem[Genovese and Wasserman, 2006]{Genovese2006}
Genovese, C.~R. and Wasserman, L. (2006).
\newblock Exceedance control of the false discovery proportion.
\newblock {\em Journal of the American Statistical Association},
  101(476):1408--1417.

\bibitem[Glasser et~al., 2013]{Glasser2013}
Glasser, M.~F., Sotiropoulos, S.~N., Wilson, J.~A., Coalson, T.~S., Fischl, B.,
  Andersson, J.~L., Xu, J., Jbabdi, S., Webster, M., Polimeni, J.~R., {Van
  Essen}, D.~C., and Jenkinson, M. (2013).
\newblock {The minimal preprocessing pipelines for the Human Connectome
  Project}.
\newblock {\em NeuroImage}, 80:105--124.

\bibitem[Goeman et~al., 2021]{Goeman2020}
Goeman, J.~J., Hemerik, J., and Solari, A. (2021).
\newblock Only closed testing procedures are admissible for controlling false
  discovery proportions.
\newblock {\em The Annals of Statistics}, 49(2):1218--1238.

\bibitem[Goeman et~al., 2019]{Goeman2019}
Goeman, J.~J., Meijer, R.~J., Krebs, T.~J., and Solari, A. (2019).
\newblock Simultaneous control of all false discovery proportions in
  large-scale multiple hypothesis testing.
\newblock {\em Biometrika}, 106(4):841--856.

\bibitem[Goeman and Solari, 2011]{Goeman2011}
Goeman, J.~J. and Solari, A. (2011).
\newblock Multiple testing for exploratory research.
\newblock {\em Statistical Science}, 26(4):584--597.

\bibitem[Gorgolewski et~al., 2015]{Gorgolewski2015}
Gorgolewski, K.~J., Varoquaux, G., Rivera, G., Schwarz, Y., Ghosh, S.~S.,
  Maumet, C., Sochat, V.~V., Nichols, T.~E., Poldrack, R.~A., Poline, J.-B.,
  Yarkoni, T., and Margulies, D.~S. (2015).
\newblock {NeuroVault.org: a web-based repository for collecting and sharing
  unthresholded statistical maps of the human brain}.
\newblock {\em Frontiers in Neuroinformatics}, 9:8.

\bibitem[Hayasaka and Nichols, 2003]{Hayasaka2003}
Hayasaka, S. and Nichols, T.~E. (2003).
\newblock {Validating cluster size inference: random field and permutation
  methods}.
\newblock {\em NeuroImage}, 20(4):2343--2356.

\bibitem[Jenkinson et~al., 2012]{Jenkinson2012}
Jenkinson, M., Beckmann, C.~F., Behrens, T.~E., Woolrich, M.~W., and Smith,
  S.~M. (2012).
\newblock {FSL}.
\newblock {\em NeuroImage}, 62(2):782--790.

\bibitem[Lindquist, 2008]{Lindquist2008}
Lindquist, M.~A. (2008).
\newblock {The Statistical Analysis of fMRI Data}.
\newblock {\em Statistical Science}, 23(4):439--464.

\bibitem[Marcus et~al., 1976]{Marcus1976}
Marcus, R., Eric, P., and Gabriel, K.~R. (1976).
\newblock On closed testing procedures with special reference to ordered
  analysis of variance.
\newblock {\em Biometrika}, 63(3):655--660.

\bibitem[Nichols, 2012]{Nichols2012}
Nichols, T.~E. (2012).
\newblock {Multiple testing corrections, nonparametric methods, and random
  field theory}.
\newblock {\em NeuroImage}, 62(2):811--815.

\bibitem[Ogawa et~al., 1992]{Ogawa1992}
Ogawa, S., Tank, D.~W., Menon, R., Ellermann, J.~M., Kim, S.~G., Merkle, H.,
  and Ugurbil, K. (1992).
\newblock {Intrinsic signal changes accompanying sensory stimulation:
  Functional brain mapping with magnetic resonance imaging}.
\newblock {\em Proceedings of the National Academy of Sciences of the United
  States of America}.

\bibitem[Poline et~al., 1997]{Poline1997}
Poline, J.~B., Worsley, K.~J., Evans, A.~C., and Friston, K.~J. (1997).
\newblock {Combining Spatial Extent and Peak Intensity to Test for Activations
  in Functional Imaging}.
\newblock {\em NeuroImage}, 5(2):83--96.

\bibitem[Rosenblatt et~al., 2018]{Rosenblatt2018}
Rosenblatt, J.~D., Finos, L., Weeda, W.~D., Solari, A., and Goeman, J.~J.
  (2018).
\newblock All-resolutions inference for brain imaging.
\newblock {\em Neuroimage}, 181:786--796.

\bibitem[{Van Essen} et~al., 2013]{VanEssen2013}
{Van Essen}, D.~C., Smith, S.~M., Barch, D.~M., Behrens, T. E.~J., Yacoub, E.,
  Ugurbil, K., and Consortium, W.-M.~H. (2013).
\newblock {The WU-Minn Human Connectome Project: An overview}.
\newblock {\em NeuroImage}, 80:62--79.

\bibitem[Vesely et~al., 2021]{Vesely2021}
Vesely, A., Finos, L., and Goeman, J.~J. (2021).
\newblock Permutation-based true discovery guarantee by sum tests.
\newblock {\em arXiv preprint arXiv:2102.11759}.

\bibitem[Woo et~al., 2014]{Woo2014}
Woo, C.-W., Krishnan, A., and Wager, T.~D. (2014).
\newblock {Cluster-extent based thresholding in fMRI analyses: Pitfalls and
  recommendations}.
\newblock {\em NeuroImage}, 91:412--419.

\bibitem[Woolrich et~al., 2001]{Woolrich2001}
Woolrich, M.~W., Ripley, B.~D., Brady, M., and Smith, S.~M. (2001).
\newblock Temporal autocorrelation in univariate linear modeling of {FMRI}
  data.
\newblock {\em NeuroImage}, 14(6):1370--1386.

\bibitem[Worsley et~al., 1992]{Worsley1992}
Worsley, K.~J., Evans, A.~C., Marrett, S., and Neelin, P. (1992).
\newblock {A Three-Dimensional Statistical Analysis for CBF Activation Studies
  in Human Brain}.
\newblock {\em Journal of Cerebral Blood Flow {\&} Metabolism}, 12(6):900--918.

\bibitem[Worsley et~al., 1996]{Worsley1996}
Worsley, K.~J., Marrett, S., Neelin, P., Vandal, A.~C., Friston, K.~J., and
  Evans, A.~C. (1996).
\newblock {A unified statistical approach for determining significant signals
  in images of cerebral activation}.
\newblock {\em Human Brain Mapping}, 4(1):58--73.

\bibitem[Yannakakis, 1981]{yannakakis1981node}
Yannakakis, M. (1981).
\newblock Node-deletion problems on bipartite graphs.
\newblock {\em SIAM Journal on Computing}, 10(2):310--327.

\end{thebibliography}
\end{document}